\documentclass[journal,twoside,web]{ieeecolor}
\let\labelindent\relax

\usepackage{generic}
\usepackage{cite}
\usepackage{textcomp}
\usepackage{amsmath,amsthm,amssymb,amsfonts,color,bm}
\usepackage{graphicx}
\usepackage{textcomp}
\usepackage{xcolor}
\usepackage{gensymb}
\usepackage{booktabs}
\usepackage{multirow}
\usepackage{tabularx}
\usepackage{subcaption}
\usepackage{enumitem}
\usepackage[fixamsmath]{mathtools}
\usepackage{physics}
\usepackage{adjustbox}

\let\norm\relax
\DeclarePairedDelimiter{\norm}{\lVert}{\rVert}
\newcommand{\cu}[1]{\mathbf{#1}}

\newcommand{\tash}[2]{\frac{\partial #1}{\partial #2}}
\newcommand{\tashh}[3]{\frac{\partial^2 #1}{\partial {#2} \partial {#3}}}

\newcommand{\diff}{\mathop{}\!\mathrm{d}}

\newcommand{\expec}{\mathbb{E}} 
\newcommand{\cov}{\operatorname{Cov}}     

\newcommand{\R}{\mathbb{R}} 

\newcommand{\diag}{\operatorname{diag}}     
\newcommand{\mineig}{\lambda_{\mathrm{min}}}  

\newcommand*{\trans}{{\mkern-1.5mu\mathsf{T}}}

\newtheorem{theorem}{Theorem}
\newtheorem{myalgorithm}[theorem]{Algorithm}
\newtheorem{assumption}[theorem]{Assumption}
\newtheorem{proposition}[theorem]{Proposition}

\newtheorem{lemma}[theorem]{Lemma}
\newtheorem{remark}[theorem]{Remark}
\newtheorem{example}[theorem]{Example}

\def\BibTeX{{\rm B\kern-.05em{\sc i\kern-.025em b}\kern-.08em
    T\kern-.1667em\lower.7ex\hbox{E}\kern-.125emX}}
\begin{document}
\title{Taylor Moment Expansion for Continuous-Discrete Gaussian Filtering and Smoothing}

\author{Zheng Zhao, Toni Karvonen, Roland Hostettler,~\IEEEmembership{Member,~IEEE}, and Simo S\"{a}rkk\"{a},~\IEEEmembership{Senior Member,~IEEE}%
	\thanks{Manuscript received Jan, 2020. 
		Corresponding author: Zheng Zhao (email: zheng.zhao@aalto.fi).}
	\thanks{Zheng Zhao, Toni Karvonen, and Simo S\"{a}rkk\"{a} were with Department of Electrical Engineering and Automation, Aalto University, Finland. }
	\thanks{Roland Hostettler was with Department of Engineering Sciences, Uppsala University.}}

\maketitle

\begin{abstract}
The paper is concerned with non-linear Gaussian filtering and smoothing in continuous-discrete state-space models, where the dynamic model is formulated as an It\^{o} stochastic differential equation (SDE), and the measurements are obtained at discrete time instants. We propose novel Taylor moment expansion (TME) Gaussian filter and smoother which approximate the moments of the SDE with a temporal Taylor expansion. Differently from classical linearisation or It\^{o}--Taylor approaches, the Taylor expansion is formed for the moment functions directly and in time variable, not by using a Taylor expansion on the non-linear functions in the model. We analyse the theoretical properties, including the positive definiteness of the covariance estimate and stability of the TME Gaussian filter and smoother. By numerical experiments, we demonstrate that the proposed TME Gaussian filter and smoother significantly outperform the state-of-the-art methods in terms of estimation accuracy and numerical stability. 
\end{abstract}

\begin{IEEEkeywords}
continuous-discrete state-space model, Gaussian filtering and smoothing, Kalman filtering and smoothing, stochastic differential equation, Taylor moment expansion
\end{IEEEkeywords}

\section{Introduction}
\IEEEPARstart{B}{ayesian} filtering and smoothing refer to algorithms that estimate the states of a stochastic time-varying system from noisy observations \cite{sarkka_2013}. They have drawn significant attention in applications such as target tracking \cite{barshalomBook2002}, signal processing \cite{zzmlsp2018, zhao2018kalman, gao2019regularized, beataBook2016}, and inverse problems \cite{JariInverse2005}. Particularly useful instances of Bayesian filters and smoothers are the Gaussian assumed density filter and smoother (see, e.g., \cite[Ch.~6]{sarkka_2013} and \cite{Kazufumi2000}), where the filtering and smoothing posterior distributions are assumed to be Gaussian. The so-called Kalman filter~\cite{Kalman1960} and Rauch--Tung--Striebel smoother~\cite{Rauch1965}  are special cases of the Gaussian filter and smoother when the state-space models are linear and Gaussian. For non-linear models, the extended Kalman filter and smoother \cite{jazwinski2007stochastic}, unscented Kalman filter and smoother \cite{eric2000UKF, julierUKF2004, SARKKA2010225}, cubature Kalman filter and smoother \cite{simonCKF2009, ARASARATNAM20112245}, and Fourier--Hermite Kalman filter \cite{juha2012, juhaSmoother2012} are implementations of the Gaussian filter and smoother.

Continuous-discrete models arise when the system dynamics are in continuous time and its observations in discrete time. The continuous dynamics are usually formulated with stochastic differential equations (SDEs) which defines a Markov process~\cite{oksendal2013stochastic}. The challenge of performing Bayesian filtering and smoothing in continuous-discrete models is the intractability of transition density of the SDEs, which requires to solve the Fokker--Planck--Kolmogorov (FPK) partial differential equation (PDE)~\cite{sarkkaarnobook2019}. In the Gaussian filtering and smoothing framework, the problem is addressed by forming a  Gaussian approximation to the SDE~\cite{Kushner-early1967, Kazufumi2000, sarkkaarnobook2019, filip-2019-iterated-smoothing}. 

The It\^{o}--Taylor expansion~\cite{mallick2010cd, david2015, ckdIndian2010} and the so-called ordinary differential equation (ODE) approach~\cite{Kazufumi2000, SARKKA2013500, SARKKA20121221} are classical ways to construct the Gaussian approximation. In the It\^{o}--Taylor expansion scheme, the idea is to approximate the SDE solution numerically in discrete time steps by using the It\^{o}--Taylor expansion of the SDE~\cite{peter-koleden-book}. Examples of this kind of methods are the Euler--Maruyama and Milstein's method~\cite{peter-koleden-book}. However, it is difficult to use higher-order It\^{o}--Taylor approximations due to the intractability of iterated It\^{o} integrals~\cite{sarkkaarnobook2019, peter-koleden-book}, and thus the methods only work with small enough time intervals. In the ODE approach, the idea is to approximate the mean and covariance ODEs for the SDE. However, the mean and covariance ODEs are not analytically tractable because they require to calculate expectations of the stochastic process. This is usually circumvented by using a Gaussian approximation. An example of the ODE approach is the continuous-discrete extended Kalman filter~\cite{jazwinski2007stochastic, sarkkaarnobook2019}, where one needs to linearise the drift function of SDEs to make the ODEs tractable.

In this paper, we propose an alternative approach to form Gaussian approximations by employing the Taylor moment expansion (TME) method \cite{SANCHO1970384, kessler1997taylorexoansion, sarkkaarnobook2019} which is a Taylor expansion of moments in temporal direction. The contributions of this paper are as follows. (1) We develop a novel Taylor moment expansion based Gaussian filter and smoother for continuous-discrete state-space models. (2) In addition, we present theoretical analysis on the positive definiteness of TME covariance estimate and the stability of TME Gaussian filter and smoother. (3) Finally, we show by numerical experiments that the proposed TME Gaussian filter and smoother outperform the state-of-the-art methods in terms of both estimation accuracy and numerical stability.

The paper is structured as follows. In Section \ref{sec:cd-general}, we formulate the continuous-discrete Gaussian filtering and smoothing problems and highlight the challenges. In Section \ref{sec:taylor-expansion}, we introduce and derive TME and then form the TME Gaussian filter and smoother, as well as the analysis on the positive definiteness of TME covariance estimate and the stability of filter and smoother. The numerical experiments are presented in Section \ref{sec:numerical-experiments}, followed by conclusion in Section~\ref{sec:conclusion}. 

\section{Continuous-Discrete Gaussian Filtering and Smoothing}
\label{sec:cd-general}
Consider a continuous-discrete state-space model
\begin{subequations}\label{equ:problem_formulation}
	\begin{align}
	\diff \cu{x}_t &= \cu{f}(\cu{x}_t, t) \diff t + \cu{L}(\cu{x}_t, t) \diff\cu{W}_t, \label{equ:problem_formulation-dyn}\\
	\cu{y}_k &= \cu{h}(\cu{x}_{k}) + \cu{v}_k,
	\label{equ:problem_formulation-meas}
	\end{align}
\end{subequations}
where $\cu{x}_t \in \R^D$ is a $D$-dimensional It\^{o} process, $\cu{y}_k\in\R^Z$ is the measurement at time $t_k$, and $\cu{W}_t$ denotes an $S$-dimensional Wiener process with diffusion matrix $\cu{Q}$. We also assume the non-linear drift and dispersion functions $\cu{f}\triangleq\cu{f}(\cu{x}_t, t)$ and $\cu{L}(\cu{x}_t, t)$ are sufficiently regular so that \eqref{equ:problem_formulation-dyn} has a weakly unique solution \cite{oksendal2013stochastic, rogers_williams_2000}. As we are mostly concerned with the continuous-time part \eqref{equ:problem_formulation-dyn}, for simplicity, we model the measurement $\cu{y}_k$ in~\eqref{equ:problem_formulation-meas} with a non-linear function $\cu{h}(\cu{x}_{k})$ and a Gaussian noise $\cu{v}_k\sim\mathcal{N}(\cu{0}, \cu{V}_k)$. Furthermore, we denote $\bm{\Gamma}\triangleq\bm{\Gamma}(\cu{x}_t, t) \triangleq \cu{L}(\cu{x}_t, t)\, \cu{Q}\, \cu{L}^\trans(\cu{x}_t, t)$, and $\Gamma_{ij}$ denotes the $i$-th row and $j$-th column entry of $\bm{\Gamma}(\cu{x}_t, t)$. 

The aim is to form Gaussian approximations to the filtering and smoothing densities for any $t_k$, $k=1,\ldots, T$ as follows:
\begin{align}
p(\cu{x}_{k}\mid\cu{y}_{1:k}) &\approx \mathcal{N}(\cu{x}_{k}\mid \cu{m}_k, \cu{P}_k),
\label{equ:filter-post}\\
p(\cu{x}_{k}\mid\cu{y}_{1:T}) &\approx \mathcal{N}(\cu{x}_{k}\mid \cu{m}^s_k, \cu{P}^s_k).
\label{equ:smoother-post}
\end{align}
Above, we have used the notation $\cu{x}_k\triangleq\cu{x}_{t_k}$ at time $t_k$, and $\cu{y}_{1:k} = \left\lbrace \cu{y}_1,\, \cu{y}_2,\, \ldots,\, \cu{y}_k\right\rbrace $. Additionally, $\cu{m}_k$, $\cu{P}_k$, and $\cu{m}^s_k, \cu{P}^s_k$ are the means and covariances of \eqref{equ:filter-post} and \eqref{equ:smoother-post}, respectively. 

In order to obtain the exact posteriors on the right hand sides of \eqref{equ:filter-post} and \eqref{equ:smoother-post}, it would be necessary to compute the transition densities $p(\cu{x}_{k}\mid\cu{x}_{k-1})$ for the continuous model \eqref{equ:problem_formulation-dyn} (see, e.g., \cite{sarkka_2013,sarkkaarnobook2019}). It turns out that the transition density is only analytically tractable in limited cases, such as for linear SDEs. This is because in general, solving the FPK PDE for non-linear SDEs is not possible in closed form. Therefore, approximations are needed. Although it is possible to estimate the transition density by solving the FPK PDE numerically \cite{fpkShalom2000, sarkkaarnobook2019}, the computational effort is extensive. In the Gaussian filtering and smoothing framework, we are interested in constructing a Gaussian approximation to the transition density:
\begin{equation}
p(\cu{x}_{k}\mid\cu{x}_{k-1}) \approx \mathcal{N}(\cu{x}_k\mid \expec\left[\cu{x}_{k}\mid\cu{x}_{k-1} \right], \cov\left[\cu{x}_{k}\mid\cu{x}_{k-1} \right]),
\label{equ:transition-to-be-approx}
\end{equation}
which is also the approach that we employ here.

\subsection{Approximating Transition Density with It\^{o}--Taylor Series}
\label{sec:moment-ito-taylor}
One classical way to approximate the transition density~\eqref{equ:transition-to-be-approx} is the It\^{o}--Taylor expansion \cite{peter-koleden-book} which can be used to form a discretised solution to the SDE by expanding It\^{o} integrals iteratively using It\^{o}'s lemma. Euler--Maruyama is the simplest instance of this kind of methods, and the solution estimate $\hat{\cu{x}}_k$ is
\begin{equation}
\hat{\cu{x}}_{k} = \cu{x}_{k-1} + \cu{f}(\cu{x}_{k-1}, t_{k-1})\,\Delta t + \cu{L}(\cu{x}_{k-1}, t_{k-1}) \, \Delta \cu{W}_k, 
\end{equation}
where $\Delta t = t_k - t_{k-1}$ is the discretization interval and $\Delta \cu{W}_t \sim \mathcal{N}(\cu{0}, \cu{Q}\,\Delta t)$. From this, we can calculate the mean and covariance in \eqref{equ:transition-to-be-approx} as
\begin{equation}
\begin{split}
\expec\left[ \cu{x}_{k}\mid\cu{x}_{k-1}\right]  &\approx \cu{x}_{k-1} + \cu{f}(\cu{x}_{k-1}, t_{k-1})\,\Delta t, \\
\cov\left[\cu{x}_{k}\mid\cu{x}_{k-1}\right] &\approx \bm{\Gamma}(\cu{x}_{k-1}, t_{k-1})\,\Delta t, 
\end{split}
\label{equ:Euler--Maruyama}
\end{equation}
However, the Euler--Maruyama scheme only works well when the time interval $\Delta t$ is small enough. Other commonly used choices are, for example, the Milstein's method and the strong order 1.5 It\^{o}--Taylor (It\^{o}-1.5) method~\cite{peter-koleden-book, sarkkaarnobook2019}. However, because of the difficulty of expanding the iterative It\^{o} integrals, it is not easy to construct higher order It\^{o}--Taylor expansion \cite{peter-koleden-book} and hence this approach is inherently low order in $\Delta t$. Even when using the It\^{o}-1.5, we need the dispersion term to be constant. It\^{o}--Taylor expansion based continuous-discrete filters and smoothers can be found, for example, in~\cite{ckdIndian2010, SARKKA20121221, mallick2010cd, david2015, marelende2005}.

\subsection{Approximating Transition Density with ODEs}
\label{sec:moment-odes}
Another widely used approach is to approximate ODEs for the first two moments of the It\^{o} process \cite{jazwinski2007stochastic}. The mean and covariance of It\^{o} process \eqref{equ:problem_formulation-dyn} for any $t\in \left(t_{k-1}, t_k \right]$ are characterised by
\begin{equation}
\begin{split}
\frac{\diff \cu{m}_t}{\diff t} &= \expec\left[ \cu{f}(\cu{x}_t, t)\right], \\
\frac{\diff \cu{P}_t}{\diff t} &= \expec\left[ \cu{f}(\cu{x}_t, t)\,(\cu{x}_t-\cu{m}_t)^\trans\right]\\
&\quad+ \expec\left[ (\cu{x}_t-\cu{m}_t)\,\cu{f}^\trans(\cu{x}_t, t)\right] + \expec\left[ \bm{\Gamma}(\cu{x}_t,t)\right],
\label{equ:dmdt_mPP}
\end{split}
\end{equation}
where $\cu{m}_t = \expec\left[ \cu{x}_t\right]$ and $\cu{P}_t = \expec\left[ \left(\cu{x}_t-\cu{m}_t\right)\,(\cu{x}_t-\cu{m}_t)^\trans\right]$ \cite{sarkkaarnobook2019}. The initial values of $\cu{m}_t$ and $\cu{P}_t$ are given at time $t_{k-1}$. Notice that when using this scheme in Gaussian filtering and smoothing, it is not necessary to directly approximate the transition density \eqref{equ:transition-to-be-approx}. By solving the ODEs, they directly give the prediction step of filtering when the initial conditions are given by the previous filtering posterior. Unfortunately, the ODEs are only tractable for linear SDEs along with certain other isolated special cases, because the expectations in \eqref{equ:dmdt_mPP} are taken with respect to $p(\cu{x}_t, t)$, which requires to solve the FPK equation beforehand \cite{jazwinski2007stochastic}. 

To disentangle the intractability problem of these ODEs, one practical solution is to linearise $\cu{f}(\cu{x}_t, t)$ and $\cu{L}(\cu{x}_t, t)$ around $\cu{m}_t$, which results in the elimination of the expectation on the right-hand side of ODEs. This leads to the continuous-discrete extended Kalman filter (CD-EKF)~\cite{jazwinski2007stochastic}. However, this linearisation approach does not perform well when the models are significantly non-linear, albeit higher-order Taylor expansion can be used \cite{peter-maybeck-book}.

Another solution is to assume that the densities are Gaussian, in which case the expectations in the ODEs can be calculated with Gaussian quadrature. For non-linear integrands, numerical methods, such as Gauss--Hermite can be used \cite{Kazufumi2000}. With the aid of linearisation or Gaussian assumption strategies, solving the ODEs is straightforward. One can leverage numerical solvers, such as linear multistep methods or Runge--Kutta (RK) methods~\cite{griffithODE2010, butcher-numerical-ode}. Through this ODE approach, one can directly estimate the mean and covariance of SDE from the initial condition, which gives the prediction step of the continuous-discrete filter. Related methods are, for example, continuous-discrete unscented/cubature/Gauss--Hermite Kalman filters~\cite{SARKKA2010225, SARKKA2013500, SARKKA20121221, simo2007ukf}. For the ODE methods for smoothing, we refer to \cite{sarkkaarnobook2019} for details. 

For simplicity, we use \textit{Linear-ODE} and \textit{Gauss-ODE} to refer to the methods solving~\eqref{equ:dmdt_mPP} using linearisation and Gaussian assumptions, respectively.

\section{Taylor Moments Expansion Gaussian Filtering and Smoothing}
\label{sec:taylor-expansion}

As a Gaussian distribution is entirely characterised by its mean and covariance, a reasonable approach to Gaussian filtering and smoothing is to use moment matching to form a Gaussian approximation to the transition density. The previously presented It\^{o}--Taylor and ODE methods are useful tools for this purpose. However, because the moments of an It\^{o} process are functions of time, it is also possible to form a (deterministic) Taylor expansion with respect to time. Although this \textit{Taylor moment expansion} (TME) approach was originally introduced for likelihood-based parametric estimation of SDEs \cite{sarkkaarnobook2019, kessler1997taylorexoansion, yacine2002, florenz1986, florenz1989}, we here propose to use it in Gaussian filtering and smoothing. 

In this section, we first derive the TME approximation, which gives the approximate temporal evolution of the mean and covariance of the underlying process characterised by the SDE~\eqref{equ:problem_formulation-dyn}. Thereafter, we approximate the transition density as a Gaussian. The proposed TME Gaussian filtering and smoothing methods for continuous-discrete state-space models are then formulated based on this approximation. 

\subsection{Taylor Expansion Moment Approximation}
Let $\phi(\cu{x}_t)$ be an arbitrary twice-differentiable scalar function of the process $\cu{x}_t$. By It\^{o}'s lemma and taking the expectation on both sides yields
\begin{equation}
\begin{split}
\diff \expec\left[ \phi(\cu{x}_t)\right] &= \expec\left[\nabla \phi(\cu{x}_t)\,\cu{f}(\cu{x}_t, t) \right] \diff t \\
&\quad+ \frac{1}{2}\expec\left[ \tr\left( \nabla\nabla^\trans \phi(\cu{x}_t)\,\bm{\Gamma}(\cu{x}_t, t)\right) \right] \diff t, 
\end{split}
\label{equ:ito-moment}
\end{equation}
where $\nabla$ and $\nabla\nabla^\trans$ give the Jacobian and Hessian of $\phi(\cu{x}_t)$, respectively. With a proper choice of  $\phi$, this will lead to the moment ODEs as shown in \eqref{equ:dmdt_mPP} \cite{sarkkaarnobook2019}. The aim now is to form a Taylor expansion of the function $\expec\left[ \phi(\cu{x}_t) \right]$. We notice that the right-hand side of \eqref{equ:ito-moment} can be reformulated with the (generalized) infinitesimal generator
\begin{equation} 
\begin{split}
\mathcal{A} g &= \tash{g}{t} + \nabla g\,\cu{f}(\cu{x}_t, t) + \frac{1}{2}\tr(\nabla\nabla^\trans g\,\bm{\Gamma}(\cu{x}_t, t)),\\
&=\tash{g}{t} + \sum_i\tash{g}{x_i}f_i(\cu{x}_t, t) + \frac{1}{2}\sum_{i, j}\frac{\partial^2g}{\partial x_i\partial x_j}\Gamma_{ij},
\label{equ:generator}
\end{split}
\end{equation}
for any regular smooth function $g$, where $x_i$ is the $i$-th component of $\cu{x}_t$ (see, e.g., \cite[Ch.~7]{oksendal2013stochastic} or \cite[Ch.~9]{sarkkaarnobook2019}). Thus \eqref{equ:ito-moment} becomes
\begin{equation}
\begin{split}
\frac{\diff \expec\left[ \phi(\cu{x}_t)\right]}{\diff t}  = \expec \left[\mathcal{A}\phi(\cu{x}_t) \right],
\end{split}
\label{equ:taylor-exp-ode}
\end{equation}
which requires that $\phi\in C^2$ is twice differentiable. We also denote by $\mathcal{A}^r$ the $r$-th iteration of the generator. By taking derivatives of \eqref{equ:taylor-exp-ode} multiple times, we have \cite{sarkkaarnobook2019}
\begin{equation}
\begin{split}
\frac{\diff \expec\left[ \phi(\cu{x}_t)\right]}{\diff t} &= \expec\left[ \mathcal{A}\phi(\cu{x}_t)\right], \\
\frac{\diff^2 \expec\left[ \phi(\cu{x}_t)\right]}{\diff t^2} &= \expec\left[ \mathcal{A}^2\phi(\cu{x}_t)\right], \\
&\quad\vdots \\
\frac{\diff^M \expec\left[ \phi(\cu{x}_t)\right]}{\diff t^M} &= \expec\left[ \mathcal{A}^M\phi(\cu{x}_t)\right],
\end{split}
\label{equ:deriv-of-moment}
\end{equation}
which requires that $\phi \in C^{2M}$. Notice that \eqref{equ:deriv-of-moment} above also requires sufficient smoothness of $\cu{f}(\cu{x}, t)$ and $\bm{\Gamma}(\cu{x}_t, t)$. We can now form an $M$-th order Taylor expansion of the function $\expec\left[ \phi(\cu{x}_{k})\right]$ at time $t_k$ and centred at time $t_{k-1}$ as follows:
\begin{equation}
\begin{split}
\expec\left[ \phi(\cu{x}_{k})\right] &\approx \sum^M_{r=0}\frac{1}{r!}\frac{\diff^r \expec\left[ \phi(\cu{x}_{k-1})\right]}{\diff t^r}\,\Delta t^r \\
&= \sum^M_{r=0}\frac{1}{r!}\expec\left[ \mathcal{A}^r\phi(\cu{x}_{k-1})\right] \,\Delta t^r,
\end{split}
\label{equ:taylor-exp-1}
\end{equation}  
where $\Delta t = t_k - t_{k-1}$. Conditioning \eqref{equ:taylor-exp-1} on $\cu{x}_{k-1}$ gives
\begin{equation}
\expec\left[ \phi(\cu{x}_{k})\mid\cu{x}_{k-1}\right ] \approx \sum^M_{r=0}\frac{1}{r!}\mathcal{A}^r\phi(\cu{x}_{k-1})\,\Delta t^r.
\label{equ:taylor-exp-final}
\end{equation}

In Gaussian filtering and smoothing, we are only interested in the function $\phi$ having certain polynomial forms. For the mean and covariance, we introduce two sets of functions: $\left\lbrace \phi_i = x_i\colon i = 1, \ldots, D \right\rbrace $ and $\left\lbrace \phi_{ij} = x_i\,x_j\colon i,j = 1, \ldots, D \right\rbrace $, where $x_i$ is the $i$-th component of $\cu{x}_k$. Then, we have the mean $\expec\left[\cu{x}_{k}\mid\cu{x}_{k-1} \right] = \begin{bmatrix}
\expec\left[\phi_1\mid \cu{x}_{k-1} \right], &\ldots, &\expec\left[\phi_D\mid \cu{x}_{k-1} \right]
\end{bmatrix}^\trans$ and the covariance $\cov\left[x_i\,x_j\mid\cu{x}_{k-1}\right] = \expec\left[\phi_{ij}\mid\cu{x}_{k-1}\right] - \expec\left[\phi_i \mid\cu{x}_{k-1}\right]\,\expec\left[\phi_j \mid\cu{x}_{k-1}\right]$. Using the approximation \eqref{equ:taylor-exp-final}, we can now form the Taylor moment expansion (TME) estimator for the transition density as shown in Algorithm \ref{def:taylor-moment-estimator}.

\begin{myalgorithm}[Taylor Moment Expansion (TME) of Transition Density]\label{def:taylor-moment-estimator}
	Let $p(\cu{x}_{k}\mid\cu{x}_{k-1})$ be the transition density of SDE \eqref{equ:problem_formulation}, where $\Delta t$ is the time interval. The $M$-th order Taylor expansion based estimates of the mean $\cu{a}_M $, the second moment $\cu{B}_M $, and the covariance $ \bm{\Sigma}_M $ of the transition density are given by
	\begin{equation}
	\begin{split}
	\cu{a}_{M}&\triangleq \cu{a}_M(\cu{x}_{k-1}, \Delta t)\\
	&= \sum^M_{r=0}\frac{1}{r!}\,\mathcal{A}^r  \cu{x}_{k-1}  \, \Delta t^r \\ &\approx \expec \left[ \cu{x}_{k}\mid\cu{x}_{k-1}\right], \\
	\cu{B}_{M}&\triangleq \cu{B}_M(\cu{x}_{k-1}, \Delta t)\\
	&= \sum^M_{r=0}\frac{1}{r!}\,\mathcal{A}^r \left( \cu{x}_{k-1}\,\cu{x}^\trans_{k-1}\right)  \Delta t^r \\ &\approx \expec \left[ \cu{x}_{k}\,\cu{x}^\trans_{k}\mid\cu{x}_{k-1}\right],\\
	\bm{\Sigma}_{M}&\triangleq \bm{\Sigma}_{M}(\cu{x}_{k-1}, \Delta t) \\
	&= \cu{B}_{M} - \cu{a}_{M}\,\cu{a}_{M}^\trans \\
	&\approx \expec\left[(\cu{x}_{k} - \cu{a}_{M})\,(\cu{x}_{k} - \cu{a}_{M})^\trans \mid \cu{x}_{k-1}\right] \\
	&\approx \cov\left[\cu{x}_{k} \mid \cu{x}_{k-1} \right],
	\end{split}
	\label{equ:taylor-exp-mean-cov}
	\end{equation}
	respectively. Here the application of the generator $\mathcal{A}$ on vector or matrix input means that we apply the operator elementwise.
\end{myalgorithm}

\begin{remark}
	\label{remark:truncate}
	Note that if an $M$-th order TME approximation is used, the covariance estimator $\bm{\Sigma}_{M}$ in Algorithm~\ref{def:taylor-moment-estimator} is a polynomial of degree $2M$ in $\Delta t$, which comes from the product $\cu{a}_M\,\cu{a}_M^\trans$. To keep the order of $\Delta t$ consistent in mean and covariance, $\bm{\Sigma}_M$ needs to be truncated to degree $M$.
\end{remark}

In addition, it is also important to recover the remainder $R(\cu{x}_{k-1}, \Delta t)$ of TME, such that \eqref{equ:taylor-exp-1} becomes
\begin{equation}
\begin{split}
\expec\left[ \phi(\cu{x}_{k})\right] &= \sum^M_{r=0}\frac{1}{r!}\expec\left[ \mathcal{A}^r\phi(\cu{x}_{k-1})\right] \,\Delta t^r + R(\cu{x}_{k-1},\Delta t).\nonumber
\end{split}
\end{equation}
By Taylor's theorem, the remainder admits the form
\begin{equation}
\begin{split}
R(\cu{x}_{k-1}, \Delta t)&=\\
&\hspace{-1.5cm}\int^{t_k}_{t_{k-1}}\int^{u_M}_{t_{k-1}}\cdots\int^{u_1}_{t_{k-1}} \expec\left[\mathcal{A}^{M+1}\phi(\cu{x}_s)\right] \diff s\diff u_1\cdots\diff u_{M}, 
\label{equ:expansion-remainder}
\end{split}
\end{equation}
provided that $\phi$ and the SDE coefficients are sufficiently smooth \cite{kessler1997taylorexoansion, florenz1989}. The convergence properties of TME can be then analysed in the same way as the Taylor expansion. 

The difference between the TME and the aforementioned ODE and It\^{o}--Taylor schemes is mainly how the Gaussian approximation to the continuous model is done. The It\^{o}--Taylor approach first discretises the SDE solution which gives the discretised approximation $\hat{\cu{x}}_k$, and then obtains the moment $\expec\left[\phi(\hat{\cu{x}}_{k})\right]$ using the approximation. In ODE approach we need to postulate certain hypothesis, such as Gaussian assumptions or linearisation before solving the ODEs~\eqref{equ:dmdt_mPP}. In contrast, TME gives the moment $\hat{\expec}\left[ \phi(\cu{x}_{k})\right]$ estimate directly, without forming the discretised approximation $\hat{\cu{x}}_k$ or approximation to the ODEs. 

\subsection{TME Gaussian Filtering and Smoothing}
Using Algorithm~\ref{def:taylor-moment-estimator}, we now formulate the proposed TME Gaussian filter and smoother by utilising an $M$-th order TME estimate of the transition density $p(\cu{x}_{k}\mid \cu{x}_{k-1}) \approx \mathcal{N}(\cu{x}_k\mid \cu{a}_M, \bm{\Sigma}_M)$. Notice that although we are using simplified notations $\cu{a}_M$, $\cu{B}_M$, and $\bm{\Sigma}_M$, those terms are functions of $\cu{x}_{k-1}$ and $\Delta t$.

Let us assume the filtering posterior from previous time step $t_{k-1}$ is $p(\cu{x}_{k-1}\mid\cu{y}_{1:k-1})= \mathcal{N}(\cu{x}_{k-1}\mid\cu{m}_{k-1}, \cu{P}_{k-1})$. We first perform prediction with respect to the continuous model~\eqref{equ:problem_formulation-dyn}, and thus the prediction density $p(\cu{x}_{k}\mid \cu{y}_{1:k-1})= \mathcal{N}(\cu{x}_{k}\mid \cu{m}_k^-, \cu{P}_k^-)$ is characterised by
\begin{align}
& \expec\left[\cu{x}_{k}\mid\cu{y}_{1:k-1} \right]\nonumber \\
&= \int \cu{x}_{k} \int p(\cu{x}_{k}\mid \cu{x}_{k-1})\,  p(\cu{x}_{k-1}\mid\cu{y}_{1:k-1})\,\diff \cu{x}_{k-1}\diff\cu{x}_{k}\nonumber \\
&\approx \int \cu{a}_{M}\, \mathcal{N}(\cu{x}_{k-1}\mid\cu{m}_{k-1}, \cu{P}_{k-1}) \diff \cu{x}_{k-1}\nonumber \\
&= \expec\left[\cu{a}_{M} \right] = \cu{m}^-_k, 
\label{equ:taylor-pred-m}
\end{align}
and
\begin{align}
&  \cov\left[\cu{x}_{k}\mid\cu{y}_{1:k-1} \right]\nonumber \\
&=\int \cu{x}_{k}\,\cu{x}^\trans_{k} \,  \int p(\cu{x}_{k}\mid \cu{x}_{k-1})\,  p(\cu{x}_{k-1}\mid\cu{y}_{1:k-1})\,\diff \cu{x}_{k-1}\diff\cu{x}_{k}\nonumber \\
&\quad- \cu{m}^-_k\,(\cu{m}^-_k)^\trans\nonumber \\
&\approx \int \cu{B}_M\, \mathcal{N}(\cu{x}_{k-1}\mid\cu{m}_{k-1}, \cu{P}_{k-1}) \diff \cu{x}_{k-1} - \cu{m}^-_k\,(\cu{m}^-_k)^\trans\nonumber\\
&= \expec\left[ \bm{\Sigma}_M + \cu{a}_M\,\cu{a}_M^\trans\right] - \cu{m}^-_k\,(\cu{m}^-_k)^\trans = \cu{P}^-_k .
\label{equ:taylor-pred}
\end{align}
Notice that we are calculating $\cu{P}^-_k$ using $\cu{P}^-_k = \expec\left[ \bm{\Sigma}_M + \cu{a}_M\,\cu{a}_M^\trans\right] - \cu{m}^-_k\,(\cu{m}^-_k)^\trans$ instead of directly $\cu{P}^-_k = \expec\left[\cu{B}_M \right]- \cu{m}^-_k\,(\cu{m}^-_k)^\trans$. Recall from Remark \ref{remark:truncate} that they are not equal, as we truncate $\bm{\Sigma}_M$ to keep the power of $\Delta t$ consistent. By using $\cu{P}^-_k = \expec\left[\cu{B}_M \right]- \cu{m}^-_k\,(\cu{m}^-_k)^\trans$, it is difficult to perform such truncation. For the positive definiteness analysis of $\cu{P}^-_k$, it is also easier to analyse $\bm{\Sigma}_M$ directly rather than $\expec\left[\cu{B}_M \right]- \cu{m}^-_k\,(\cu{m}^-_k)^\trans$. The prediction covariance $\cu{P}^-_k$ is not always positive definite, because $\bm{\Sigma}_M$ is a Taylor approximation to the true covariance, which is an issue that will be discussed in Section~\ref{sec:numerical_stability}. 

Finally, the Gaussian filtering posterior \eqref{equ:filter-post} is obtained by conditioning the joint distribution
\begin{equation}
p(\cu{x}_{k}, \cu{y}_k\mid\cu{y}_{1:k-1}) = 
\mathcal{N}\left(\begin{bmatrix}
\cu{x}_{k} \\ \cu{y}_k
\end{bmatrix}\,\,\middle|\,\,\begin{bmatrix}
\cu{m}^-_{k} \\ \bm{\mu}_k
\end{bmatrix}, \begin{bmatrix}
\cu{P}^-_{k} & \cu{C}_k\\
\cu{C}^\trans_k & \cu{S}_k
\end{bmatrix}\right), \nonumber
\end{equation}
on $\cu{y}_k$, which is the so-called update step in Gaussian filtering.  

As to the smoothing posterior \eqref{equ:smoother-post}, it is obtained from 
\begin{equation}
\begin{split}
&p(\cu{x}_{k}, \cu{x}_{k+1}\mid\cu{y}_{1:T}) \\
&= p(\cu{x}_{k}\mid\cu{x}_{k+1}, \cu{y}_{1:k})\, p(\cu{x}_{k+1}\mid\cu{y}_{1:T})
\end{split}
\end{equation}
by marginalising $\cu{x}_{k+1}$, where  $p(\cu{x}_{k+1}\mid\cu{y}_{1:T})$ is the smoothing posterior from subsequent time step $t_{k+1}$. In addition,  $p(\cu{x}_{k}\mid\cu{x}_{k+1}, \cu{y}_{1:k})$ is obtained by conditioning the joint distribution 
\begin{equation}
\begin{split}
&p(\cu{x}_{k}, \cu{x}_{k+1}\mid\cu{y}_{1:k}) \\
&=\mathcal{N}\left(\begin{bmatrix}
\cu{x}_{k} \\ \cu{x}_{k+1}
\end{bmatrix}\,\,\middle|\,\, \begin{bmatrix}
\cu{m}_k \\ \cu{m}^-_{k+1}
\end{bmatrix}, 
\begin{bmatrix}
\cu{P}_{k} & \cu{D}_{k+1}\\
\cu{D}^\trans_{k+1} & \cu{P}^-_{k+1}
\end{bmatrix}\right)
\end{split}
\label{equ:smooth-gp}
\end{equation}
on $\cu{x}_{k+1}$. Using $M$-th order TME, the cross-covariance $\cu{D}_{k+1}$ in \eqref{equ:smooth-gp} is calculated as
\begin{align}
\cu{D}_{k+1} &= \cov\left[\cu{x}_{k},\cu{x}^\trans_{k+1}\mid \cu{y}_{1:k} \right] \nonumber\\
&= \iint\cu{x}_{k}\,\cu{x}^\trans_{k+1}\,p(\cu{x}_{k+1}\mid \cu{x}_{k})\\
&\hspace{0.7cm}\times p(\cu{x}_{k}\mid\cu{y}_{1:k})\diff\cu{x}_{k}\,\diff\cu{x}_{k+1}- \cu{m}_k\,(\cu{m}^-_{k+1})^\trans\nonumber\\
&\approx \int \cu{x}_{k}\,\cu{a}_M^\trans \, \mathcal{N}(\cu{x}_k\mid\cu{m}_{k}, \cu{P}_{k}) \diff \cu{x}_{k} - \cu{m}_k\,(\cu{m}^-_{k+1})^\trans\nonumber\\
&= \expec\left[\cu{x}_{k}\,\cu{a}_M^\trans \right]- \cu{m}_k\,(\cu{m}^-_{k+1})^\trans.\nonumber
\end{align}
In practice, it is more efficient to compute $\cu{D}_{k+1}$ in the prediction step of filtering, so that the sigma-point evaluations are shared if using numerical integration. 

The calculation of $\bm{\mu}_k$, $\cu{C}_k$, and $\cu{S}_k$, and the rest of the smoothing procedures are the same as in the general Gaussian filter and smoother \cite{sarkka_2013, sarkkaarnobook2019}. The complete procedures of TME Gaussian filter and smoother are formulated in Algorithms~\ref{alg:tme-filter} and \ref{alg:tme-smoother}, respectively.

\begin{myalgorithm}[TME Gaussian Filter] Starting from initial filtering condition $\cu{x}_0 \sim \mathcal{N}(\cu{x}_0\mid \cu{m}_0, \cu{P}_0)$, the equations of TME Gaussian filter for $k=1, 2, \ldots, T$ are as follows:\label{alg:tme-filter}
	\begin{itemize}[leftmargin=*]
		\item \text{Prediction}:
		\begin{align}
		\cu{m}_k^- &= \int  \cu{a}_{M} \, \mathcal{N}(\cu{x}_{k-1}\mid\cu{m}_{k-1}, \cu{P}_{k-1})\diff \cu{x}_{k-1},\nonumber\\
		\cu{P}^-_{k} &= \int  \left( \bm{\Sigma}_M + \cu{a}_M\,\cu{a}_M^\trans\right) \, \mathcal{N}(\cu{x}_{k-1}\mid\cu{m}_{k-1}, \cu{P}_{k-1})\diff \cu{x}_{k-1} \nonumber\\
		&\quad- \cu{m}_k^-\,(\cu{m}_k^-)^\trans.\label{equ:cd-taylor-pred}
		\end{align}
		\item \text{Update}:
		\begin{equation}
		\begin{split}
		\bm{\mu}_k &= \int \cu{h}(\cu{x}_{k})\, \mathcal{N}(\cu{x}_k\mid \cu{m}_k^-, \cu{P}_k^-)\diff\cu{x}_k, \\
		\cu{S}_k &= \int (\cu{h}(\cu{x}_{k}) - \bm{\mu}_k)\,(\cu{h}(\cu{x}_{k}) - \bm{\mu}_k)^\trans\\
		&\hspace{0.6cm} \times \mathcal{N}(\cu{x}_k\mid \cu{m}_k^-, \cu{P}_k^-)\diff\cu{x}_k, \\
		\cu{C}_k &= \int (\cu{x}_{k}-\cu{m}^-_k)\,(\cu{h}(\cu{x}_{k}) - \bm{\mu}_k)^\trans \\
		&\hspace{0.6cm}\times \mathcal{N}(\cu{x}_k\mid \cu{m}_k^-, \cu{P}_k^-)\diff\cu{x}_k, \\
		\cu{K}_k &= \cu{C}_k \, \cu{S}_k^{-1}, \\
		\cu{m}_k &= \cu{m}_k^- + \cu{K}_k \, (\cu{y}_k - \bm{\mu}_k), \\
		\cu{P}_k &= \cu{P}_k^- - \cu{K}_k \, \cu{S}_k \, \cu{K}_k^\trans.
		\end{split}
		\end{equation}
	\end{itemize}
\end{myalgorithm}
\begin{myalgorithm}[TME Gaussian Smoother] Starting from end smoothing condition $\cu{x}_T \sim \mathcal{N}(\cu{x}_T\mid \cu{m}^s_T, \cu{P}^s_T)$, the equations of TME Gaussian smoother for $k=T-1, T-2, \ldots, 1$ are as follows:\label{alg:tme-smoother}
	\begin{align}
	\cu{D}_{k+1} &= \int \cu{x}_{k}\,\cu{a}_M^\trans \, \mathcal{N}(\cu{x}_k\mid\cu{m}_{k}, \cu{P}_{k}) \diff \cu{x}_{k}- \cu{m}_k\,(\cu{m}^-_{k+1})^\trans, \nonumber\\
	\cu{G}_k &= \cu{D}_{k+1}\,\left(\cu{P}^-_{k+1}\right)^{-1}, \label{equ:cd-cross-cov}\\
	\cu{m}^s_k &= \cu{m}_k + \cu{G}_k\, \left(\cu{m}^s_{k+1}-\cu{m}^-_{k+1} \right), \nonumber\\
	\cu{P}^s_k &= \cu{P}_k +  \cu{G}_k\, \left(\cu{P}^s_{k+1}-\cu{P}^-_{k+1} \right)\,\cu{G}_k^\trans.\nonumber
	\end{align}
\end{myalgorithm}

The calculation of Gaussian integrals in Algorithms~\ref{alg:tme-filter} and \ref{alg:tme-smoother} is intractable for many non-linear integrands (i.e., $\cu{a}_M$, $\cu{B}_M$, and $\bm{\Sigma}_M$). Herein, we consider numerically approximating them by using quadrature and sigma-point methods,  such as Gauss--Hermite \cite{Kazufumi2000}, unscented transform \cite{julierUKF2004}, and spherical cubature method \cite{simonCKF2009}. It is also worth mentioning that as the sigma-point approximation is an operation of linearly weighted summation, the positive definiteness of $\cu{P}_k^-$ is inherited from $\bm{\Sigma}_M$ provided that the quadrature weights are positive. This is true for Gauss--Hermite quadrature, spherical cubature, and unscented transformation with suitable parameter selection. In this article we assume that such positive-weight integration rule is used. Thus in the following analysis, we are only concerned with the positive definiteness of $\bm{\Sigma}_M$.

\section{Theoretical Analysis on Taylor Moment Expansion Gaussian Filter and Smoother}
\label{sec:theoretical-analysis}
In this section, we analyse the positive definiteness of TME estimates and the stability of the TME Gaussian filter and smoother. Notice that in Lemma~\ref{lemma:general-covariance} and Section~\ref{sec:stability}, we consider an SDE of the form
\begin{equation}
\diff\cu{x}_t = \cu{f}(\cu{x}_t)\diff t + \cu{L}\diff\cu{W}_t,
\label{equ:sde-for-analysis}
\end{equation}
to simplify the analysis, where the drift function $\cu{f}(\cu{x}_t)$ is time homogeneous, and the dispersion function $\cu{L}$ is constant. As a consequence, $\bm{\Gamma}$ is positive semi-definite (p.s.d.) and the time partial differentiation operation $\partial (\cdot)/\partial t$ in the generator \eqref{equ:generator} disappears. In other theorems and propositions, the more general SDE~\eqref{equ:problem_formulation-dyn} is used.

\subsection{Positive Definiteness of Taylor Moments Expansion}
\label{sec:numerical_stability}
In the Gaussian filtering and smoothing context, it is essential for the covariance estimate to stay positive definite (p.d.). Unfortunately, this is not always true when using TME, as we truncate the full Taylor expansion \cite{iacus-sde-simu}. For example, by Algorithm~\ref{def:taylor-moment-estimator}, the second order TME covariance estimate of a one-dimensional SDE is
\begin{equation}
\Sigma_{2} = \Gamma\,\Delta t + f'\,\Gamma\,\Delta t^2, 
\label{equ:pd-example}
\end{equation}
where $f'$ is the derivative of the drift function. It is apparent that $\Sigma_{2}$ can be negative because $f'$ is not always positive. The estimate $\Sigma_{2}$ is positive definite if and only if the inequality $f'>-\frac{1}{\Delta t}$ holds. It implies that $f'$ has to be non-negative if one needs $\Sigma_{2}$ to be positive for all $\Delta t>0$.

In fact, we show that the positive definiteness of the TME covariance estimate is jointly determined by the model itself (SDE coefficients), the time interval $\Delta t$, and the expansion order $M$. 

\begin{theorem}\label{theorem:psd}
	The $M$-th order TME covariance estimate $\bm{\Sigma}_M$ is positive definite for $\Delta t$ on an interval $U\subseteq \R^+$, if 
	\begin{equation}
	\begin{split}
	P_M(\Delta t) &=\sum^M_{r=1}w_r \, \Delta t^r>0, 
	\end{split}
	\label{equ:lemma1-poly1}
	\end{equation}
	for all $\Delta t\in U$, where $w_r = \frac{1}{r!}\mineig\left( \bm{\Phi}_{\cu{x}_t, r}\right)$, and $\mineig(\cdot)$ denotes the minimum eigenvalue of a square matrix. The coefficients are
	\begin{equation}
	\bm{\Phi}_{\cu{x}_t, r} = \mathcal{A}^r \left( \cu{x}_t\,\cu{x}_t^\trans\right)  - \sum^r_{s=0}\binom{r}{s}\,\mathcal{A}^s\cu{x}_t\,(\mathcal{A}^{r-s}\cu{x}_t)^\trans, 
	\label{equ:psd-first}
	\end{equation}
	where $\binom{r}{s}$ denotes the binomial coefficient. 
\end{theorem}
\begin{proof}
	From Algorithm~\ref{def:taylor-moment-estimator}, we know that the $(u, v)$-th entry of $\bm{\Sigma}_M$ is
	\begin{equation}
	\begin{split}
	\left[ \Sigma_M\right]_{uv} &= \sum^M_{r=0}\frac{1}{r!}\mathcal{A}^r(x_ux_v)\Delta t^r \\
	&-\left(\sum^M_{r=0}\frac{1}{r!}\mathcal{A}^rx_u\Delta t^r\right)\left(\sum^M_{r=0}\frac{1}{r!}\mathcal{A}^rx_v\Delta t^r\right),
	\label{equ:temp-poly}
	\end{split}
	\end{equation}
	where $x_u$ and $x_v$ are the $(u, v)$-th component of $\cu{x}_t$. Using the Cauchy product of finite series, we get 
	\begin{equation}
	\begin{split}
	\left[ \Sigma_M\right]_{uv} &= \sum^M_{r=0}\left[ \frac{1}{r!}\mathcal{A}^r(x_ux_v) -\left(\sum^r_{s=0}\frac{\mathcal{A}^sx_u\,\mathcal{A}^{r-s}x_v}{s!(r-s)!}\right)\right] \Delta t^r \\
	&= \sum^M_{r=0}\frac{1}{r!}\left[\mathcal{A}^r(x_ux_v) - \sum^r_{s=0}\binom{r}{s}\,\mathcal{A}^sx_u\,\mathcal{A}^{r-s}x_v \right],
	\end{split}
	\label{equ:temp-poly2}
	\end{equation}
	where we truncate the polynomial~\eqref{equ:temp-poly} of $\Delta t$ up to degree $M$ (see Remark \ref{remark:truncate}). Equation~\eqref{equ:temp-poly2} can be rearranged into matrix form
	\begin{equation}
	\bm{\Sigma}_M = \sum^M_{r=1}\frac{1}{r!}\bm{\Phi}_{\cu{x}_t, r} \Delta t^r.
	\label{equ:lemma1-sum}
	\end{equation}
	From Algorithm~\ref{def:taylor-moment-estimator}, we know that $\bm{\Phi}_{\cu{x}_t, r}$ is symmetric for all $r=1, \ldots, M$, thus its eigenvalues are all real. Using Weyl's inequality \cite{horn-johnson-1991}, 
	\begin{equation}
	\mineig\left(\bm{\Sigma}_M \right) \geq \sum^M_{r=1}\frac{1}{r!}\mineig\left(\bm{\Phi}_{\cu{x}_t, r} \right) \Delta t^r. 
	\label{equ:lemma1-poly2}
	\end{equation}
	Thus if the polynomial on the right hand of \eqref{equ:lemma1-poly2} is strictly positive for all $\Delta t \in U$, then $\bm{\Sigma}_M$ is positive definite. 
\end{proof}

Equation~\eqref{equ:lemma1-sum} reveals that the covariance estimate $\bm{\Sigma}_M$ is a polynomial function of $\Delta t$ with coefficients formed by $\left\lbrace \bm{\Phi}_{\cu{x}_t, r}\colon r=1,\ldots, M\right\rbrace $. However, to explicitly find the minimum eigenvalue of $\bm{\Sigma}_M$ might be difficult. One way would be to force all of the coefficients $\bm{\Phi}_{\cu{x}_t, r}$ to be p.d. so that $\bm{\Sigma}_M$ is p.d., but it would restrict the model class significantly. The idea behind Theorem~\ref{theorem:psd} is to construct a lower bound for the minimum eigenvalue of $\bm{\Sigma}_M$ in terms of the eigenvalues of the coefficients matrices $\bm{\Phi}_{\cu{x}_t, r}$, without restricting them to be p.d. 

Directly showing the conditions of a polynomial function, such as \eqref{equ:lemma1-poly1}, to be positive might be challenging. A more useful way in the analysis is to find if it has no real roots on an interval. In this case, tools like Budan's theorem or Sturm's theorem may be useful to count the exact number of roots of a polynomial \cite{algebra-geometry-base}. 

In the following Proposition~\ref{propo:1-2}, we give applications of using Theorem~\ref{theorem:psd} for the positive definiteness of the first two order of TME expansion. 
\begin{proposition}
	\label{propo:1-2}
	Let $\bm{\Sigma}_1$ and $\bm{\Sigma}_2$ be the TME estimates of the covariance with expansion order $1$ and $2$, respectively. Then
	\begin{enumerate}
		\item $\bm{\Sigma}_2$ is p.d. for $\Delta t>0$ if $(\bm{\Gamma}\,\nabla\cu{f})^\trans  + \bm{\Gamma}\,\nabla\cu{f}$ and $\bm{\Gamma}$ are p.s.d., and one of $\bm{\Gamma}$ and $(\bm{\Gamma}\,\nabla\cu{f})^\trans  + \bm{\Gamma}\,\nabla\cu{f}$ is p.d.
		\item $\bm{\Sigma}_3$ is p.d. for $\Delta t>0$ if $\bm{\Phi}_{\cu{x}_t, 3}$ is p.s.d. and $\mineig(\bm{\Phi}_{\cu{x}_t, 2})>\frac{-2\sqrt{6}}{3}\sqrt{\mineig(\bm{\Phi}_{\cu{x}_t, 1})\,\mineig(\bm{\Phi}_{\cu{x}_t, 3})}$.
	\end{enumerate}
\end{proposition}
\begin{proof}
	By Algorithm~\ref{def:taylor-moment-estimator}, 
	\begin{equation}
	\begin{split}
	\bm{\Sigma}_{2}&= \bm{\Phi}_{\cu{x}_t, 1}\Delta t + \frac{1}{2}\bm{\Phi}_{\cu{x}_t, 2}\Delta t^2\\
	&=\bm{\Gamma}\Delta t + \frac{1}{2} \left((\bm{\Gamma}\,\nabla\cu{f})^\trans  + \bm{\Gamma}\,\nabla\cu{f}  \right) \Delta t^2.
	\end{split}
	\end{equation}
	Thus $\bm{\eta}^\trans \,\bm{\Sigma}_2\,\bm{\eta} >0$ for any real non-zero vector $\bm{\eta}$ and $\Delta t >0$, if $\mineig(\bm{\Gamma})>0$ and $\mineig((\bm{\Gamma}\,\nabla\cu{f})^\trans  + \bm{\Gamma}\,\nabla\cu{f})\geq 0$ or $\mineig(\bm{\Gamma})=0$ and $\mineig((\bm{\Gamma}\,\nabla\cu{f})^\trans  + \bm{\Gamma}\,\nabla\cu{f})> 0$.
	
	For $\bm{\Sigma}_3$, by Theorem \ref{theorem:psd}, we have the polynomial 
	\begin{equation}
	P_3(\Delta t) = w_1 \Delta t + w_2\Delta t^2 + w_3\Delta t^3. 
	\end{equation}
	The polynomial $P_3(\Delta t)$ is positive and has no real roots for $\Delta t >0$, if and only if $w_2 > -2\sqrt{w_1w_3}$ and $w_3\geq 0$, which is equivalent to $\mineig(\bm{\Phi}_{\cu{x}_t, 2})>\frac{-2\sqrt{6}}{3}\sqrt{\mineig(\bm{\Phi}_{\cu{x}_t, 1})\,\mineig(\bm{\Phi}_{\cu{x}_t, 3})}$ and $\mineig(\bm{\Phi}_{\cu{x}_t, 3})\geq0$. It follows that $\bm{\Sigma}_3$ is p.d.
\end{proof}
\begin{remark}
	By considering all $\Delta t>0$, we have:
	\begin{enumerate}
		\item  A necessary condition for the $M$-th order TME covariance estimate to be p.d. is that $\bm{\Phi}_{\cu{x}_t, M}$ is p.s.d..
		\item The $1, 2, \ldots, M$ order TME covariance estimates are all p.d. if and only if $\left\lbrace \bm{\Phi}_{\cu{x}_t, r}\colon r=1,\ldots, M\right\rbrace $ are all p.s.d. and at least one of them is p.d. This can be proved by using recursion $\bm{\Sigma}_{M} = \bm{\Sigma}_{M-1} + \frac{\Delta t^M}{M!}\bm{\Phi}_{\cu{x}_t, M}$.
		\item In the limit $\Delta t\rightarrow 0$, the TME covariance estimate will always be p.d., provided $\bm{\Gamma}$ is p.d. 
	\end{enumerate} 
\end{remark}

\begin{example}
	\label{remark:benes}
	The TME variance estimate of SDE
	\begin{equation}
	\diff x_t = \tanh x_t \diff t + \diff W_t, 
	\end{equation}
	where $W_t$ is a standard Wiener process, is always p.d. This follows from Theorem~\ref{theorem:psd} and observing that $\Phi_{x, 1}=1>0$, $\Phi_{x, 2} = 1-\tanh^2 x_t \ge 0$, and $\left\lbrace \Phi_{x, r}=0\colon r\geq 3\right\rbrace$.
\end{example}

The coefficients $\bm{\Phi}_{\cu{x}_t, r}$, the expansion order $M$, and the time interval $\Delta t$ jointly determine the positive definiteness of TME covariance estimate. The properties of $\bm{\Phi}_{\cu{x}_t, r}$ are more of interest, as we usually have $M$ and $\Delta t$ fixed. Next, we show that $\bm{\Phi}_{\cu{x}_t, r}$ is only concerned with the SDE coefficients. 

\begin{lemma}
	\label{lemma:general-covariance}
	Consider the SDE~\eqref{equ:sde-for-analysis}. Let $\Phi_{\cu{x},r}^{u, v} \triangleq \left[\bm{\Phi}_{\cu{x}, r} \right]_{uv}$ be the $u$-th column and $v$-th row entry of $\bm{\Phi}_{\cu{x}, r}$. We denote $\alpha^u_{r} \triangleq \alpha_r(x_u) = \mathcal{A}^r(x_u)$, and partial derivative $\partial_i\alpha^u_{r} \triangleq \partial \alpha^u_{r} / \partial x_i$. Then a general expression of $\Phi_{\cu{x},r\geq1}^{u, v}$ is
	\begin{equation}
	\begin{split}
	\Phi_{\cu{x},r}^{u, v} &= \sum_{i, j}^D\sum^{r-1}_{s=0}\binom{r-1}{s}\left( \partial_i\alpha^u_s\,\partial_j\alpha^v_{r-s-1}\right)\Gamma_{ij} +\mathcal{A}\Phi_{\cu{x},r-1}^{u, v}\\ &=\sum^{r-1}_{s=0}\mathcal{A}^{s}\sum^{r-s-1}_{l=0}\binom{r-s-1}{l}\trace\left((\nabla\alpha^u_s)^\trans\,\nabla\alpha^v_{r-s-1-l}\,\bm{\Gamma}\right)
	\end{split}
	\label{equ:lemma-general-phi}
	\end{equation}
	starting from $\Phi^{u,v}_{\cu{x}, 0}=0$. 
\end{lemma}
\begin{proof}
	See Appendix~\ref{sec:append-2}.
\end{proof}

Lemma~\ref{lemma:general-covariance} gives an explicit form of $\bm{\Phi}_{\cu{x}_t, r}$, which is shown to be the function of $\cu{f}$, $\bm{\Gamma}$, and their partial derivatives. It implies that once $M$ and $\Delta t$ are given, the positive definiteness of $\bm{\Sigma}_M$ fully depends on the SDE coefficients. The functions $\cu{f}$ and $\bm{\Gamma}$ have to satisfy certain properties for $\bm{\Sigma}_M$ to be positive definite.

\begin{example}
	\label{remark:general-1-4}
	Let us consider a one-dimensional It\^{o} process 
	\begin{equation}
	\diff x_t = f(x_t)\diff t + L\diff W_t,
	\end{equation} then by Lemma~\ref{lemma:general-covariance},
	\begin{equation}
	\begin{split}
	\Phi_{x, 0} &= 0,\\
	\Phi_{x, 1} &= \Gamma, \\
	\Phi_{x, 2} &= 2f'\Gamma, \\
	\Phi_{x, 3} &= 2(2(f')^2 + 2ff'' + f'''\Gamma)\Gamma, \\
	\Phi_{x, 4} &= ((9(f'')^2 + 6ff''''+16f'''f')\Gamma + 8(f')^3  \\
	&\quad+6f^2f'''+26ff''f'+\frac{24}{16}f'''''\Gamma^2)\Gamma,
	\end{split}
	\end{equation}
	where $\Gamma = L^2\,Q$, and $Q$ is the diffusion constant of the Wiener process $W_t$.
\end{example}

\subsection{Stability of TME Gaussian Filter and Smoother}
\label{sec:stability}
It is important and useful that the filter and smoother are in some sense \emph{stable}. Some classical stability results for linear Kalman filters can be found in~\cite{jazwinski2007stochastic,AndersonMoore1981} while more recent results on the stability of different types of non-linear Kalman filters have been analysed in~\cite{Reif1999,XiongZhangChan2006,Karvonen2018}.
In this section, we follow~\cite{Karvonen2018} and prove that the TME Gaussian filter and smoother are stable in the mean-square sense if a number of assumptions on the system and the sigma-point approximation, verifiable before the filter is run, are satisfied.
This means that we show that
\begin{equation*}
\sup_{k \geq 1} \expec\big[ \norm{\cu{x}_{k} - \cu{m}_k}^2 \big] < \infty,
\end{equation*}
and 
\begin{equation*}
\sup_{k \geq 1} \expec\big[ \norm{\cu{x}_{k} - \cu{m}^s_k}^2 \big] < \infty,
\end{equation*}
where the expectation is taken over all state and measurement trajectories.

For simplicity, we assume that $\cu{V}_k = \cu{V}$ for all $k \geq 1$, and that the measurement model is linear: $\cu{h}(\cu{x}_k)=\cu{H}\, \cu{x}_k$ for some matrix $\cu{H}$. Then the generic continuous-discrete SDE~\eqref{equ:sde-for-analysis} has the discretised form 
\begin{equation}
\begin{split}
\cu{x}_{k} &= \cu{a}(\cu{x}_{k-1}, \Delta t) + \bm{\tau}(\cu{x}_{k-1}, \Delta t) , \\
\cu{y}_k &= \cu{H}\,\cu{x}_k + \cu{v}_k,
\end{split}
\label{equ:stab-discrete}
\end{equation}
where $\bm{\tau}(\cu{x}_{k-1})\triangleq\bm{\tau}(\cu{x}_{k-1}, \Delta t)$ is a zero-mean random variable whose covariance $\cov[\bm{\tau}(\cu{x}_{k-1})]= \bm{\Sigma}(\cu{x}_{k-1}, \Delta t)$. Notice that we denote by $\cu{a}(\cu{x}_{k-1})\triangleq\cu{a}(\cu{x}_{k-1}, \Delta t)$ and $\bm{\Sigma}(\cu{x}_{k-1})\triangleq\bm{\Sigma}(\cu{x}_{k-1}, \Delta t)$ the \textit{exact} mean and covariance functions of $\cu{x}_t$, respectively. It follows that $\cu{a} = \cu{a}_M + \cu{R}_M$, where $\cu{R}_M(\cu{x}_{k-1})\triangleq \cu{R}_M(\cu{x}_{k-1}, \Delta t)$ is the Taylor remainder \eqref{equ:expansion-remainder}. The assumptions needed for the stability analysis of the system \eqref{equ:stab-discrete} are collected below in Assumption~\ref{assum:1}.
If $\cu{A}$ is a matrix, $\norm{\cu{A}}$ stands for the spectral norm in the following.

\begin{assumption}\label{assum:1}\leavevmode The following properties hold:
	\begin{enumerate}
		\item There are non-negative constants $C_M$, $\lambda_\tau$, and $\lambda_P$ such that $\sup_{ k \geq 1 } \norm{\cu{R}_M(\cu{x}_{k-1})}\leq C_M$ almost surely, $\sup_{k \geq 1} \expec[\tr(\bm{\Sigma}(\cu{x}_{k-1}))]\leq\lambda_\tau$, and $\sup_{k \geq 1} \expec [ \trace(\cu{P}_k) ]  \leq \lambda_P$.
		\item There is $C \geq 0$ such that
		\begin{equation*}
		\norm{ \cu{a}_M(\cu{x}) - \mathcal{S}_{\cu{m}, \cu{P}} ( \cu{a}_M) } \leq \norm{\nabla\cu{a}_M(\cu{x})}^2 \norm{ \cu{x} - \cu{m}}^2 + C \trace(\cu{P})
		\end{equation*}
		for any vectors $\cu{x}$ and $\cu{m}$ and any positive semi-definite matrix $\cu{P}$, where $\mathcal{S}_{\cu{m},\cu{P}}(\cu{g})$ stands for the sigma-point approximation of the Gaussian integral
		\begin{equation*}
		\int \cu{g}(\cu{x}) \,\mathcal{N}(\cu{x} \mid \cu{m}, \cu{P}) \diff \cu{x}.
		\end{equation*}
		\item There is $\lambda \geq 0$ such that $\sup_{k \geq 1} \norm{\cu{I} - \cu{K}_k \cu{H}} \leq \lambda$ almost surely and 
		\begin{equation*}
		\lambda_f^2 \triangleq \sup_{k \geq 1} \, \lambda^2 \sup_{\cu{x}} \, \norm{\nabla \cu{a}_M(\cu{x})}^2 < 1/4.
		\end{equation*}
	\end{enumerate}
	
\end{assumption}

\begin{theorem} \label{ref:thm-stability-filter} Suppose that Assumption \ref{assum:1} is satisfied. 
	Then the TME Gaussian filter for system \eqref{equ:stab-discrete} has
	\begin{equation*}
	\expec\left[ \norm{\cu{x}_{k} - \cu{m}_k}^2 \right] \leq (4\lambda_f^2)^k \trace(\cu{P}_0) + \frac{C_e}{1-4\lambda_f^2}
	\end{equation*}
	for all $k \geq 1$, where $C_e$ is defined in~\eqref{equ:stability-bound-const}.
\end{theorem}
\begin{proof}
It is easy to see that~\cite[Proof of Theorem~IV.3]{Karvonen2018}
	\begin{equation*}
	\kappa \triangleq \sup_{k \geq 1} \, \norm{\cu{K}_k} \leq \lambda_P \norm{\cu{H}} \norm{\cu{V}^{-1}}.
	\end{equation*}	
	Denote $\cu{A}_k = \cu{I}-\cu{K}_k \cu{H}$.
	Using the discretised system \eqref{equ:stab-discrete}, the filtering error can be written as
	\begin{equation*}
	\begin{split}
	\cu{e}_k \triangleq{}& \cu{x}_{k} - \cu{m}_k \\
	={}& \cu{a}(\cu{x}_{k-1}) + \bm{\tau}(\cu{x}_{k-1}) - \cu{m}_k^- - \cu{K}_k\,( \cu{y}_k - \cu{H} \,\cu{m}_k^-) \\
	={}& \cu{A}_k\, \big[ \cu{a}(\cu{x}_{k-1}) - \mathcal{S}_{\cu{m}_{k-1}, \cu{P}_{k-1}}(\cu{a}_M) \big] \\
	&+ (\cu{I} - \cu{K}_k \,\cu{H}) \,\bm{\tau}(\cu{x}_{k-1}) - \cu{K}_k \,\cu{v}_k \\
	={}& \cu{A}_k \, \big[ \cu{a}_M(\cu{x}_{k-1}) - \mathcal{S}_{\cu{m}_{k-1}, \cu{P}_{k-1}}(\cu{a}_M) \big] \\
	&+ \cu{A}_k \, \cu{R}_M(\cu{x}_{k-1}) + \cu{A}_k \, \bm{\tau}(\cu{x}_{k-1}) - \cu{K}_k \, \cu{v}_k.
	\end{split}
	\end{equation*}
	The inequality $(a_1 + \cdots + a_n)^2 \leq n (a_1^2 + \cdots + a_n^2)$ gives
	\begin{equation} \label{equ:ek-inequality}
	\begin{split}
	\expec\left[ \norm{\cu{e}_k}^2 \right] \leq{}& 4 \expec\left[ \norm[\big]{ \cu{A}_k \,\left[ \cu{a}_M(\cu{x}_{k-1}) - \mathcal{S}_{\cu{m}_{k-1}, \cu{P}_{k-1}}(\cu{a}_M) \right] }^2 \right] \\
	&+ 4 \expec\big[ \norm{ \cu{A}_k \,\cu{R}_M(\cu{x}_{k-1}) }^2 \big] \\
	&+ 4\expec\big[ \norm{\cu{A}_k \, \bm{\tau}(\cu{x}_{k-1})}^2 \big] + 4\expec\big[ \norm{ \cu{K}_k \, \cu{v}_k }^2 \big].
	\end{split}
	\end{equation}
	Assumption~\ref{assum:1} yields the following bounds:
	\begin{align*}
	\expec\Big[ \norm[\big]{\cu{A}_k \big[ \cu{a}_M(\cu{x}_{k-1}) - \mathcal{S}_{\cu{m}_{k-1}, \cu{P}_{k-1}}&(\cu{a}_M) \big]}^2\Big]  \\
	&\leq \lambda_f^2 \expec\big[\norm{ \cu{e}_{k-1}}^2\big] + C \lambda^2 \lambda_P, \\
	\expec\big[\norm{ \cu{A}_k \, \cu{R}_M(\cu{x}_{k-1}) }^2\big] &\leq C_M^2 \lambda^2, \\
	\expec\big[ \norm{\cu{A}_k \, \bm{\tau}(\cu{x}_{k-1})}^2 \big] &\leq \lambda^2 \lambda_{\tau}, \\
	\expec\big[ \norm{ \cu{K}_k \, \cu{v}_k}^2 \big] &\leq \trace(\cu{V}) \, \kappa^2.
	\end{align*}
	Upon insertion of these estimates into~\eqref{equ:ek-inequality} we get the recursive mean-square error inequality
	\begin{equation*}
	\expec\big[ \norm{\cu{e}_k}^2 \big] \leq 4\lambda_f^2\, \expec\big[ \norm{ \cu{e}_{k-1}}^2 \big] + C_e,
	\end{equation*}
	where
	\begin{equation} \label{equ:stability-bound-const}
	C_e = 4 \big( \lambda^2 [C \lambda_P + C_M^2 + \lambda_\tau] + \trace(\cu{V}) \kappa^2 \big).
	\end{equation}
	Because we have assumed that $4\lambda_f^2 < 1$, the claim then follows from the discrete Gr\"{o}nwall's inequality (e.g.,~\cite[Theorem~IV.2]{Karvonen2018}).
\end{proof}

The Assumption~\ref{assum:1} postulates conditions on the sigma-point approximations and systems. A trivial example to satisfy the assumptions is that the drift function $\cu{f}$ is smooth enough and all of its partial derivatives up to certain orders are uniformly bounded. More practical examples that satisfy Assumption~\ref{assum:1} can be found in~\cite{Karvonen2018}. It is typically necessary that the measurement model matrix $\cu{H}$ is a scaled identity matrix and the discretised dynamics $\cu{a}$ in~\eqref{equ:stab-discrete} defines an exponentially stable system.

\begin{proposition}\label{prop:smoother}
Suppose that Assumption~\ref{assum:1} is satisfied. Then
\begin{equation*}
  \begin{split}
  \expec\big[ \norm{\cu{x}_k - \cu{m}^s_k}^2 \big] \leq{}& 2\trace( \cu{P}_0 ) + \frac{2 C_e}{1-4\lambda_f^2} \\
  &+ 2 \, \max_{ 1 \leq i \leq k} \, \expec \big[ \norm{\cu{G}_i}^2 \norm{\cu{m}^s_{i+1} - \cu{m}^-_{i+1}}^2 \big].
  \end{split}
\end{equation*}
for any $k \geq 1$, where $C_e$ is defined in~\eqref{equ:stability-bound-const}.
\end{proposition}
\begin{proof} 
  The inequality $(a_1 + \cdots + a_n)^2 \leq n(a_1^2 + \cdots + a_n^2)$ and the smoother recursion yield
  \begin{equation*}
    \expec\big[ \norm{\cu{x}_k - \cu{m}^s_k}^2 \big] \leq 2 \, \expec\big[ \norm{\cu{e}_k}^2 \big]  + 2 \, \expec \big[ \norm{\cu{G}_k}^2 \norm{\cu{m}^s_{k+1} - \cu{m}^-_{k+1}}^2 \big].
  \end{equation*}
  The claim now follows from Theorem~\ref{ref:thm-stability-filter}.
\end{proof}

Proposition~\ref{prop:smoother} implies boundedness of the smoothing error, for example, when $\norm{\cu{G}_k}$ is bounded almost surely, and $\expec\left[ \norm{\cu{m}^s_{k+1} - \cu{m}^-_{k+1}}^2 \right]$ is bounded. 

\section{Numerical Experiments}
\label{sec:numerical-experiments}
To examine the effectiveness of the TME estimator in Algorithm~\ref{def:taylor-moment-estimator}, we first conduct experiments on the moment estimation of SDEs. After that, we finally examine and compare the accuracy and numerical stability of the proposed TME Gaussian filter and smoother with the state-of-the-art methods. 

\subsection{Moment Estimation of SDEs}
\label{sec:experiment-moments}
In this part, we consider two non-linear SDEs:
\begin{equation}
\diff x_t = \tanh x_t \diff t + \diff W_t,  \label{equ:tanh}
\end{equation} 
and
\begin{equation}
\diff x_t = -a^2\sin x_t\,\cos^3x_t \diff t + a\cos^2x_t \diff W_t, 
\label{equ:sincos3}
\end{equation}
where $W_t$ is a standard Wiener process. The aim is to compare the estimates of the moments of the transition densities. The true mean and covariance are estimated using Monte Carlo (MC) sampling with $10^6$ independent trajectories. We simulated the samples from the models using Euler--Maruyama with sufficiently small time interval ($10^{-5}$~s). The initial conditions for model \eqref{equ:tanh} and \eqref{equ:sincos3} were $x_0=0.5$ and $x_0=1$, respectively. The estimates were examined in the time interval $T = 0$~s to $T = 5$~s.

We chose the following methods as described in Sections~\ref{sec:moment-odes} and~\ref{sec:moment-ito-taylor} to compare with the TME method:
\begin{itemize}
	\item the ODE approach by solving~\eqref{equ:dmdt_mPP} using Gaussian assumption, 4th order Runge--Kutta solver, and 3rd order Gauss--Hermite integration (Gauss-RK4);
	\item the ODE approach by solving~\eqref{equ:dmdt_mPP} using linearisation, and 4th order Runge--Kutta solver (Linear-RK4);
	\item the Ito--Taylor strong order 1.5 based approach from~\cite{ckdIndian2010, peter-koleden-book} (It\^{o}-1.5).
\end{itemize}

\begin{figure}[h!]
	\centering
	\includegraphics[width=.7\linewidth]{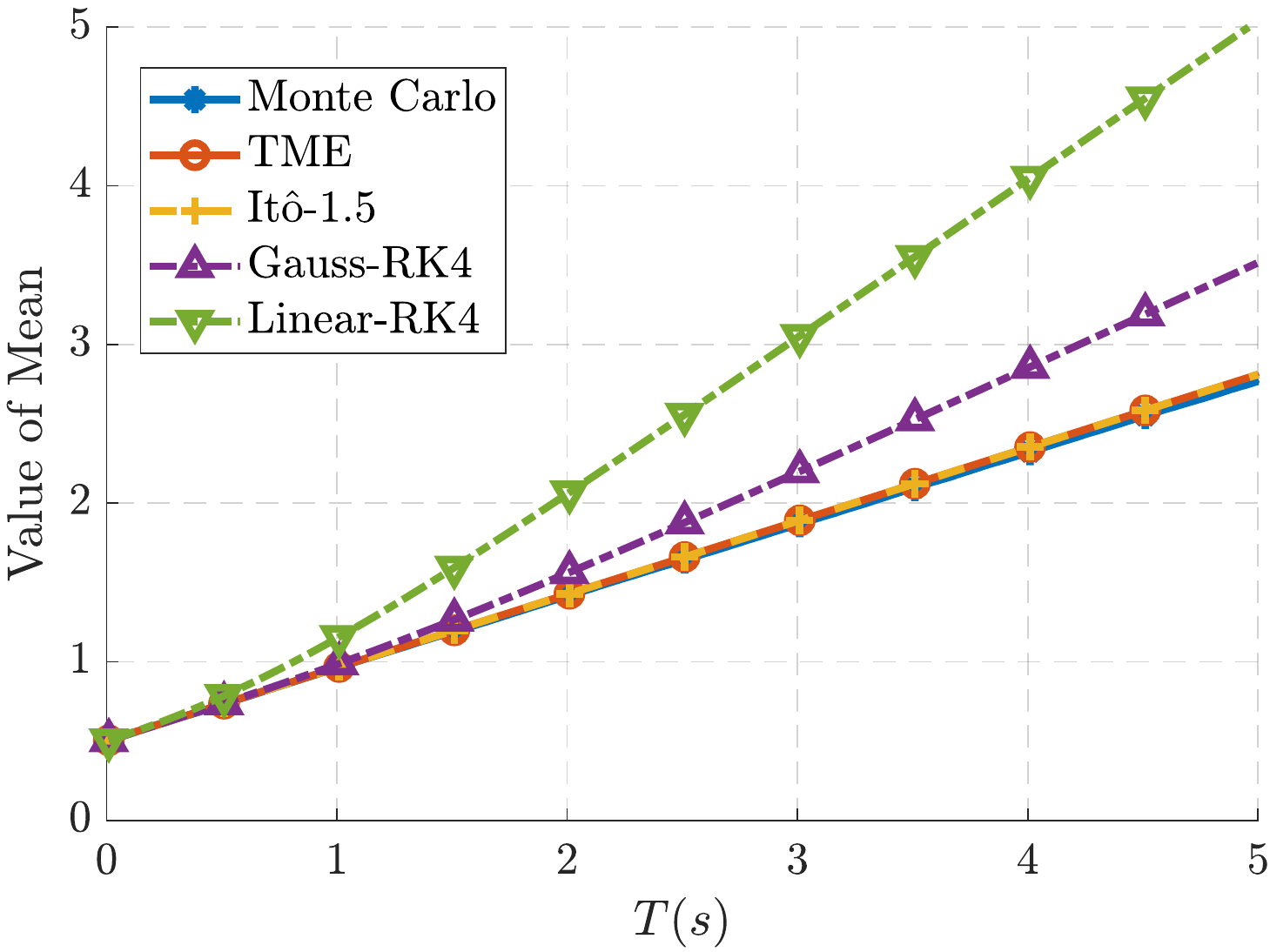}
	\includegraphics[width=.7\linewidth]{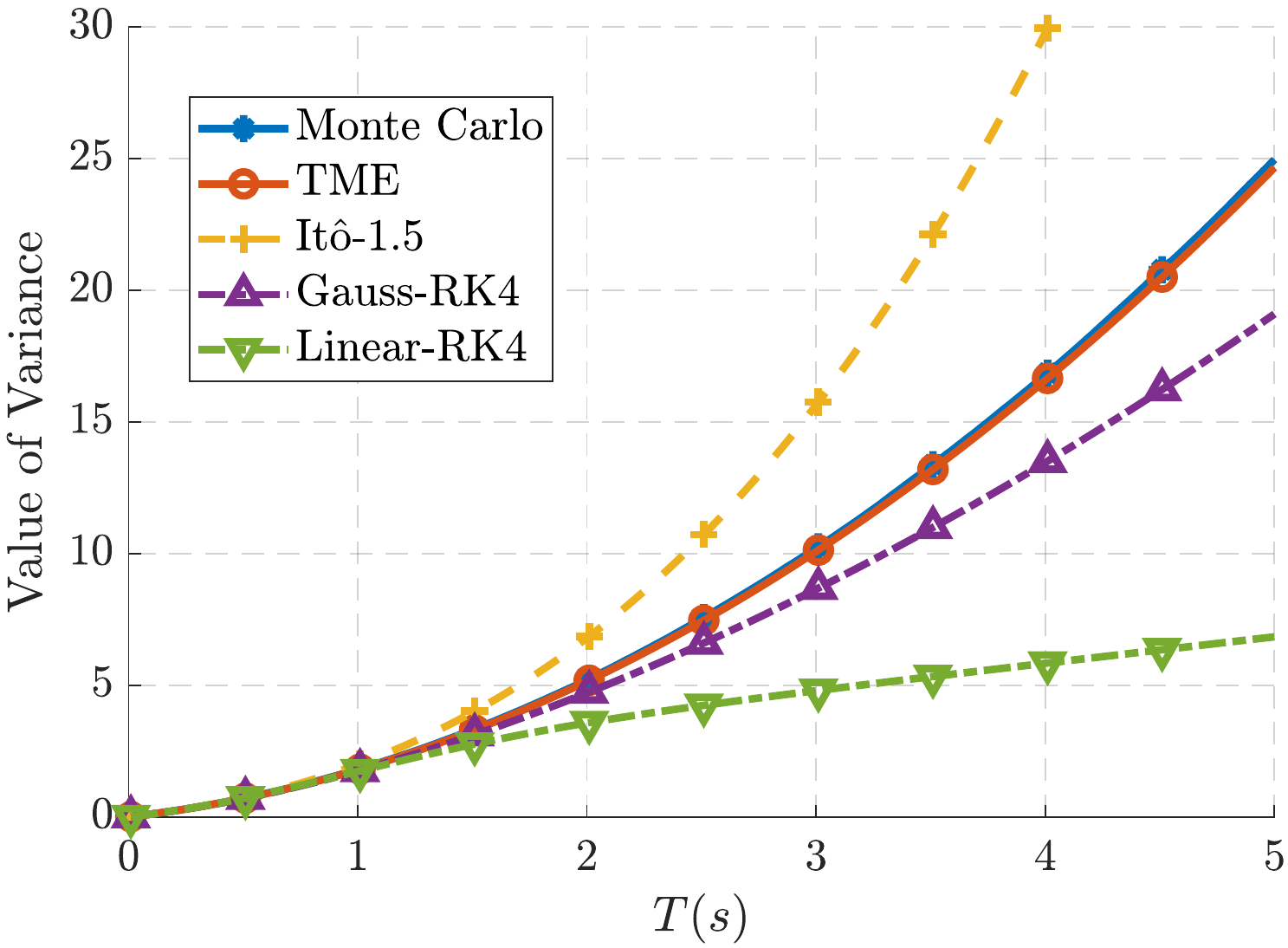}
	\caption{The mean and variance estimates of model \eqref{equ:tanh}. }
	\label{fig:moment-benes}
\end{figure}

For the first model \eqref{equ:tanh}, in Figure~\ref{fig:moment-benes}, we show the mean and variance estimates as functions of time. We observe that the It\^{o}-1.5 and (the second order) TME methods coincide for the estimation of mean function, and are closest to the Monte Carlo result. This is because their formulations for this model are identical, and exact to the true mean of \eqref{equ:tanh}. The Gauss-RK4 and Linear-RK4 can only estimate the mean accurately within short time intervals, while Gauss-RK4 is slightly better than Linear-RK4. 

When estimating the variance of \eqref{equ:tanh}, only the TME method succeeds to follow the Monte Carlo result closely. The variance estimate yield by TME is $\Delta t + (1-\tanh^2 x_0)\Delta t^2$, which is exact to this model (see, Example~\ref{remark:benes}). The Gauss-RK4, Linear-RK4, and It\^{o}-1.5 all deviate from Monte Carlo for long time intervals. 

For the second model \eqref{equ:sincos3}, we obtain the Monte Carlo samples by using its explicit solution $x_t = \atan(aW_t + \tan(x_0))$ directly rather than Euler--Maruyama~\cite{peter-koleden-book}. The parameter $a$ controls the non-linearity of the model, where we choose $a=1.5$ in this experiment. The It\^{o}-1.5 method is not applicable for this model, as the dispersion term depends on $x_t$, thus we adopt scalar Milstein's method instead. 
\begin{figure}[h!]
	\centering
	\includegraphics[width=.49\linewidth]{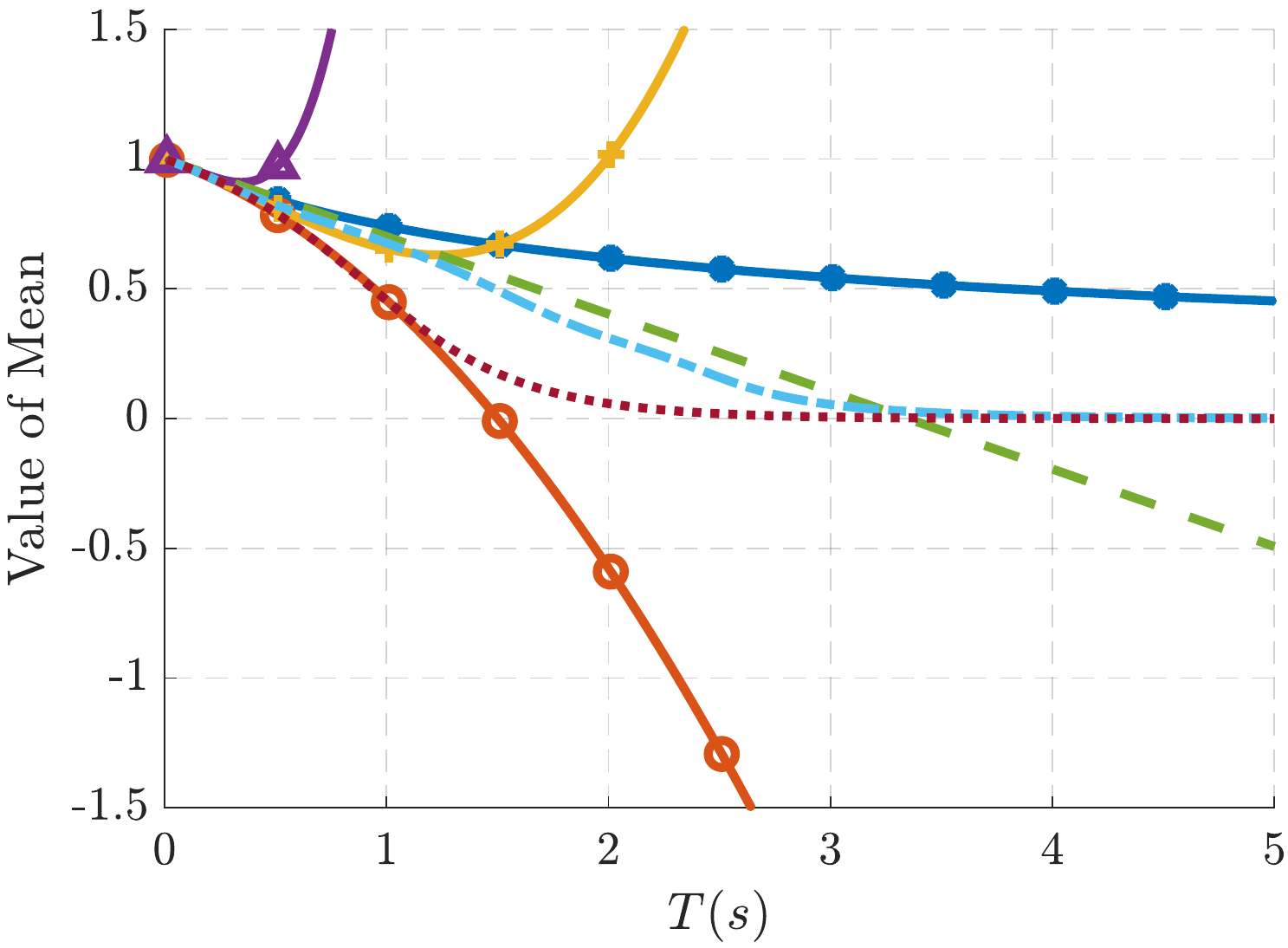}
	\includegraphics[width=.49\linewidth]{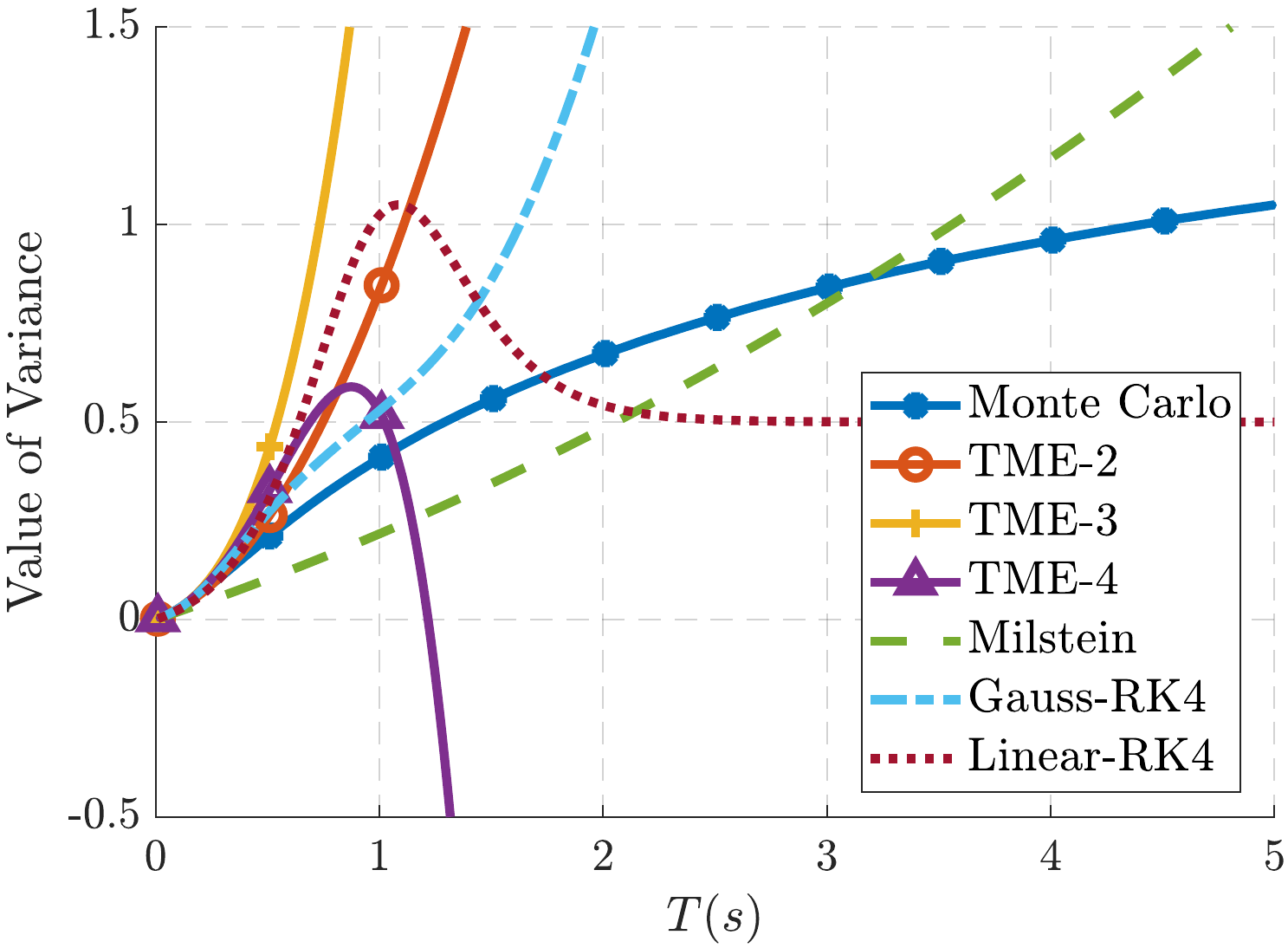}
	\caption{The mean and covariance estimates of model \eqref{equ:sincos3}.}
	\label{fig:moment-sinx-value}
\end{figure}
\begin{figure}[h!]
	\centering
	\includegraphics[width=.49\linewidth]{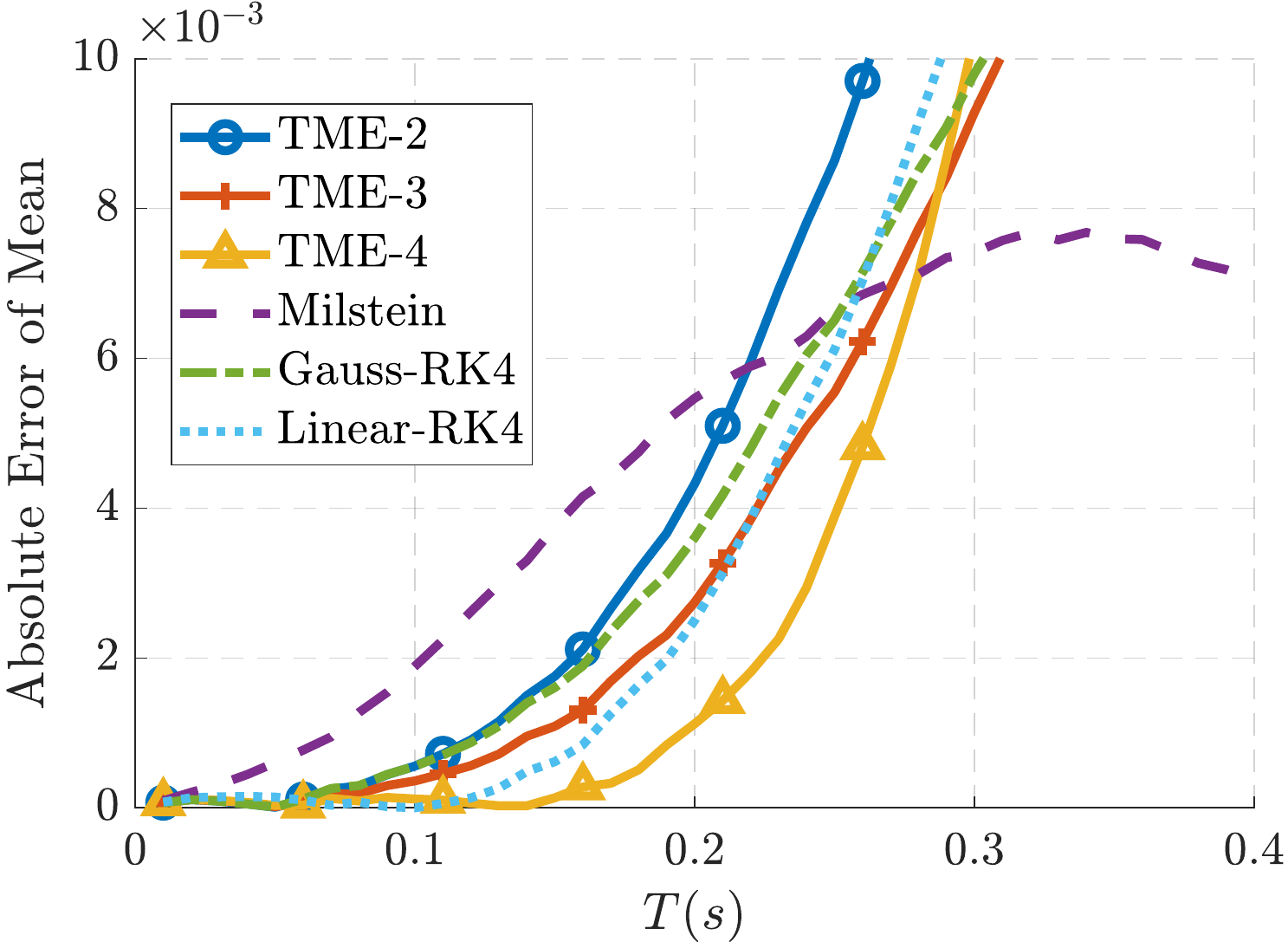}
	\includegraphics[width=.49\linewidth]{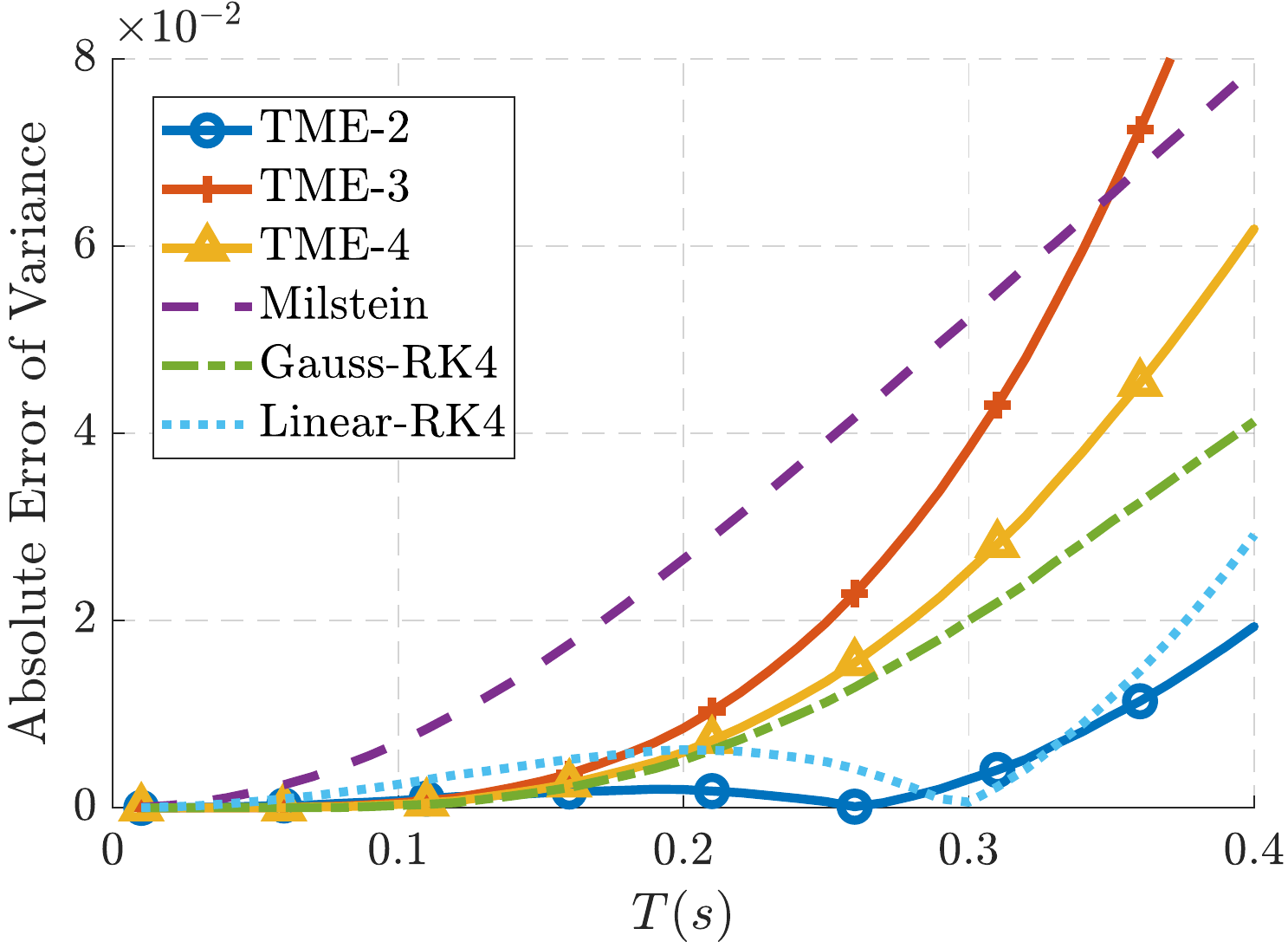}
	\caption{The absolute error of the mean and variance estimates of model \eqref{equ:sincos3}. }
	\label{fig:moment-sinx-error}
\end{figure}

Figure~\ref{fig:moment-sinx-value} shows the mean and variance estimates estimates of model~\eqref{equ:sincos3}. We find that although the Gauss-RK4 and Linear-RK4 follow the Monte Carlo mean estimate well, they still have large deviation. For the variance estimate, none of the methods works well with long time intervals. The TME estimates with different orders (TME-2, TME-3, and TME-4) diverge from around $T=1$~s, which is expected, as they are polynomial type of estimates. For long time intervals, the variance estimates by all these methods fail to different extents. 

Notice that the variance estimates of TME in Figure~\ref{fig:moment-sinx-value} are also the values of polynomial~\eqref{equ:lemma1-poly1} in Theorem~\ref{theorem:psd}. It can be observed that the TME-2 (second order) and TME-3 estimates of variance are always positive definite for $\Delta t>0$, while TME-4 will turn to negative from around $\Delta t>1$.

Because of the non-linearity of this model, it is more demonstrative to show the absolute errors of the estimates in a short time interval ($T=0$~s to $T=0.4$~s), as given in Figure~\ref{fig:moment-sinx-error}. The absolute error of each estimate is calculated with respect to the Monte Carlo simulation. We clearly see that by increasing the TME order, the estimates of mean and variance get better. The TME-4 outperforms other methods for the mean estimate, while TME-2 is surprisingly the best for estimation of the variance. 

\subsection{3D Coordinate Turn Tracking}
\label{sec:3d-tracking}
In this part, we conduct Gaussian filtering and smoothing on a 3D coordinated turn model. Performing filtering and smoothing on this model is considered challenging due to its non-linearities and high dimensionality \cite{barshalomBook2002, ckdIndian2010}. The model is given by
\begin{subequations}
	\begin{align}
	\diff \cu{x}_t &= \cu{f}(\cu{x}_t)\diff t + \cu{L}\diff \cu{W}_t, \label{equ:3d-coor-model-dynamic}\\
	\cu{y}_k &= \begin{bmatrix}
	\sqrt{p_x^2 + p_y^2 + p_z^2} \\
	\tan^{-1}(p_y / p_x) \\
	\tan^{-1} (p_z / \sqrt{p_x^2 + p_y^2})
	\end{bmatrix}+ \bm{\epsilon}_k, 
	\label{equ:3d-coor-model-measur}
	\end{align}
\end{subequations}
where the state $\cu{x}_t = \begin{bmatrix}
p_x & v_x & p_y & v_y & p_z & v_z & \theta
\end{bmatrix}^\trans$ and 
\begin{equation}
\begin{split}
\cu{f}(\cu{x}_t) &= \begin{bmatrix}
v_x & -\theta v_y & v_y & \theta v_x & v_z & 0 & 0
\end{bmatrix}^\trans, \\
\cu{L} &= \diag \begin{bmatrix}
0& \sigma_1& 0& \sigma_1& 0& \sigma_1& \sigma_2
\end{bmatrix}.
\end{split}
\end{equation}
In this model, we denote by $p_x$, $p_y$, and $p_z$ the position of the target in Cartesian coordinates, and $v_x$, $v_y$, and $v_z$ are the corresponding velocities. State $\theta$ governs the turning rate of the target which controls the non-linearity of this model. In addition, $\cu{W}_t$ is a standard Wiener process. The discrete measurement model \eqref{equ:3d-coor-model-measur} gives observations from a position-fixed radar and returns the range, angle of azimuth, and elevation \cite{SARKKA20121221, barshalomBook2002}. The measurement noise $\bm{\epsilon}_k \sim \mathcal{N}(\bm{\epsilon}_k\mid\cu{0}, \cu{V})$ and $\cu{V} = \diag\begin{bmatrix}\sigma_r^2 & \sigma_\theta^2 & \sigma_\phi^2\end{bmatrix}$. 

We simulate the trajectories using a similar parameter setting as described in \cite{ckdIndian2010, SARKKA20121221}, except that we choose the initial turning rate $\theta_0 = 30\degree/\text{s}$, which is significantly more challenging than the original settings ($3$ to $6\degree/\text{s}$). The difference of the two settings is illustrated in Figure~\ref{fig:turn-rate-diff}.
\begin{figure}[h!]
	\centering
	\includegraphics[width=.49\linewidth]{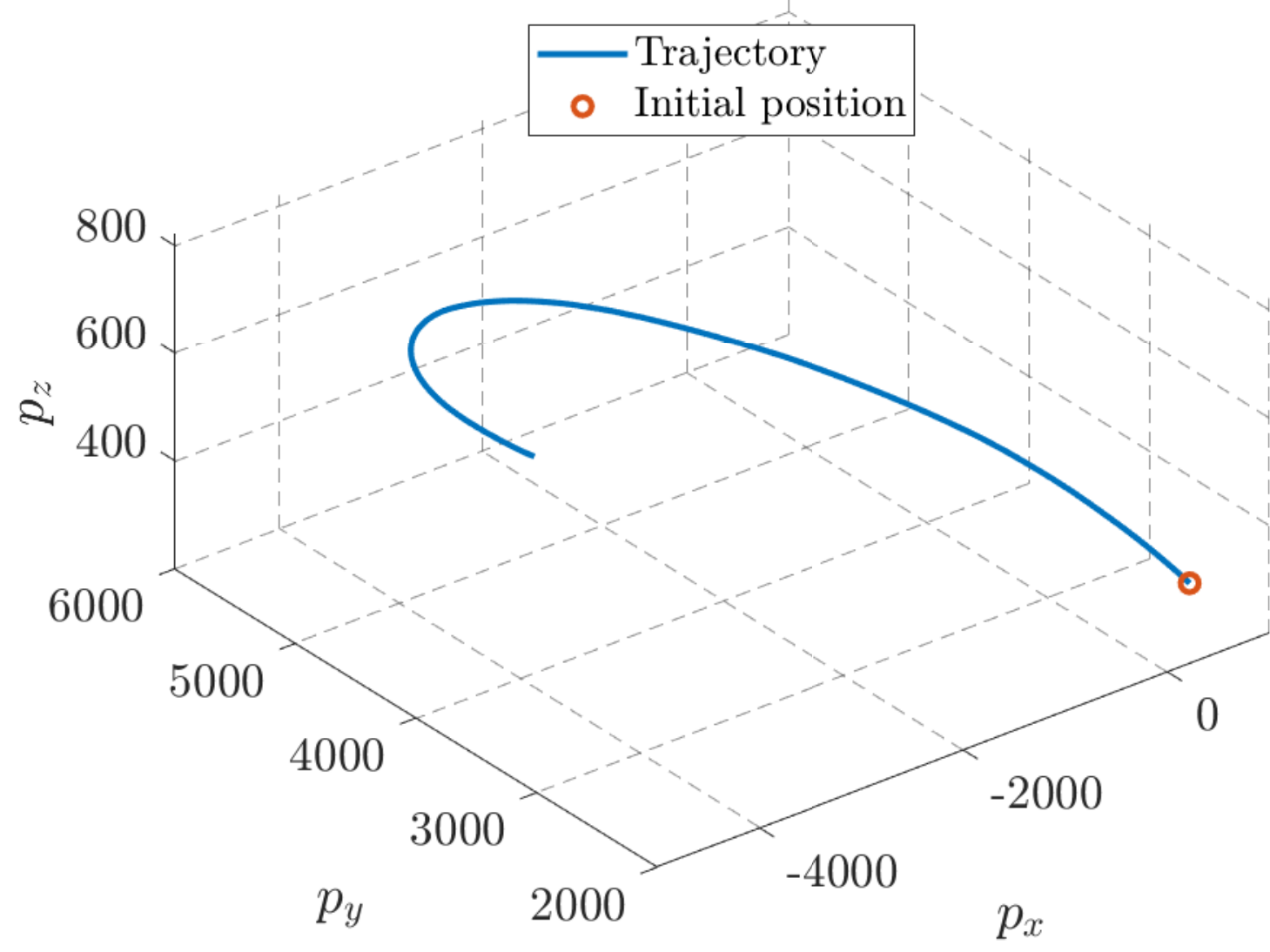}
	\includegraphics[width=.49\linewidth]{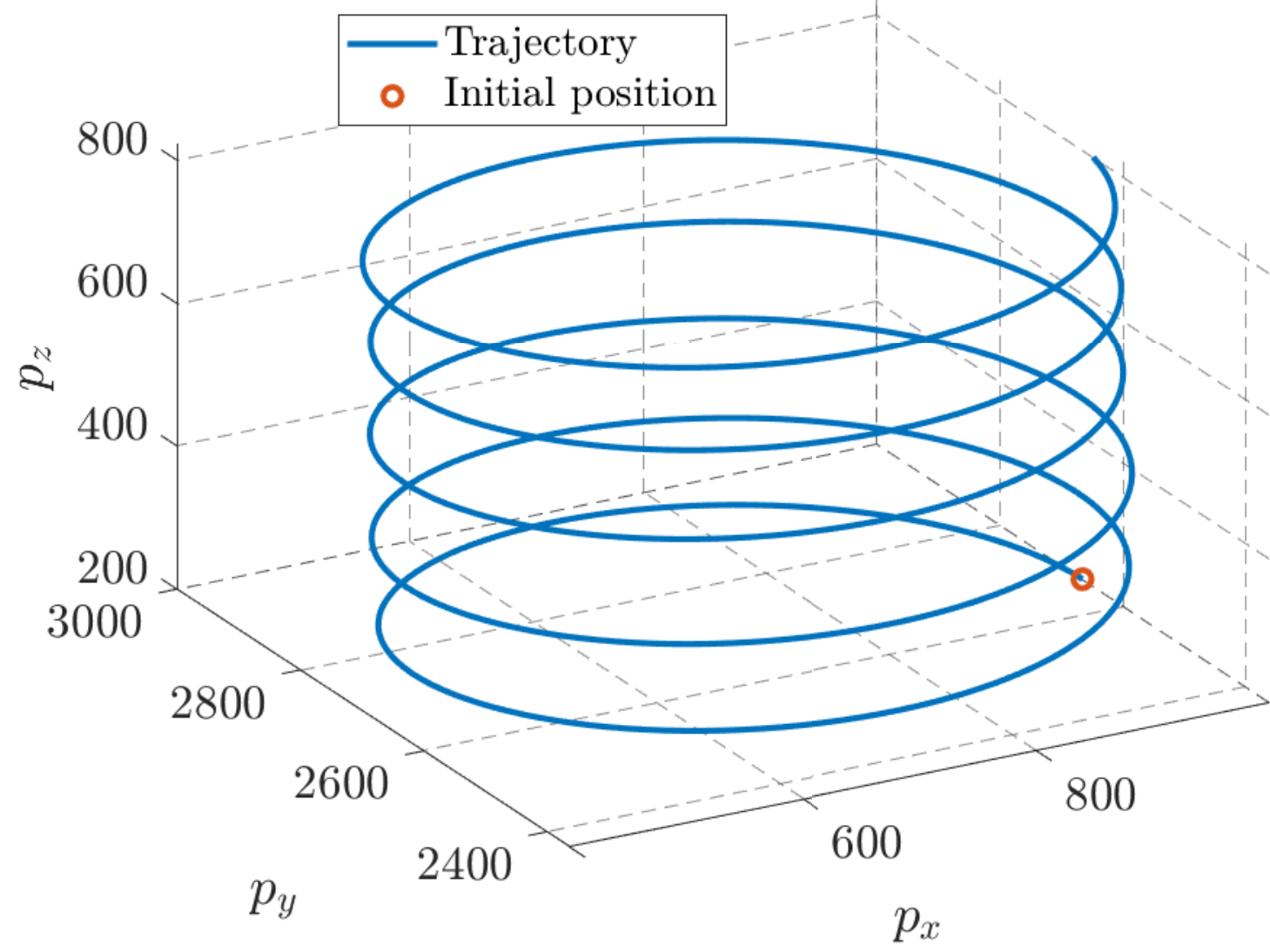}
	\caption{Examples of trajectory simulation up to 60~s. The left figure shows the setting of initial turning rate $\theta_0 =3\degree/\text{s}$ \cite{ckdIndian2010, SARKKA20121221}, while on the right figure, we use $\theta_0=30\degree/\text{s}$.}
	\label{fig:turn-rate-diff}
\end{figure}

The other parameters are the same as in \cite{ckdIndian2010, SARKKA20121221}: We choose process covariance parameters $\sigma_1=\sqrt{0.2}$ and $\sigma_2=7\times 10^{-3}$, and measurement covariance parameters $\sigma_r=50\,\text{m}$, $\sigma_\phi=0.1\degree$, and $\sigma_\theta=0.1\degree$. The initial condition is drawn from a normal distribution with mean $\cu{m}_0 = [\begin{smallmatrix}1000\,\text{m}& 0 \,\text{m/s}& 2650\, \text{m}& 150\,\text{m/s}& 200 \, \text{m}& 10 \,\text{m/s}& 30\degree/\text{s}\end{smallmatrix}]^\trans$ and covariance $\cu{P}_0=\diag[\begin{smallmatrix}100^2& 100^2& 100^2& 100^2& 100^2& 100^2& 10^2\end{smallmatrix}]$. We simulate the ground-truth trajectories using Euler--Maruyama with small enough time step $\Delta t \times 10^{-5}$~s, where $\Delta t$ is the time interval between two measurements. The total time length of the trajectory is fixed to $T=210$~s.

To test the effectiveness of filters and smoothers, we select the time interval $\Delta t$ range from $0.5$~s to $9$~s, because the filters and smoothers may fail with large $\Delta t$. We also use additional integration steps in prediction and smoothing, and choose the number of integration steps $M$ from $\log_2(M)=0$ to $\log_2(M)=5$. 

\begin{figure*}[t!]
	\centering
	\includegraphics[width=0.19\linewidth]{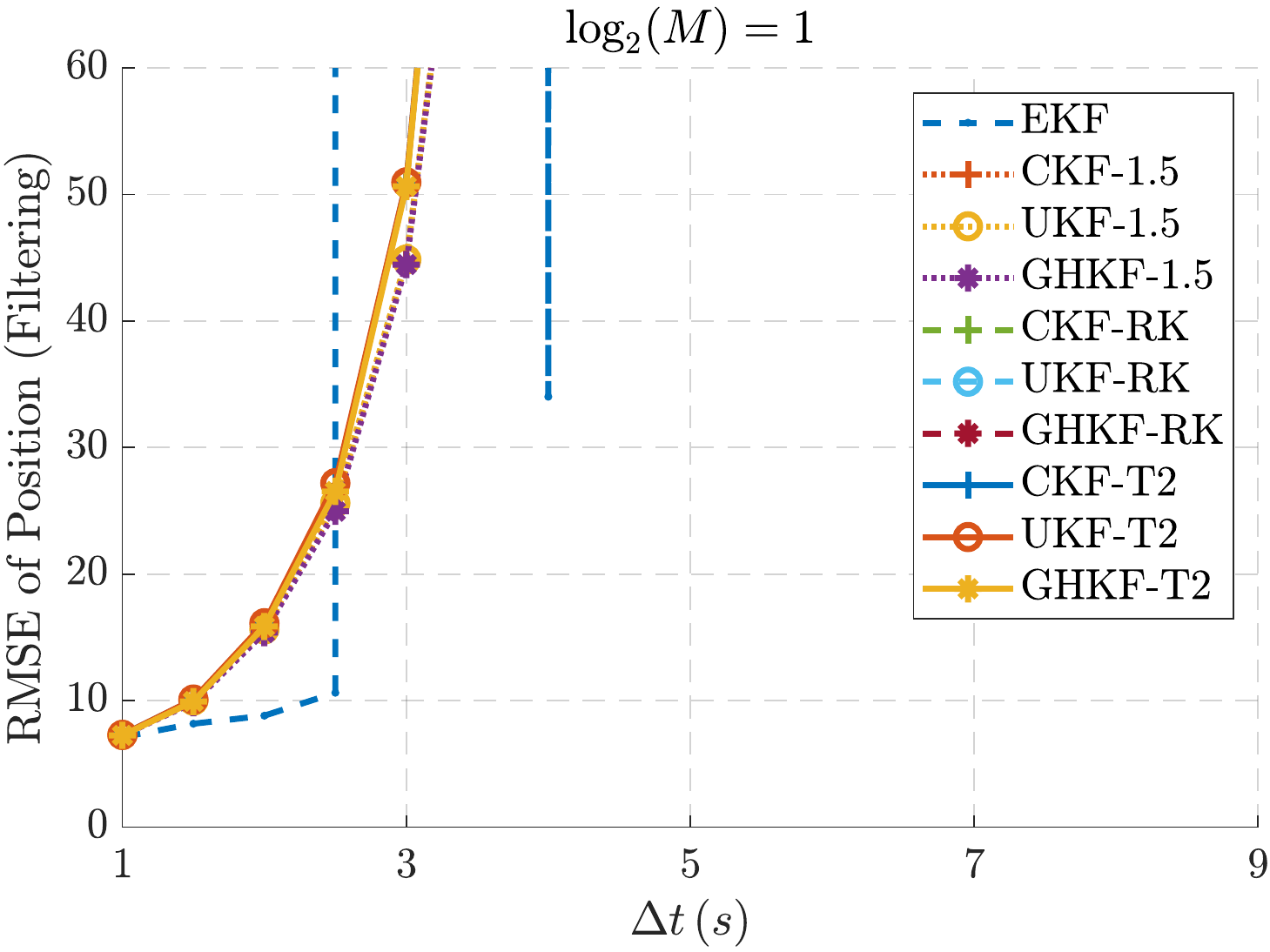}
	\includegraphics[width=0.19\linewidth]{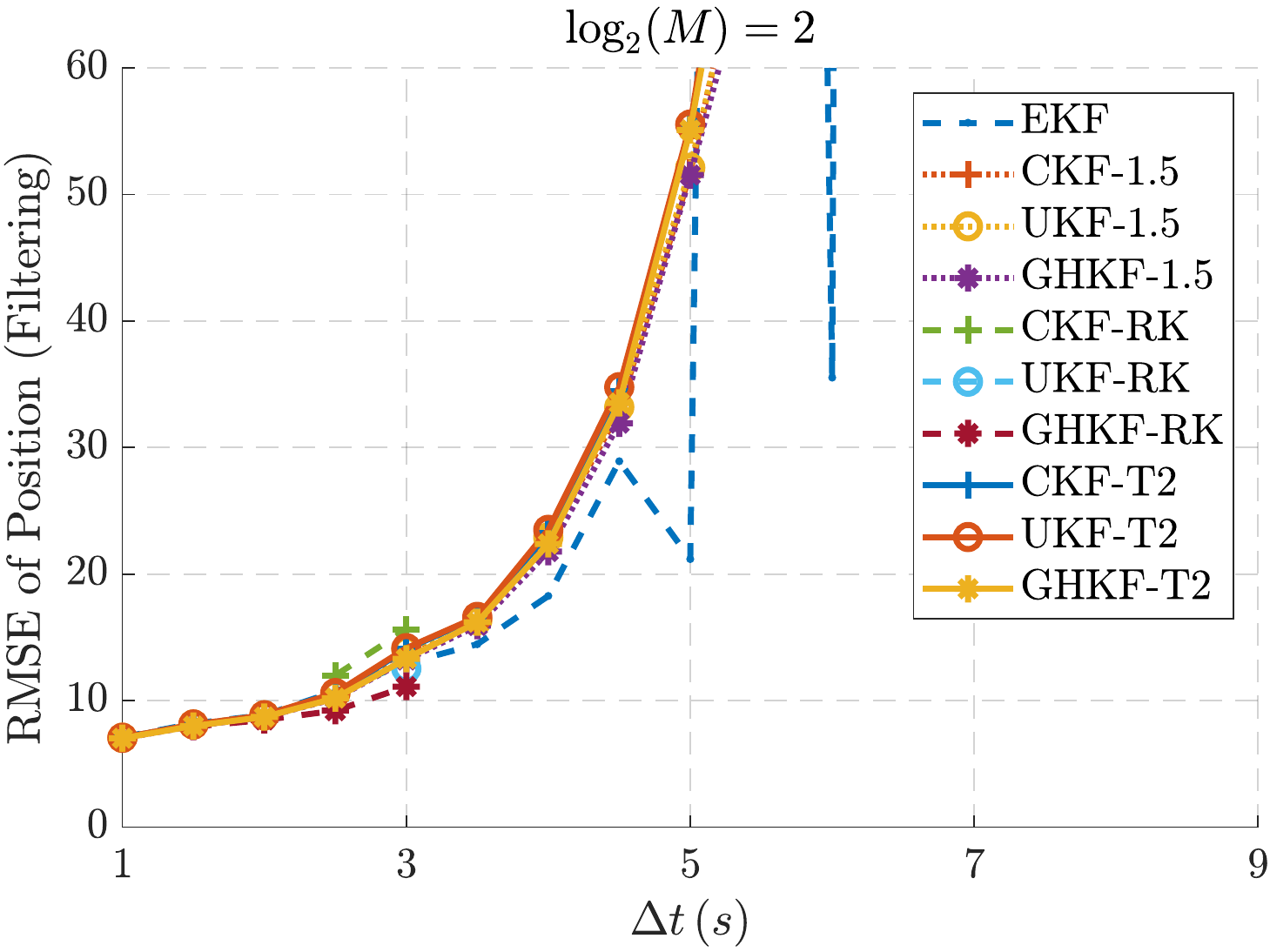}
	\includegraphics[width=0.19\linewidth]{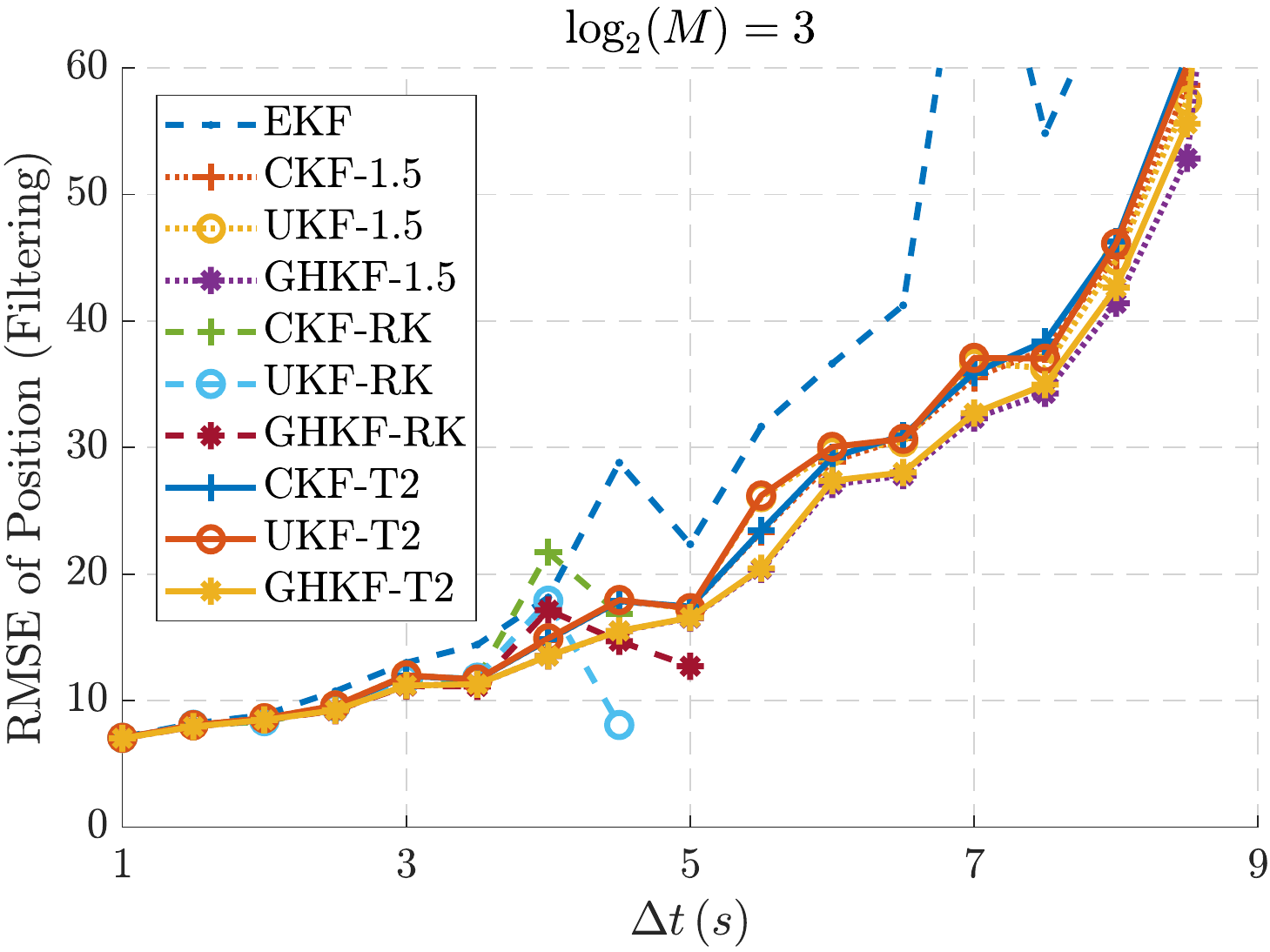}
	\includegraphics[width=0.19\linewidth]{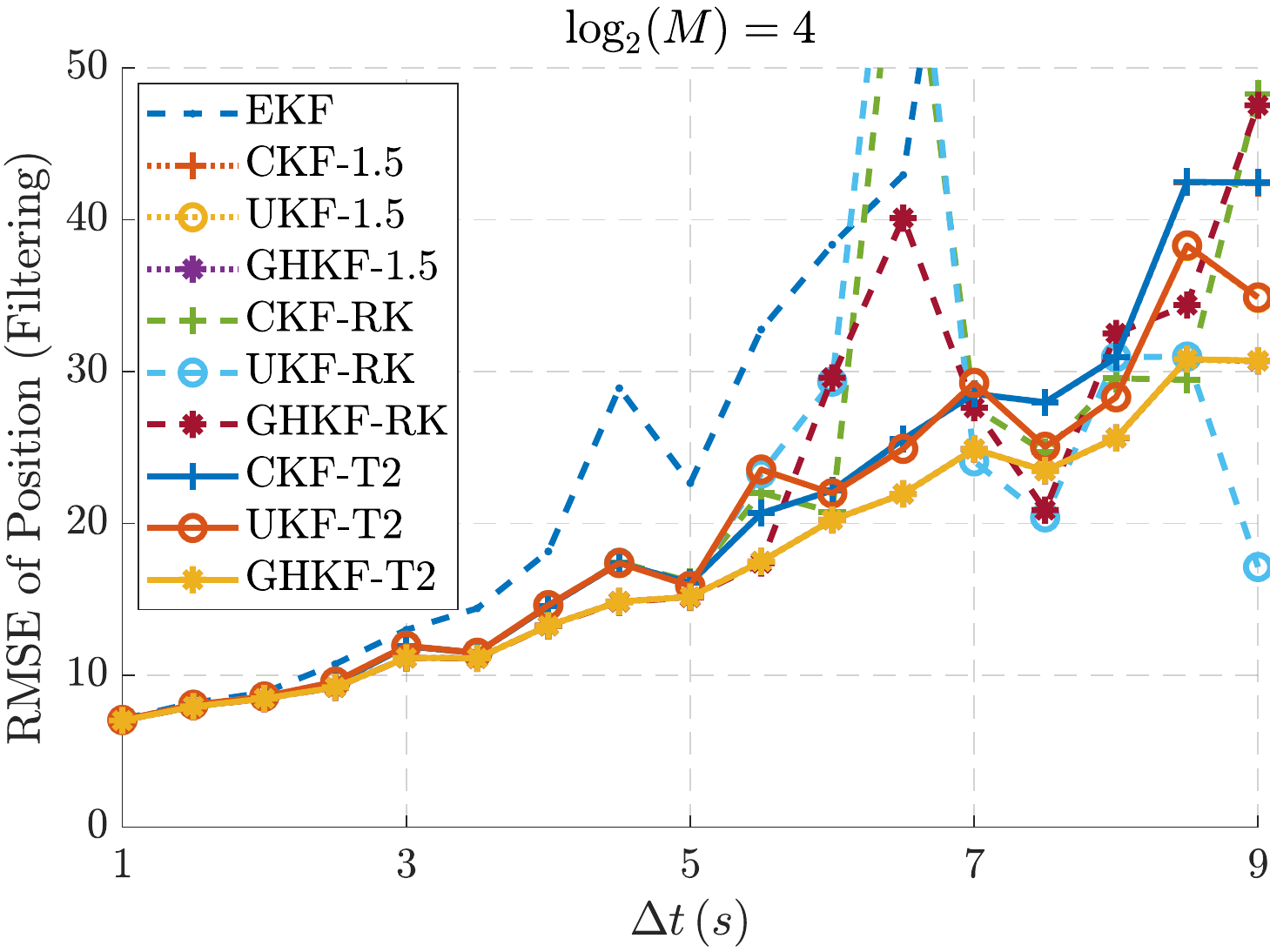}
	\includegraphics[width=0.19\linewidth]{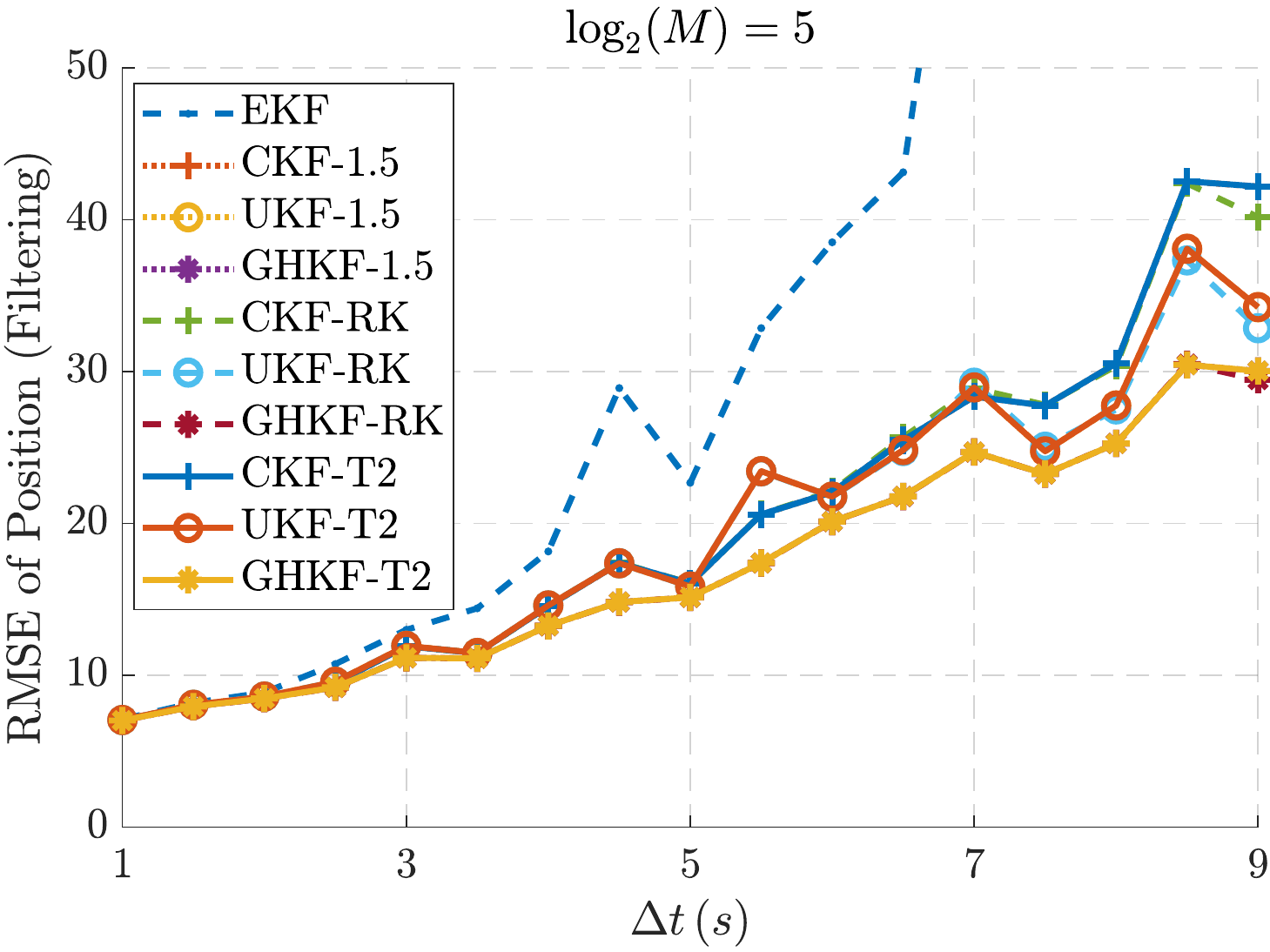} \\
	\includegraphics[width=0.19\linewidth]{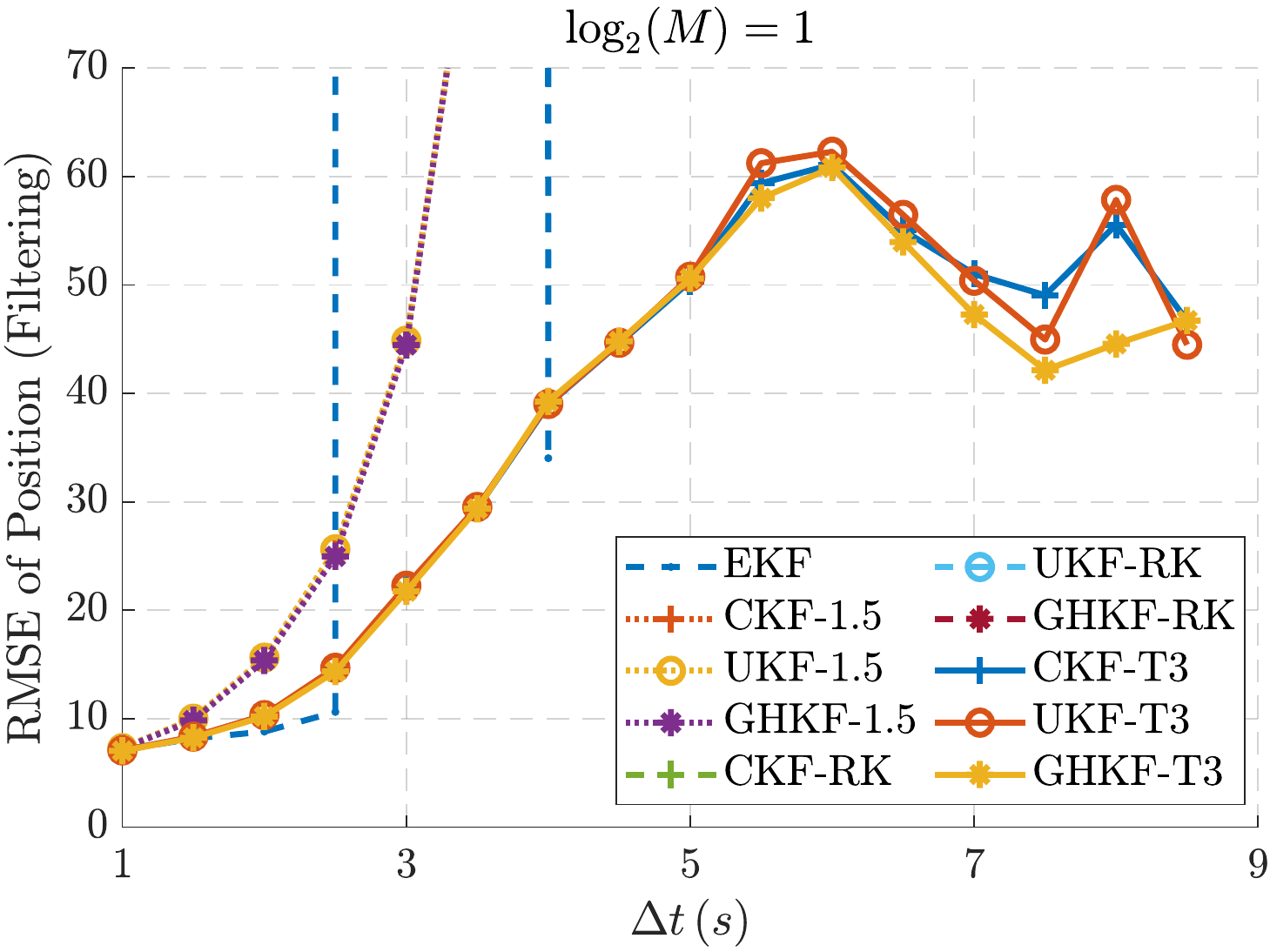}
	\includegraphics[width=0.19\linewidth]{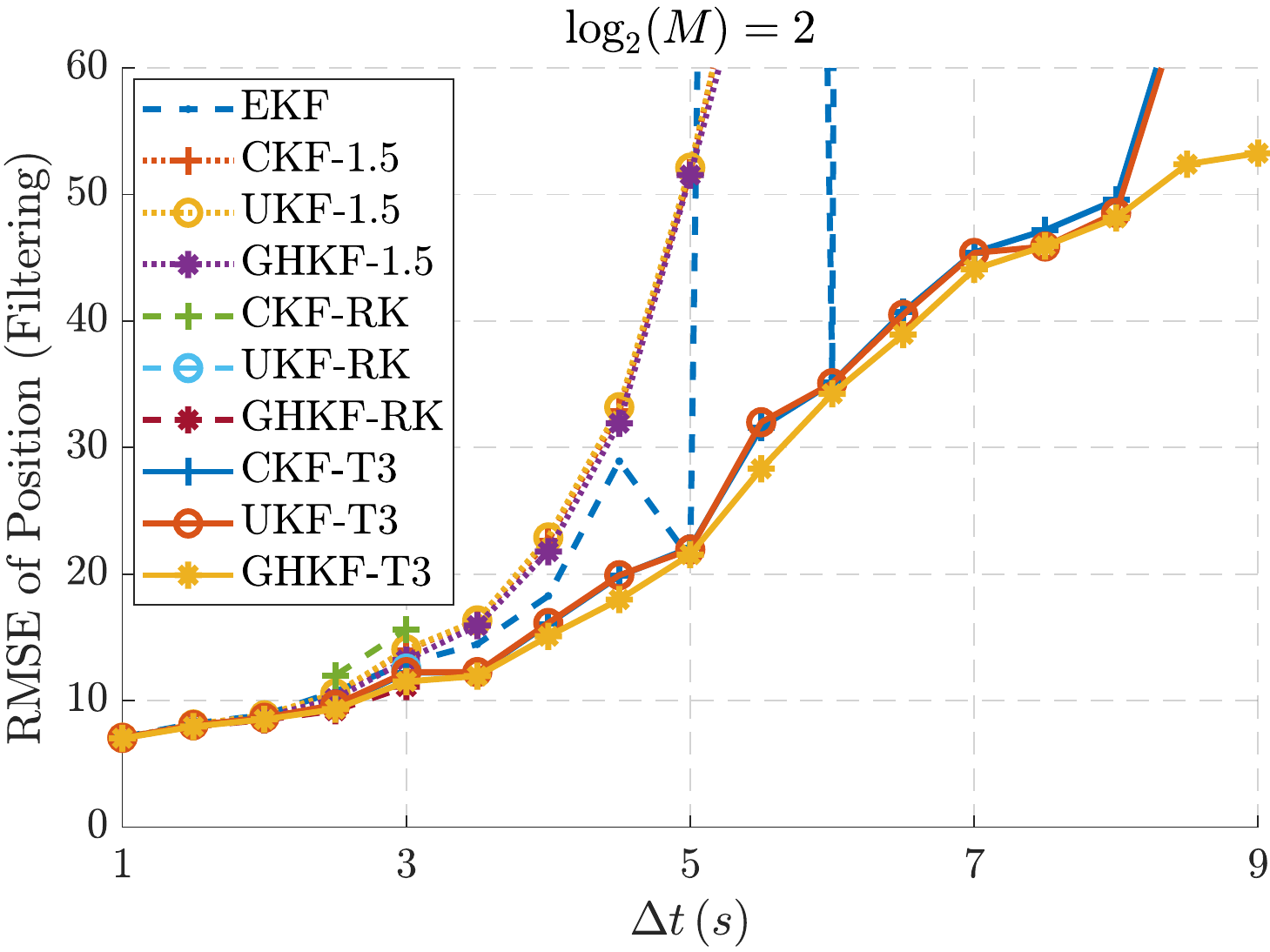}
	\includegraphics[width=0.19\linewidth]{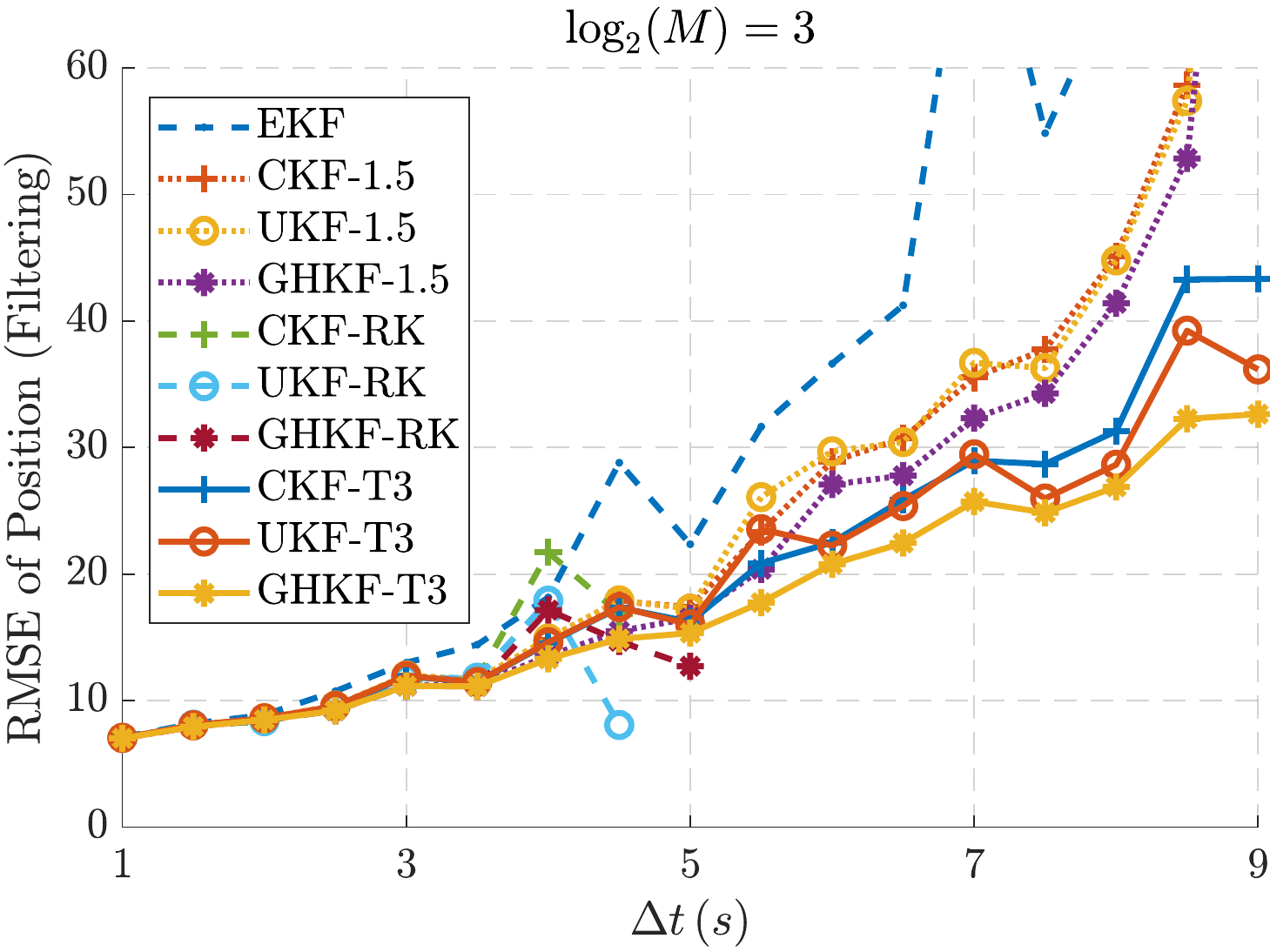}
	\includegraphics[width=0.19\linewidth]{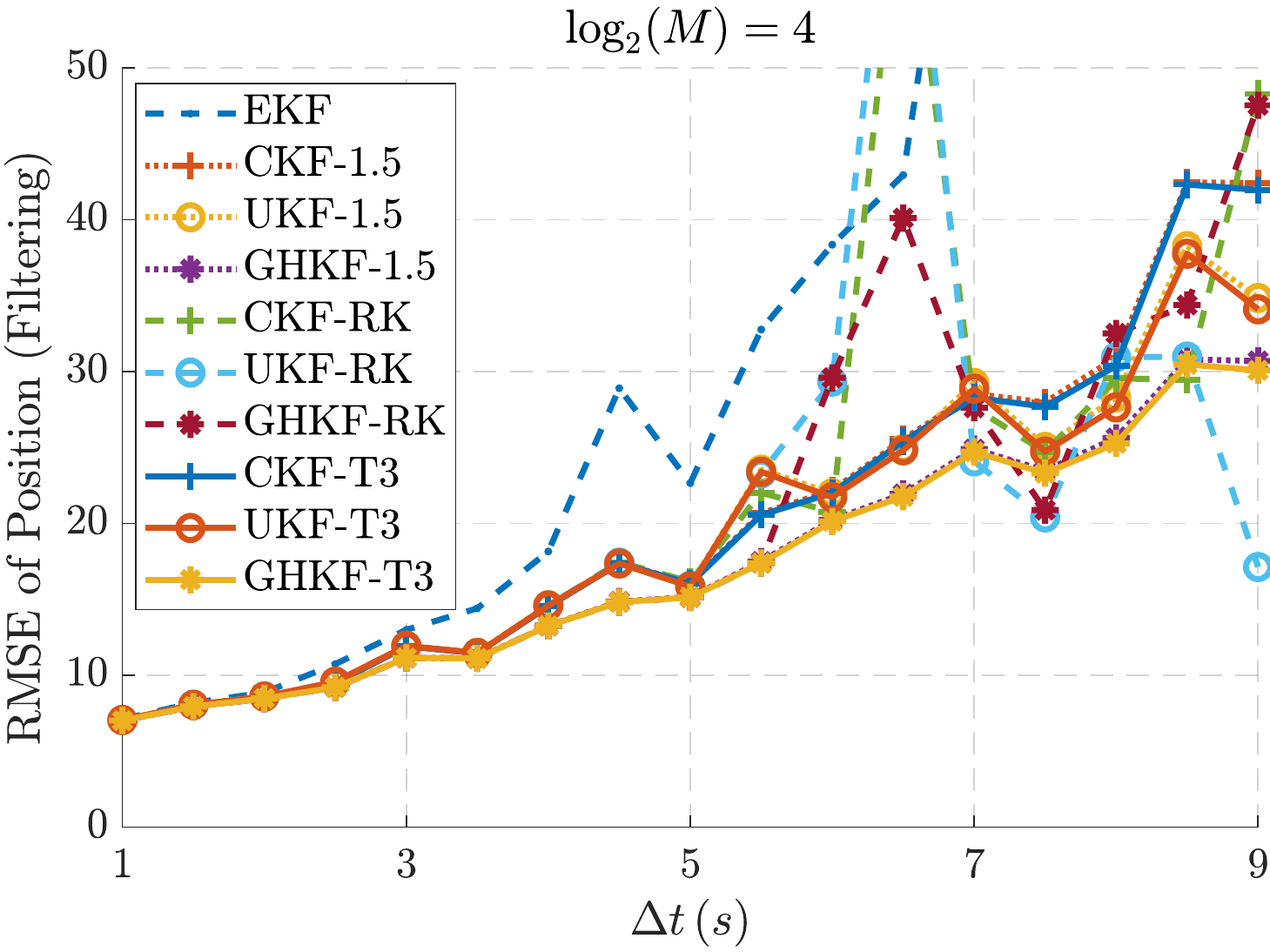}
	\includegraphics[width=0.19\linewidth]{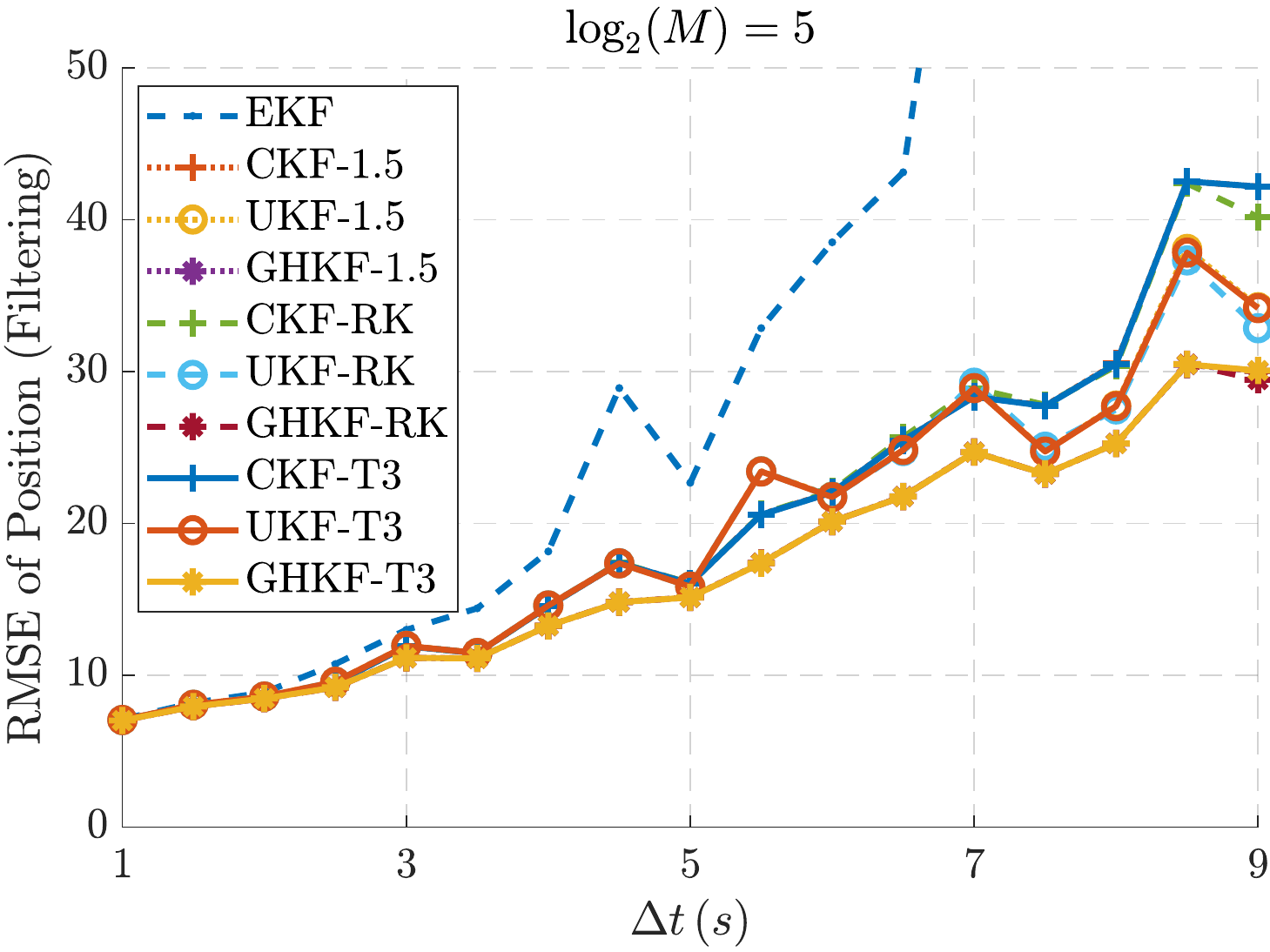} \\
	\includegraphics[width=0.19\linewidth]{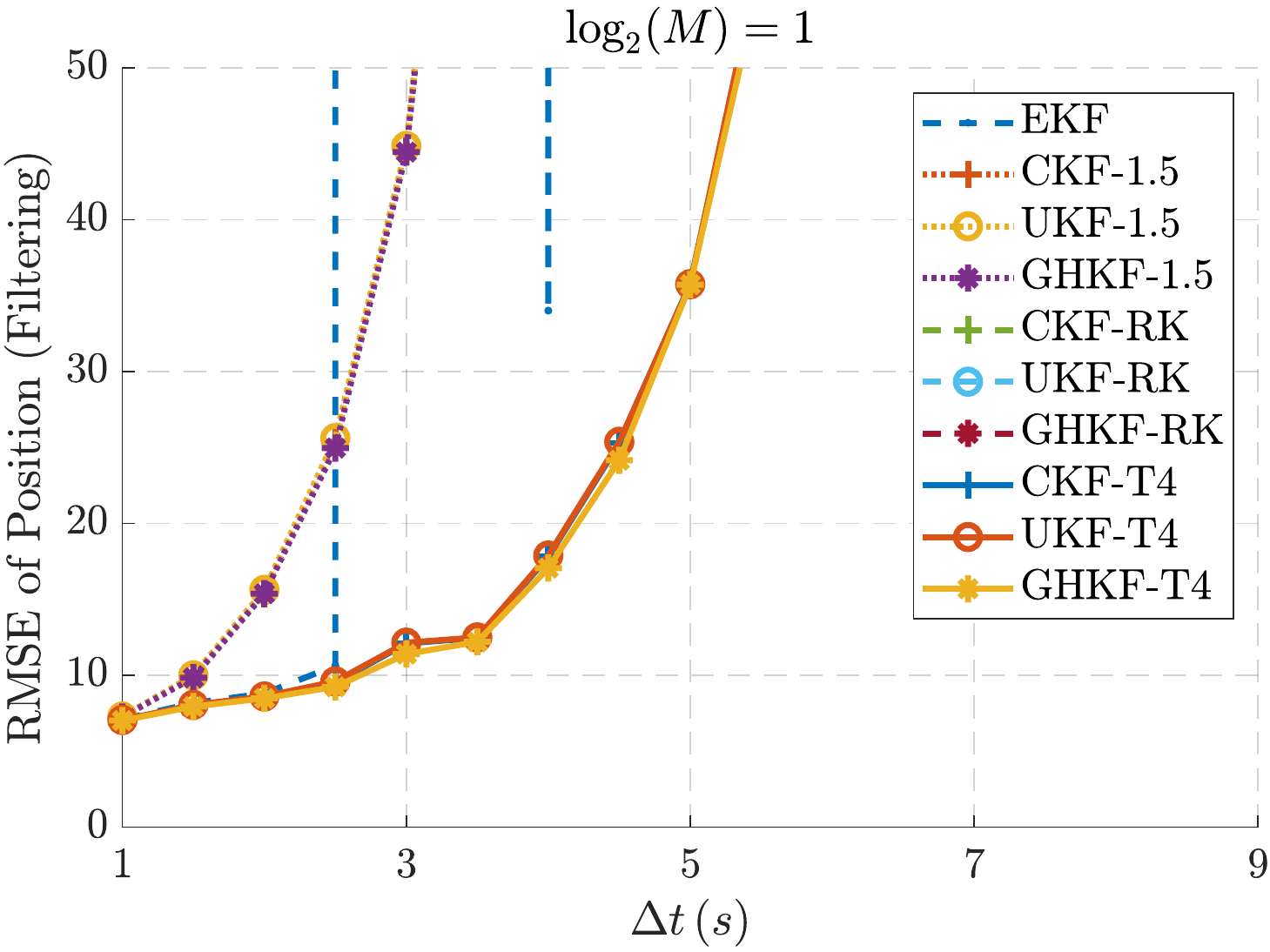}
	\includegraphics[width=0.19\linewidth]{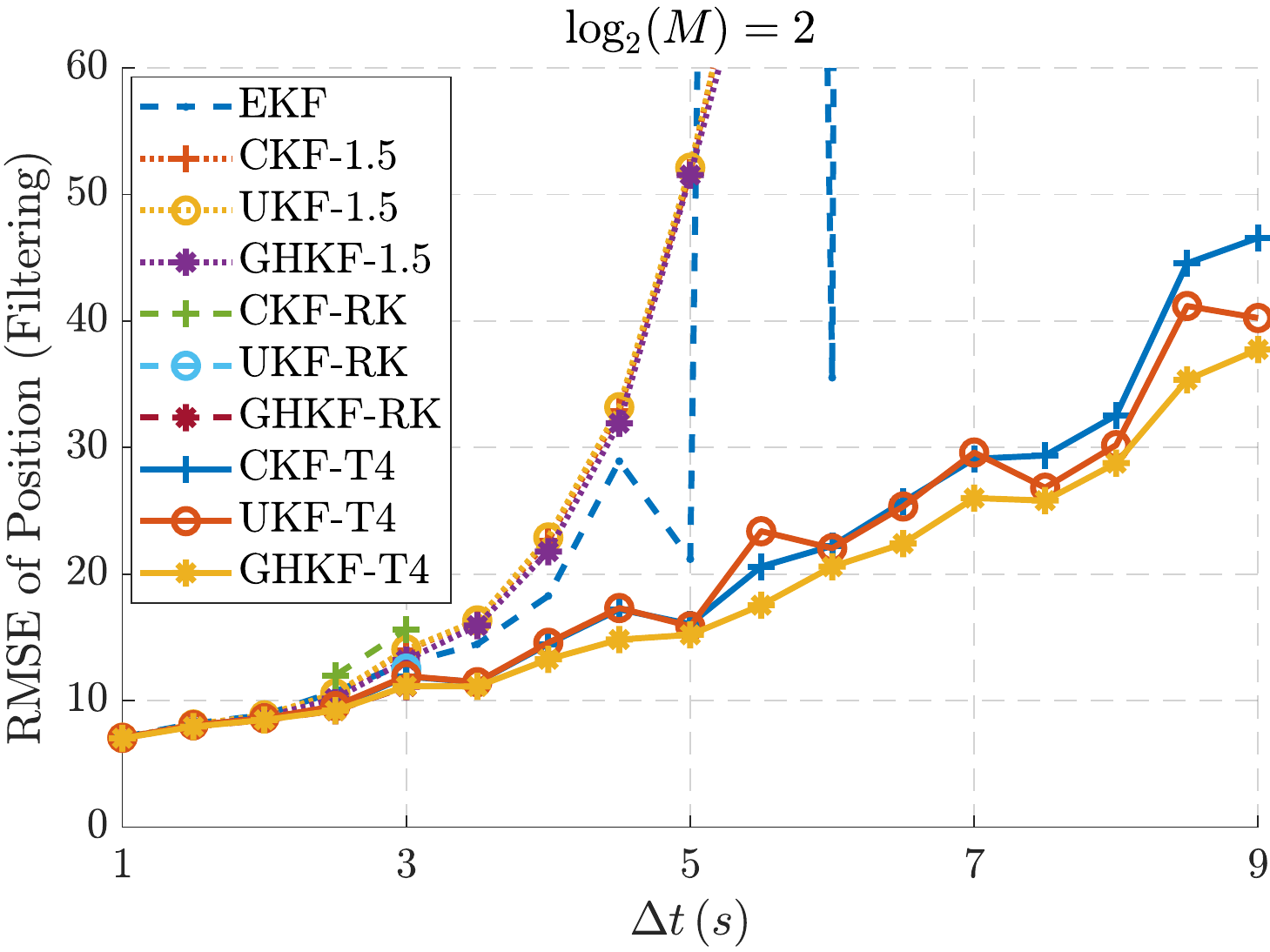}
	\includegraphics[width=0.19\linewidth]{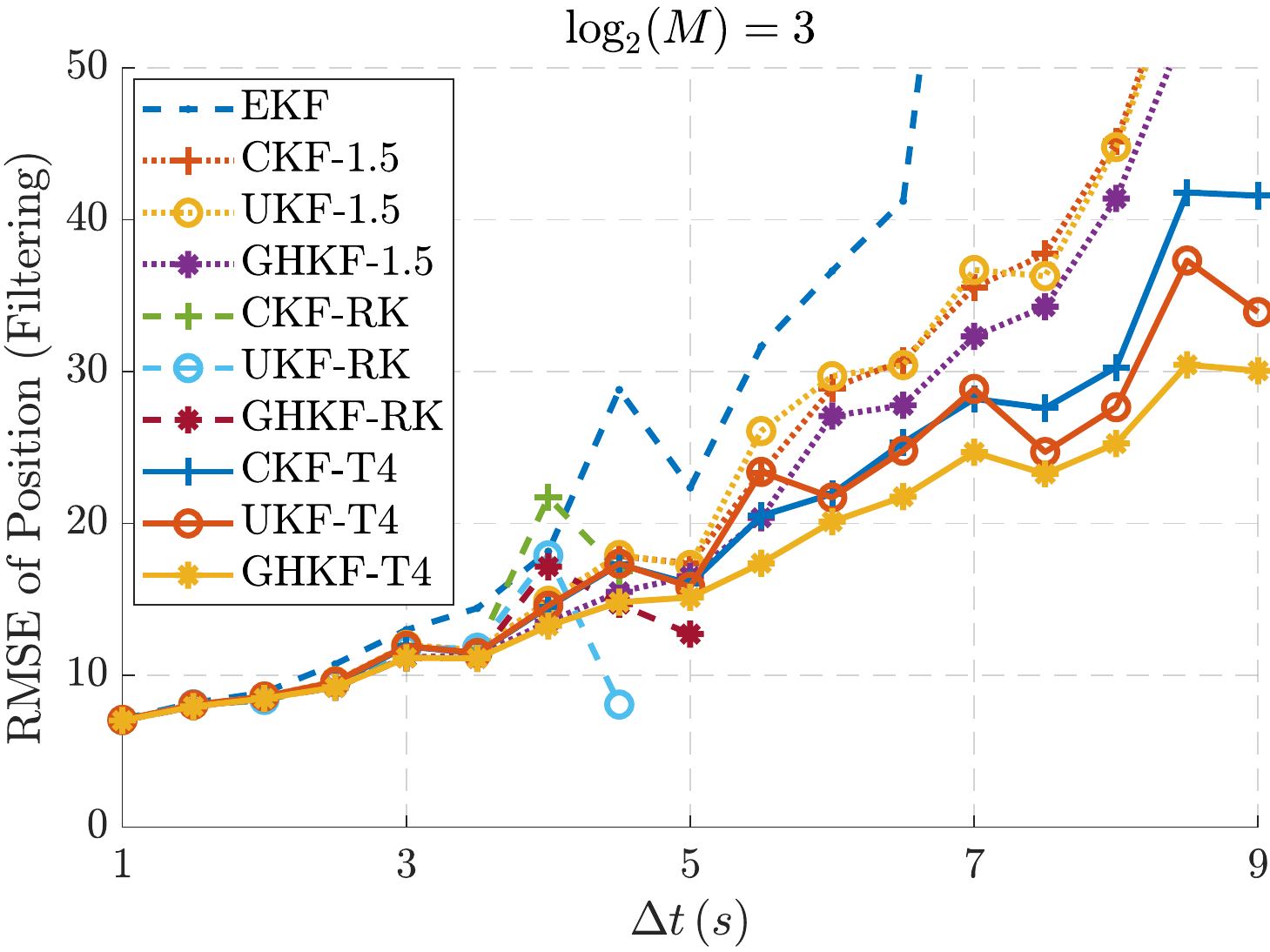}
	\includegraphics[width=0.19\linewidth]{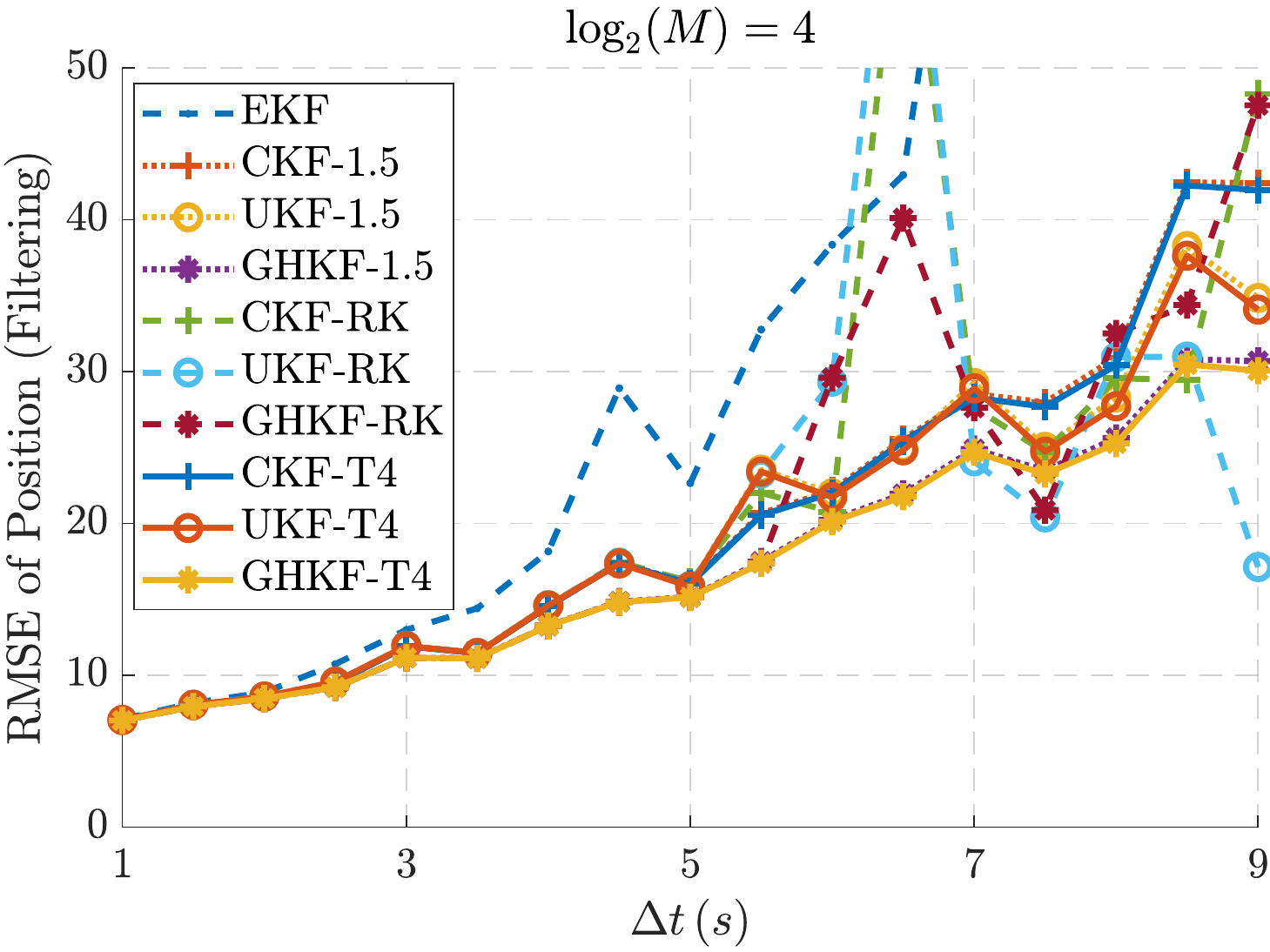}
	\includegraphics[width=0.19\linewidth]{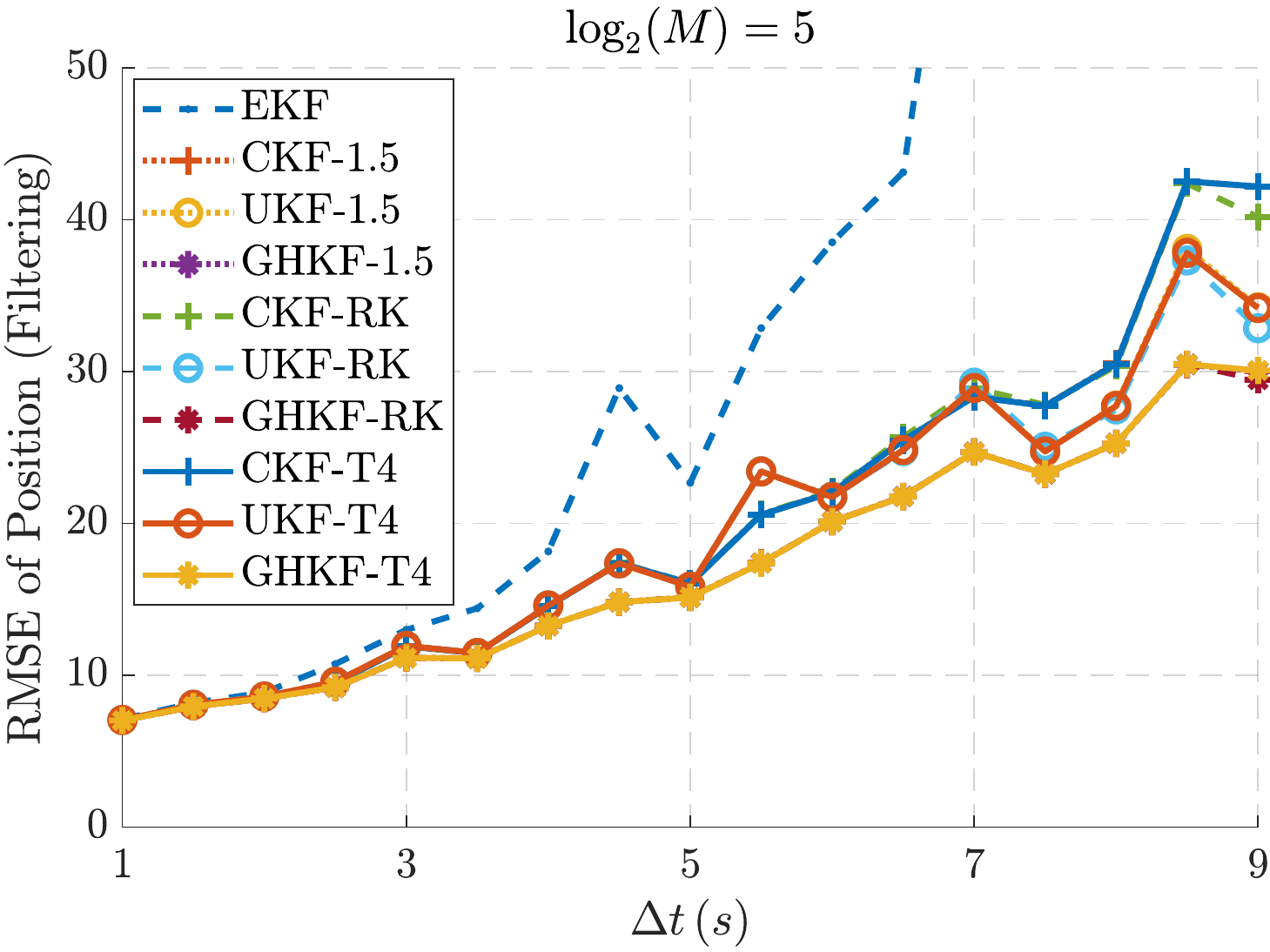}
	\caption{The RMSE of filtering results over 100 independent Monte Carlo runnings. The $M$ and $\Delta t$ are the number of integration steps and the time interval, respectively. From top to bottom rows, we control the expansion order of TME filters to be 2, 3, and 4. From left to right columns, we control the number of integration steps $M$. }
	\label{fig:rmse-filter-all}
\end{figure*}

For the comparison of the filtering and smoothing accuracy, we run 100 independent Monte Carlo trials of the selected methods and calculate the root mean square error (RMSE) of the position states, which is defined as
\begin{equation}
\mathrm{RMSE} = \sqrt{\sum^{100}_i\sum^{T/\Delta t}_j\sum_{k=1, 3, 5}\frac{( \cu{x}^{i,j}_{k} - \hat{\cu{x}}^{i,j}_{k})^2}{3\times 100 \times (T/ \Delta t)}},
\label{equ:rmse-def}
\end{equation}
where $\cu{x}^{i,j}_1$, $\cu{x}^{i,j}_3$, and $\cu{x}^{i,j}_5$ are the ground truth positions at time $j\Delta t$ from $i$-th Monte Carlo trial, and the $\hat{\cu{x}}^{i,j}_1$, $\hat{\cu{x}}^{i,j}_3$, and $\hat{\cu{x}}^{i,j}_5$ are their corresponding (filtering or smoothing) estimates. For the comparison of numerical stability, we record the number of divergences, which is defined by counting the manifestation of non-positive definite covariances and unbounded (NaN) estimates. The state-of-the-art Gaussian filters and smoothers for comparison are listed in Table~\ref{tbl:list-of-methods}. 

\begin{table}[h!]
	\centering\scriptsize
	\begin{tabular}{lccc}
		\toprule
		\textbf{Name} & \textbf{Description} & \textbf{Method} & \textbf{Integration} \\ \midrule
		EKF/S-RK & \multirow{4}{*}{\begin{tabular}[c]{@{}c@{}}ODE type of\\ Gaussian \\ filter and \\ smoother \\ with 4th order \\ Runge--Kutta \\ solver \cite{sarkkaarnobook2019, SARKKA2013500}\end{tabular}} & \begin{tabular}[c]{@{}c@{}}Solving \eqref{equ:dmdt_mPP} with\\ linearisation\\(Linear-ODE)\end{tabular} & Not needed \\ \cline{3-4} 
		CKF/S-RK &  & \multirow{3}{*}{\begin{tabular}[c]{@{}c@{}}\\Solving \eqref{equ:dmdt_mPP} with\\Gaussian\\ assumption\\(Gauss-ODE)\end{tabular}} & \begin{tabular}[c]{@{}c@{}}Spherical\\ cubature\end{tabular} \\ \cline{4-4} 
		UKF/S-RK &  &  & \begin{tabular}[c]{@{}c@{}}Unscented\\ transform\end{tabular} \\ \cline{4-4} 
		GHKF/S-RK &  &  & \begin{tabular}[c]{@{}c@{}}3rd order\\ Gauss--Hermite\end{tabular} \\ \midrule
		CKF/S-1.5 & \multirow{3}{*}{\begin{tabular}[c]{@{}c@{}}It\^{o}-1.5 \\ Gaussian \\ filter and \\ smoother \\\cite{ckdIndian2010, SARKKA20121221}\end{tabular}} & \multirow{3}{*}{\begin{tabular}[c]{@{}c@{}}It\^{o}--Taylor\\ discretisation \\ with strong\\ order 1.5~\cite{ckdIndian2010}   \end{tabular}} & \begin{tabular}[c]{@{}c@{}}Spherical\\ cubature\end{tabular} \\ \cline{4-4} 
		UKF/S-1.5 &  &  & \begin{tabular}[c]{@{}c@{}}Unscented\\ transform\end{tabular} \\ \cline{4-4} 
		GHKF/S-1.5 &  &  & \begin{tabular}[c]{@{}c@{}}3rd order\\ Gauss--Hermite\end{tabular} \\ \midrule
		CKF/S-T* & \multirow{3}{*}{\begin{tabular}[c]{@{}c@{}}TME \\ Gaussian\\ filter and \\ smoother \\ (Alg. \ref{alg:tme-filter} and \ref{alg:tme-smoother})\end{tabular}} & \multirow{3}{*}{\begin{tabular}[c]{@{}c@{}}\\$*$-th order \\ TME \\ (Alg. \ref{def:taylor-moment-estimator})\end{tabular}} & \begin{tabular}[c]{@{}c@{}}Spherical\\ cubature\end{tabular} \\ \cline{4-4} 
		UKF/S-T* &  &  & \begin{tabular}[c]{@{}c@{}}Unscented\\ transform\end{tabular} \\ \cline{4-4} 
		GHKF/S-T* &  &  & \begin{tabular}[c]{@{}c@{}}3rd order\\ Gauss--Hermite\end{tabular}\\ \bottomrule
	\end{tabular}
	\caption{The list of the state-of-the-art methods compared in 3D coordinate turn tracking. For the ODE type of Gaussian smoothers, we use the classical type I smoothers as described in \cite{SARKKA2013500}. }
	\label{tbl:list-of-methods}
\end{table}

\begin{figure*}[t!]
	\centering
	\includegraphics[width=0.19\linewidth]{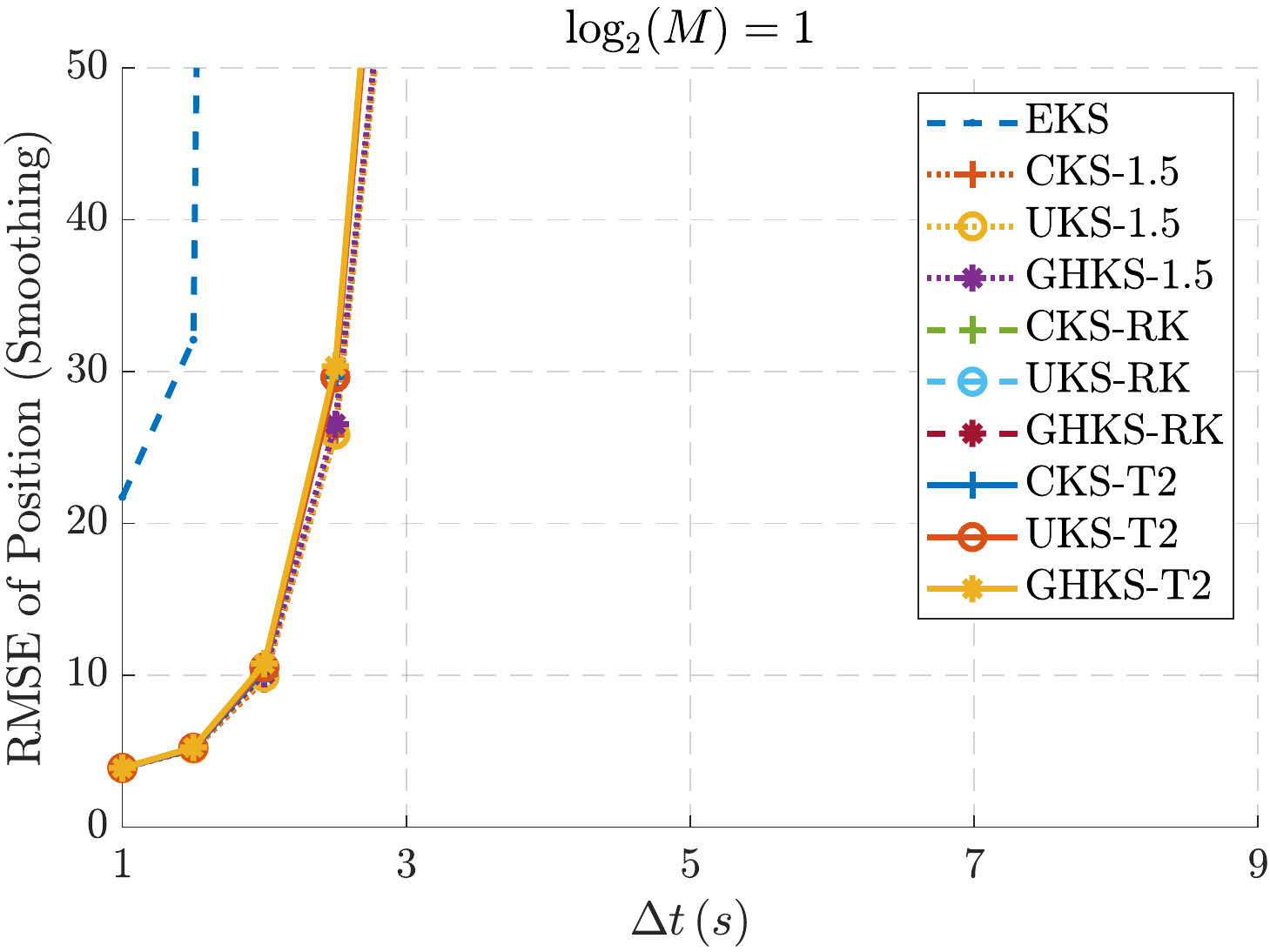}
	\includegraphics[width=0.19\linewidth]{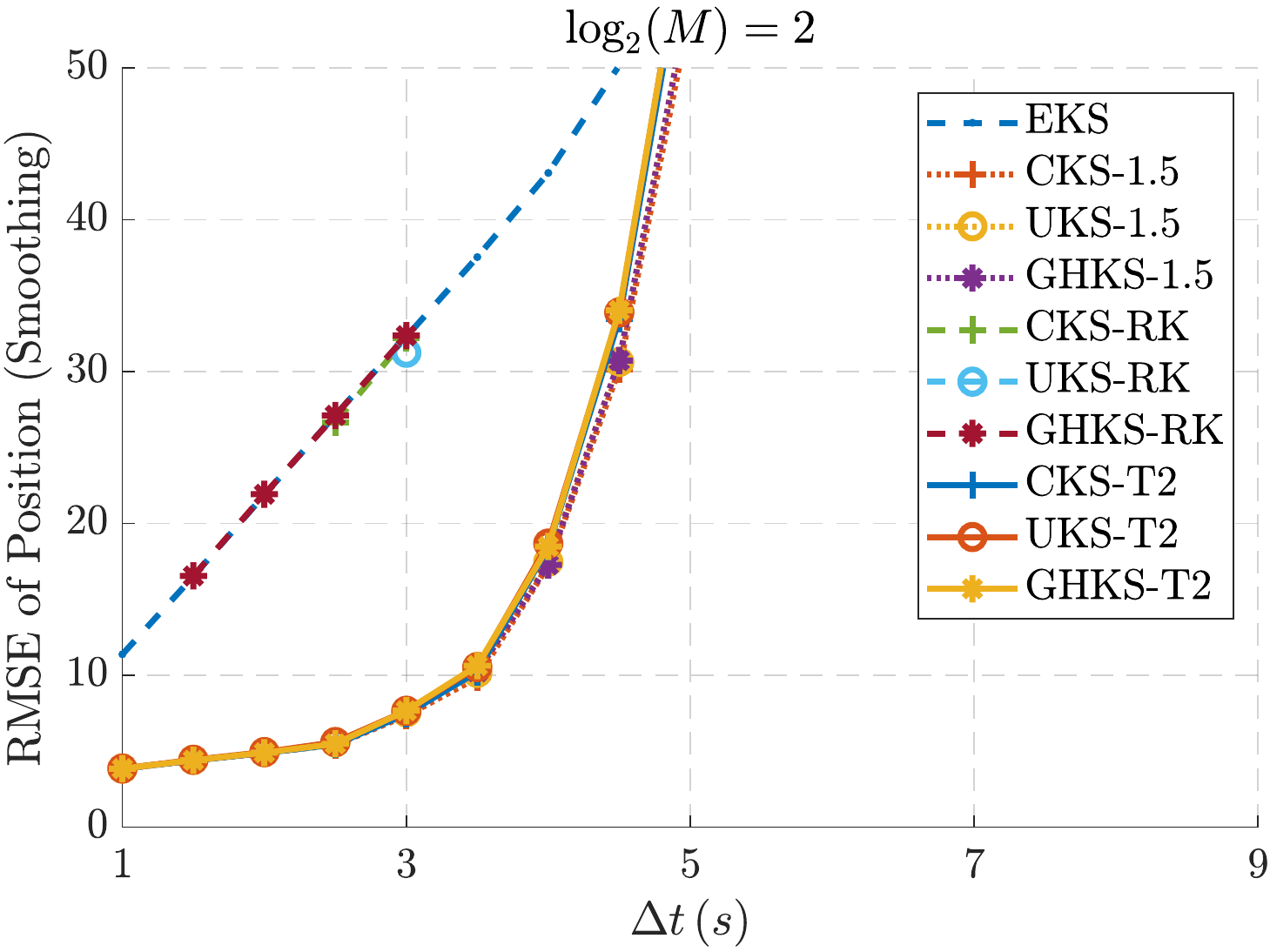}
	\includegraphics[width=0.19\linewidth]{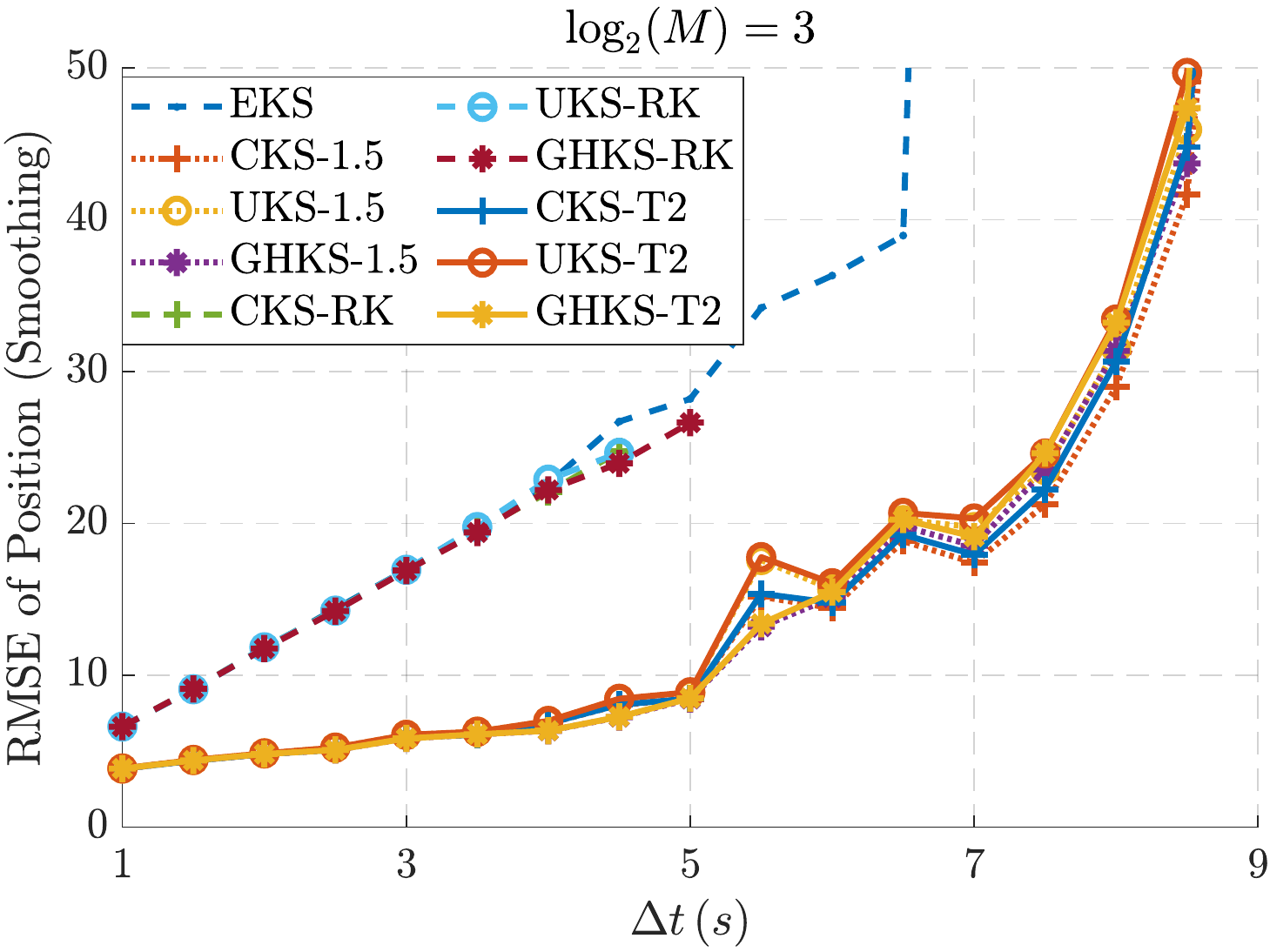}
	\includegraphics[width=0.19\linewidth]{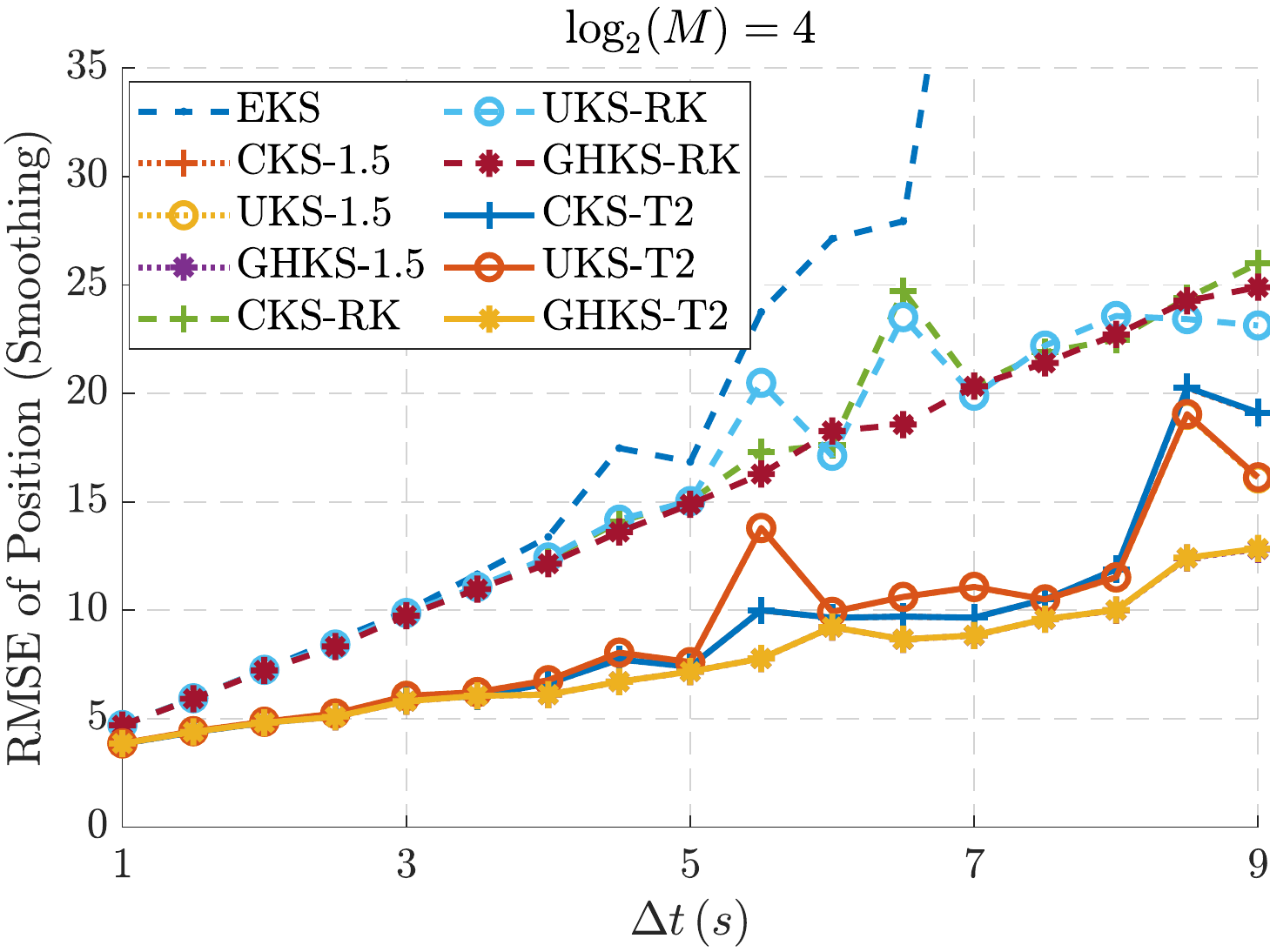}
	\includegraphics[width=0.19\linewidth]{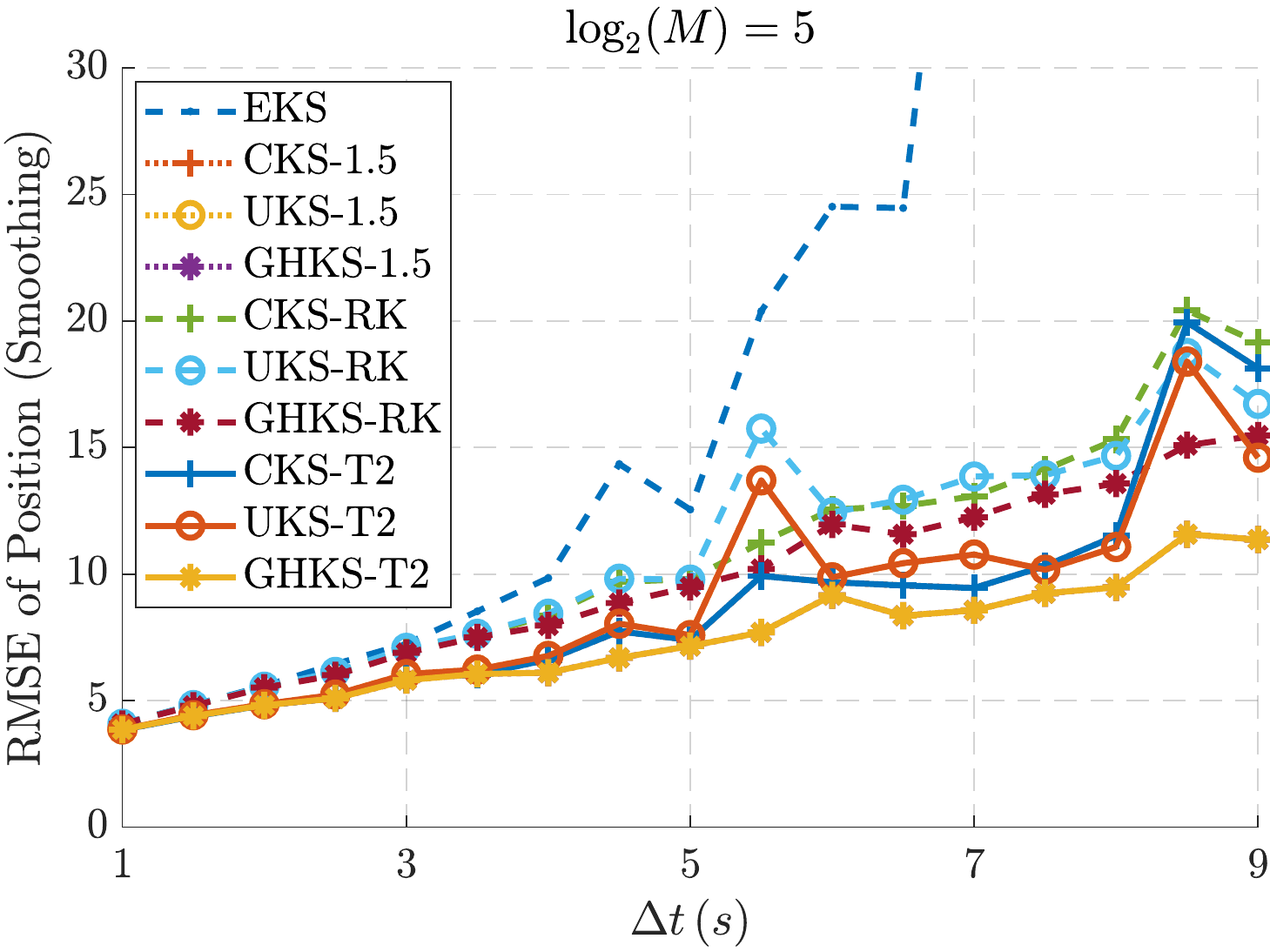} \\
	\includegraphics[width=0.19\linewidth]{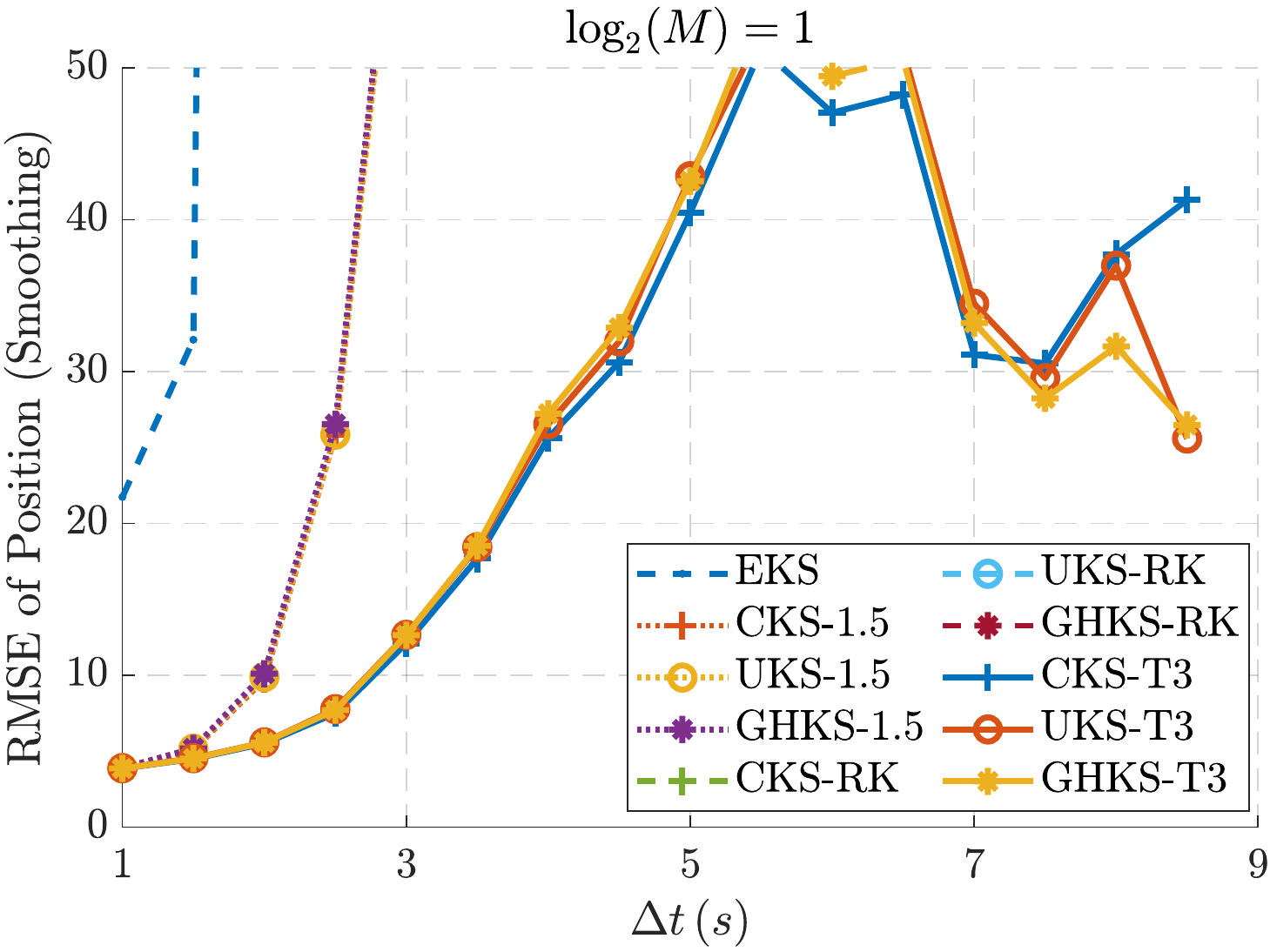}
	\includegraphics[width=0.19\linewidth]{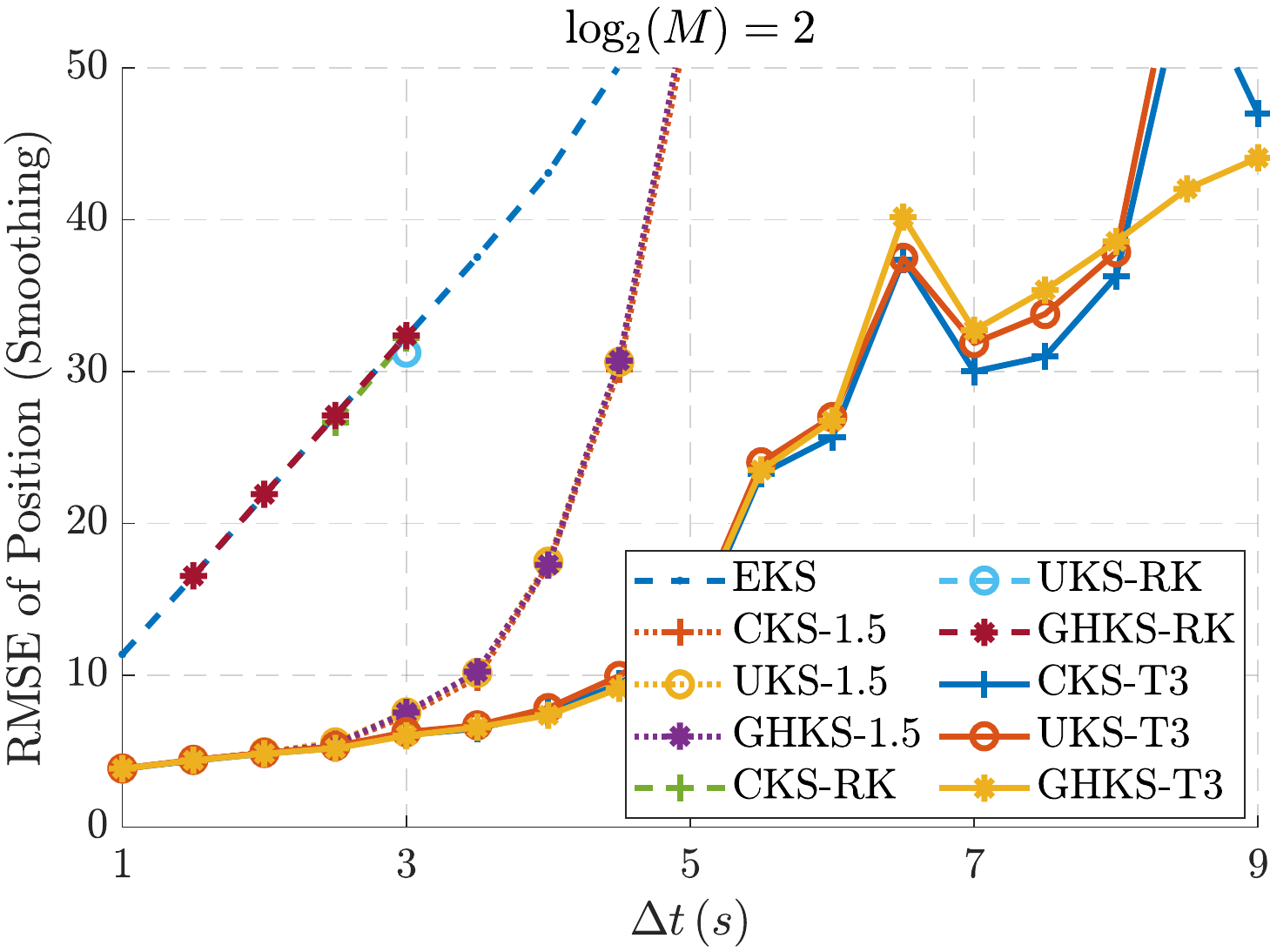}
	\includegraphics[width=0.19\linewidth]{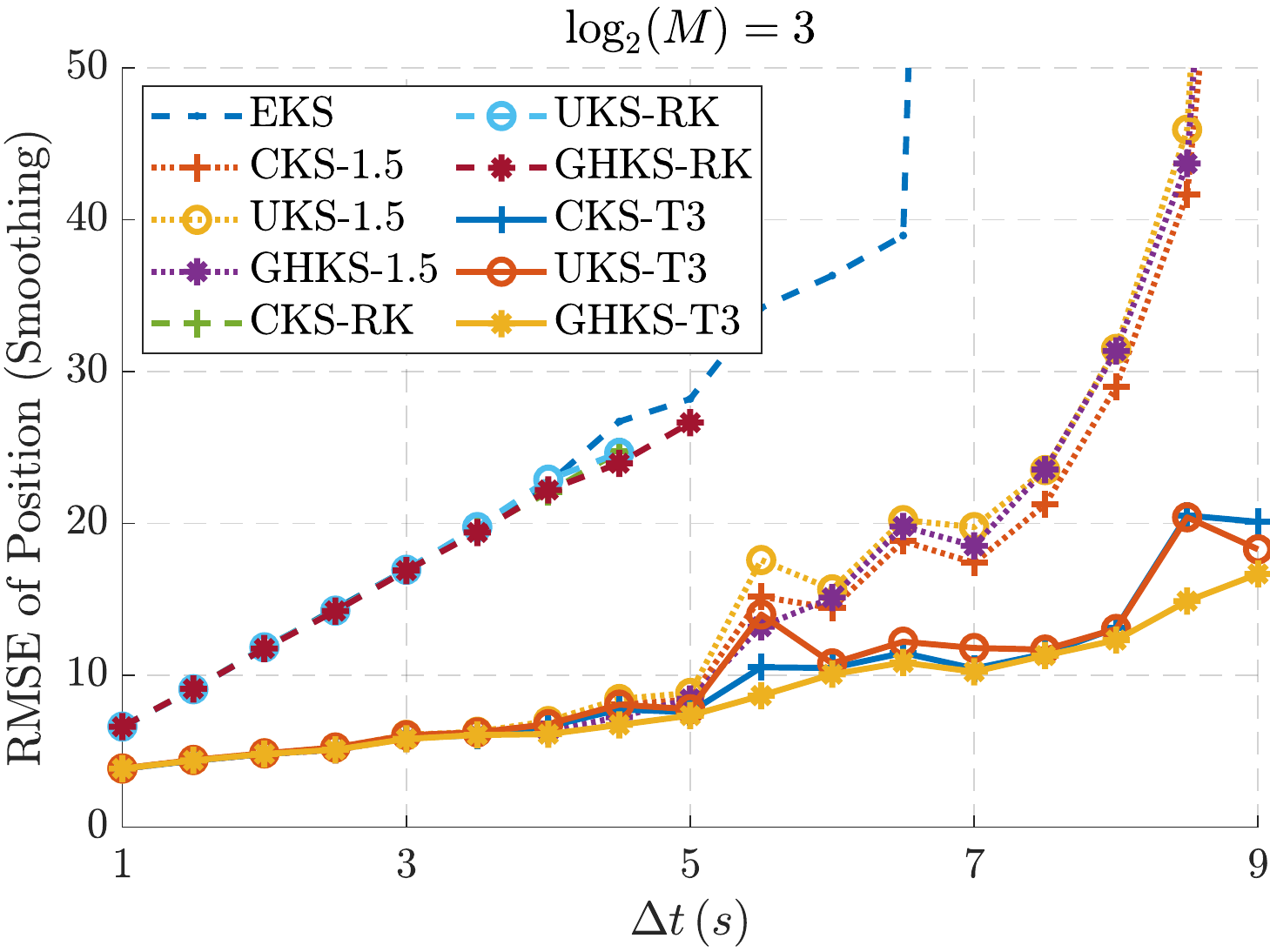}
	\includegraphics[width=0.19\linewidth]{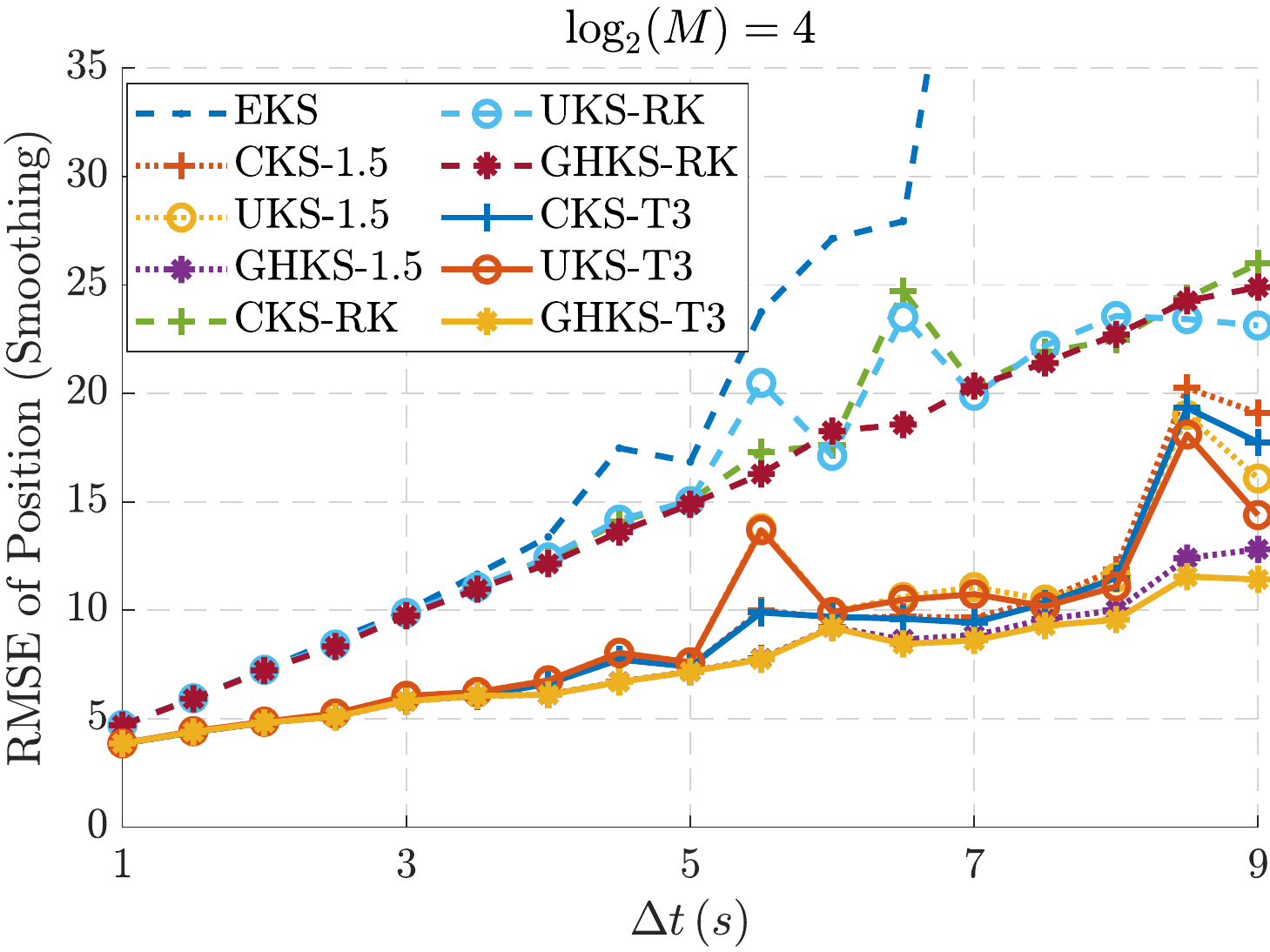}
	\includegraphics[width=0.19\linewidth]{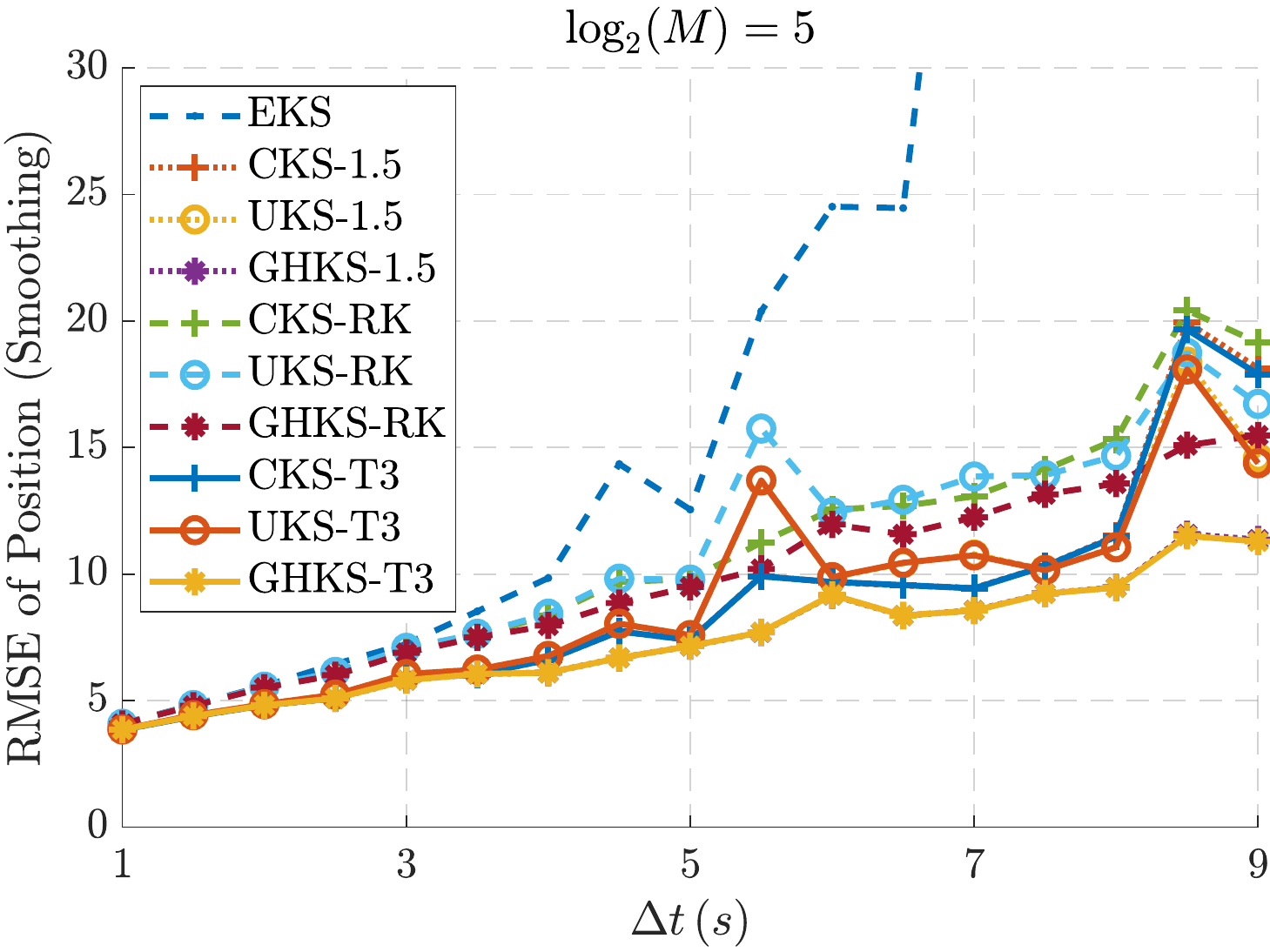} \\
	\includegraphics[width=0.19\linewidth]{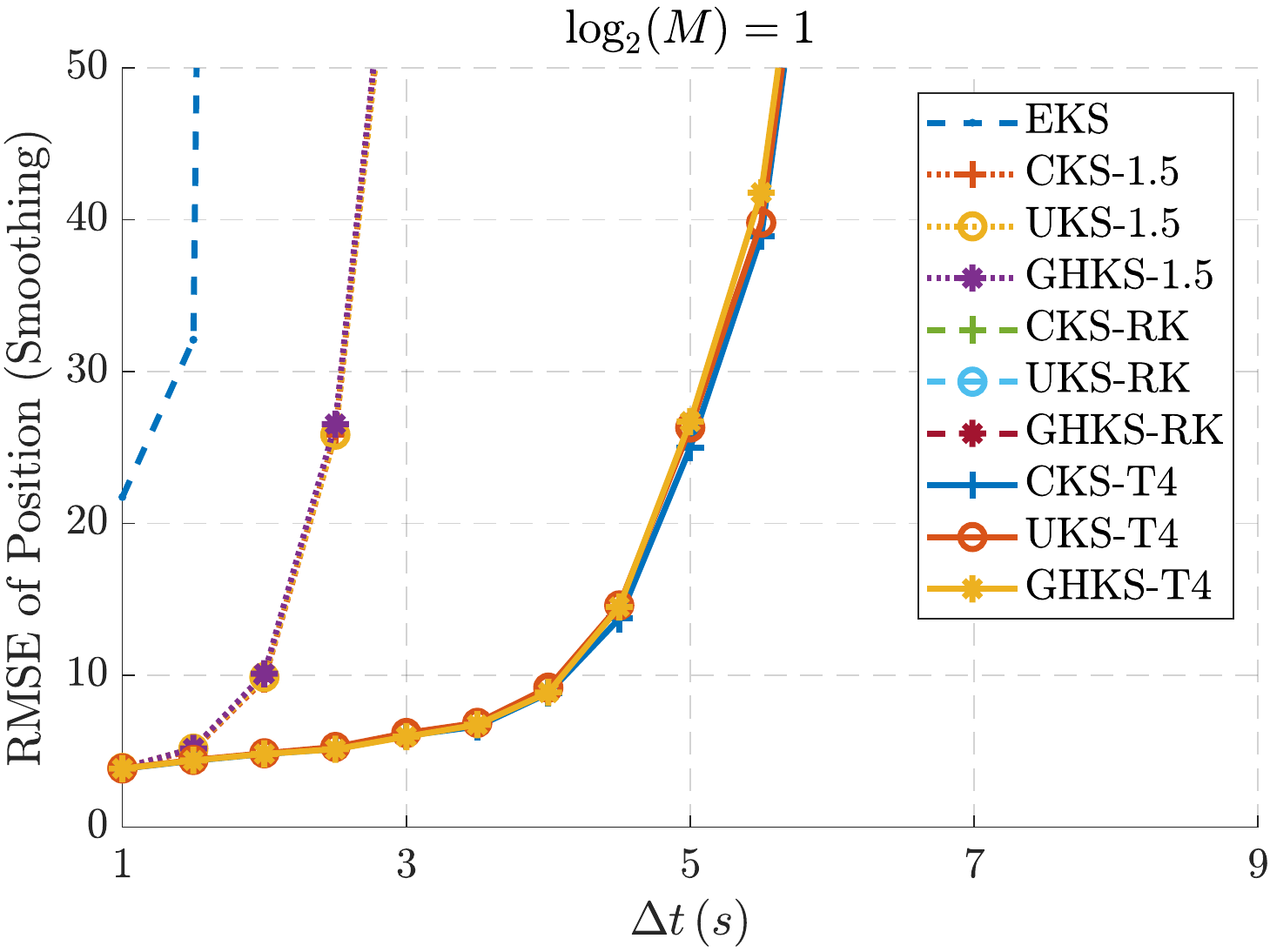}
	\includegraphics[width=0.19\linewidth]{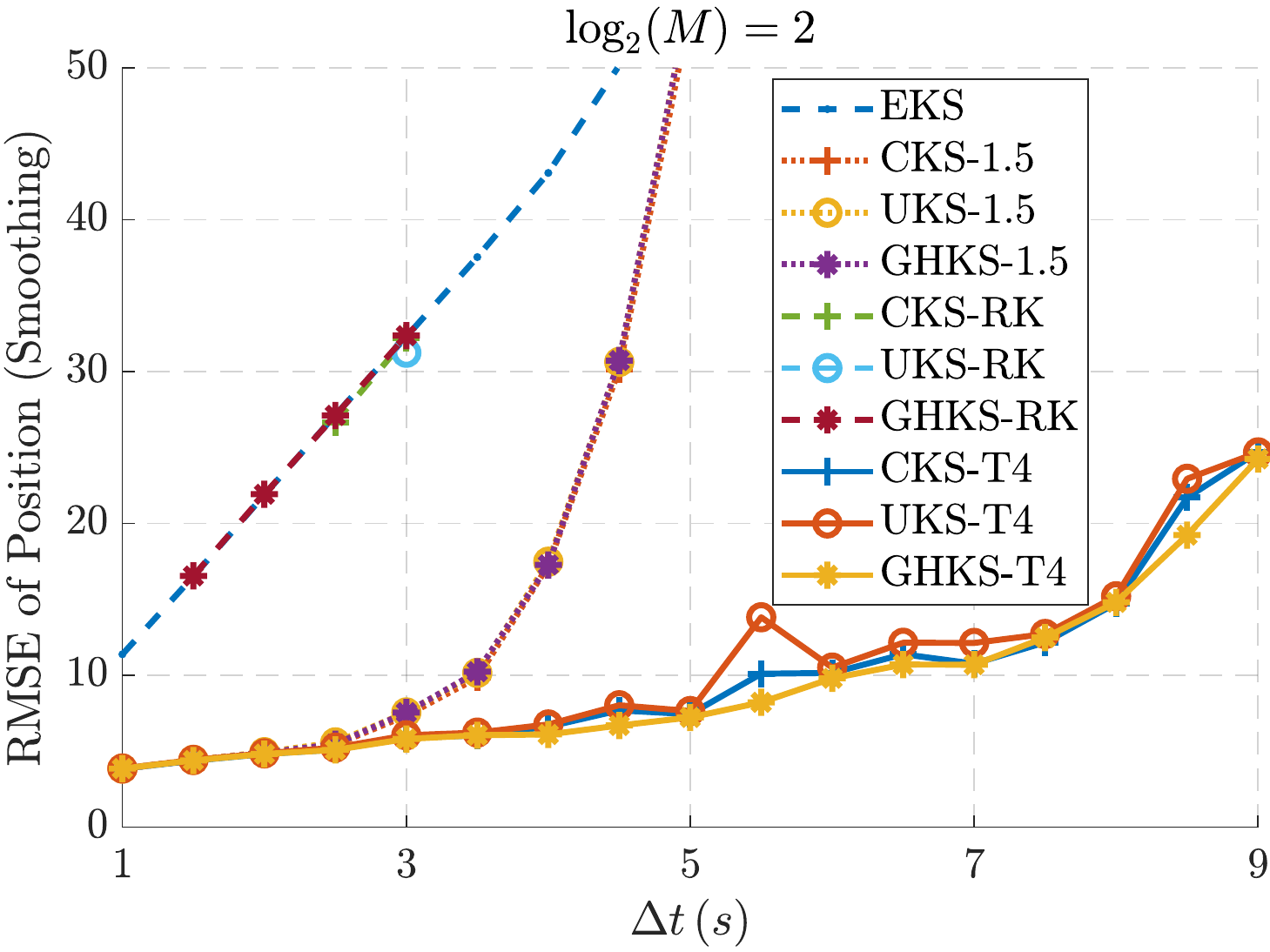}
	\includegraphics[width=0.19\linewidth]{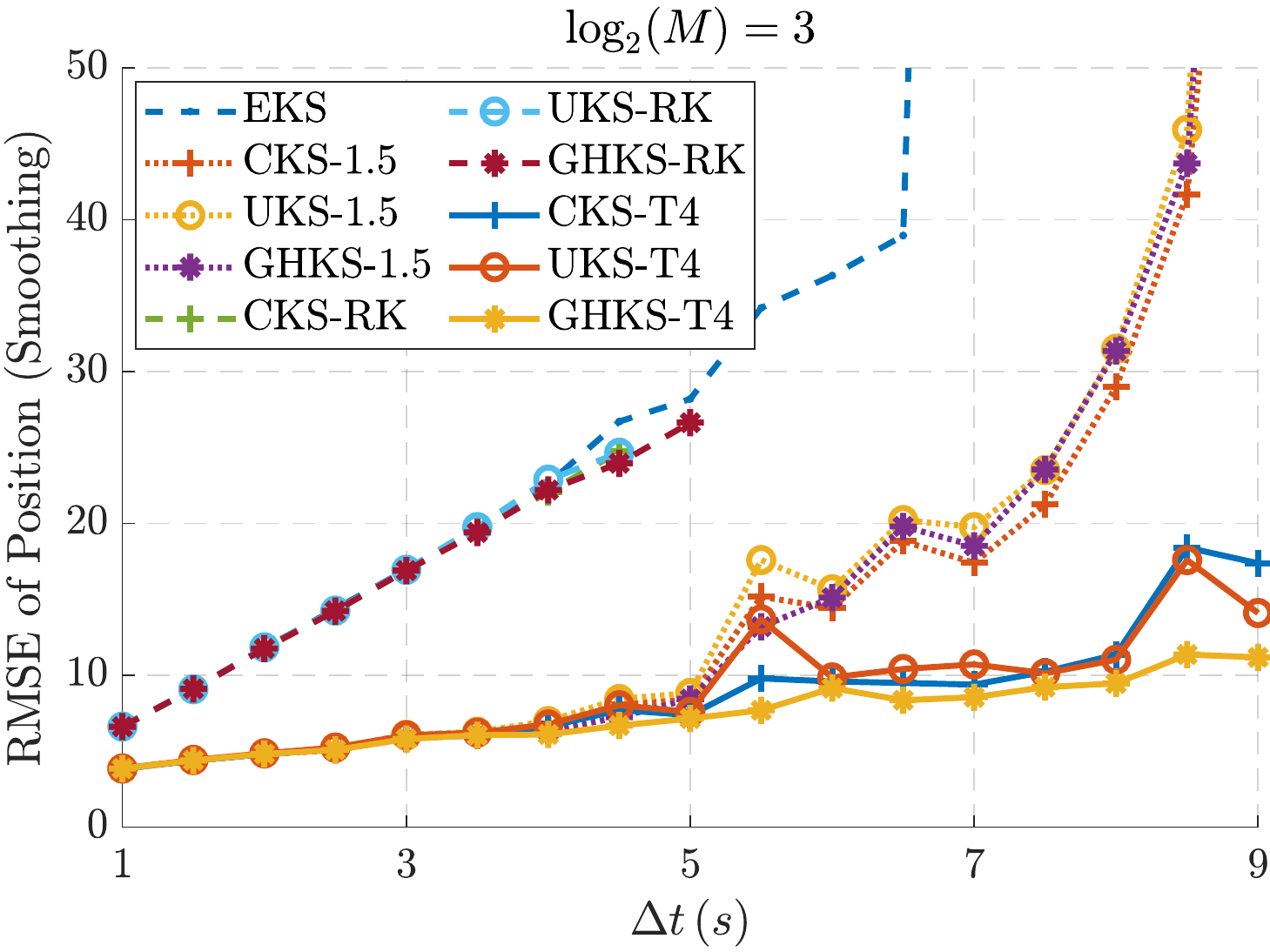}
	\includegraphics[width=0.19\linewidth]{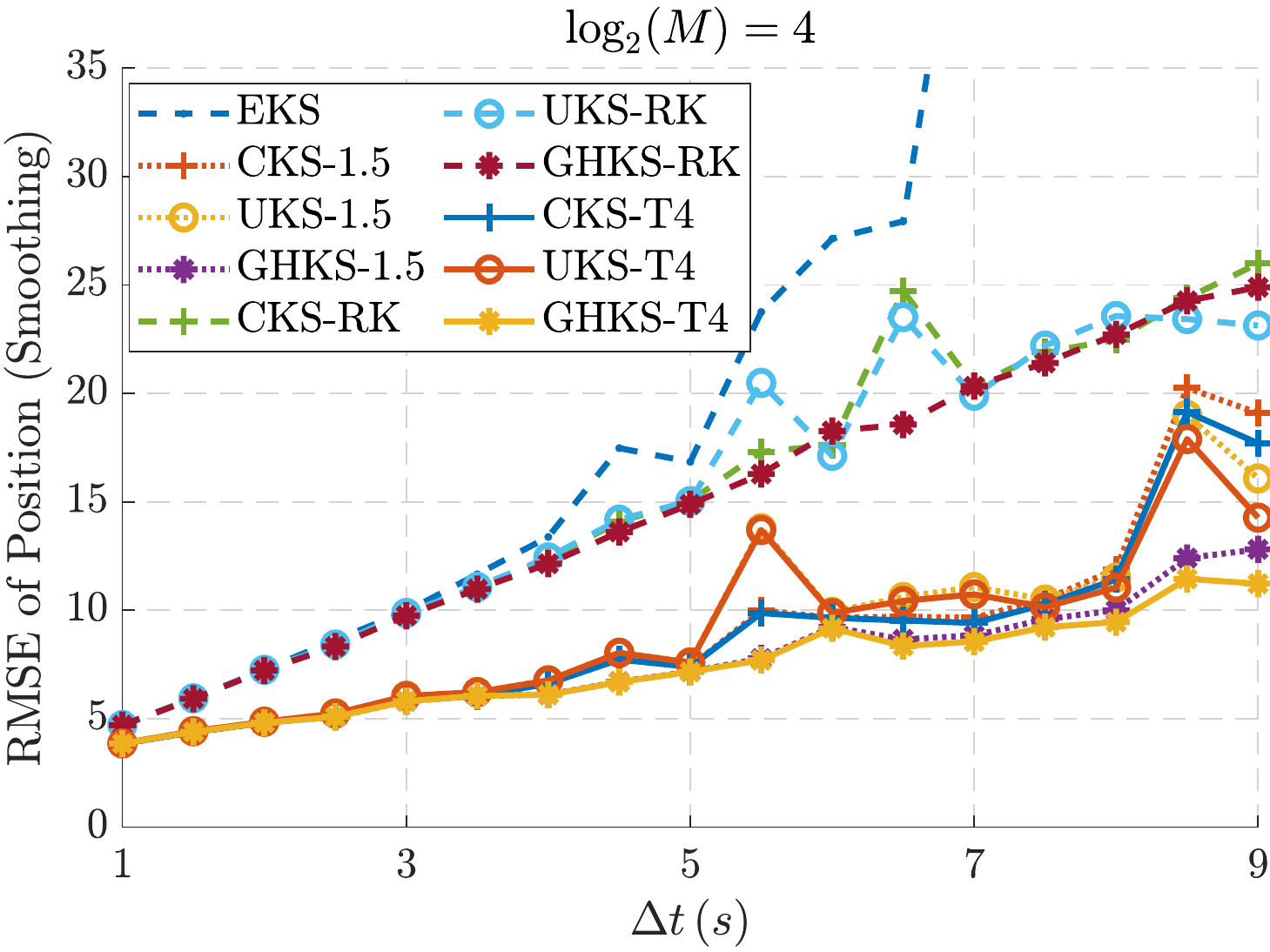}
	\includegraphics[width=0.19\linewidth]{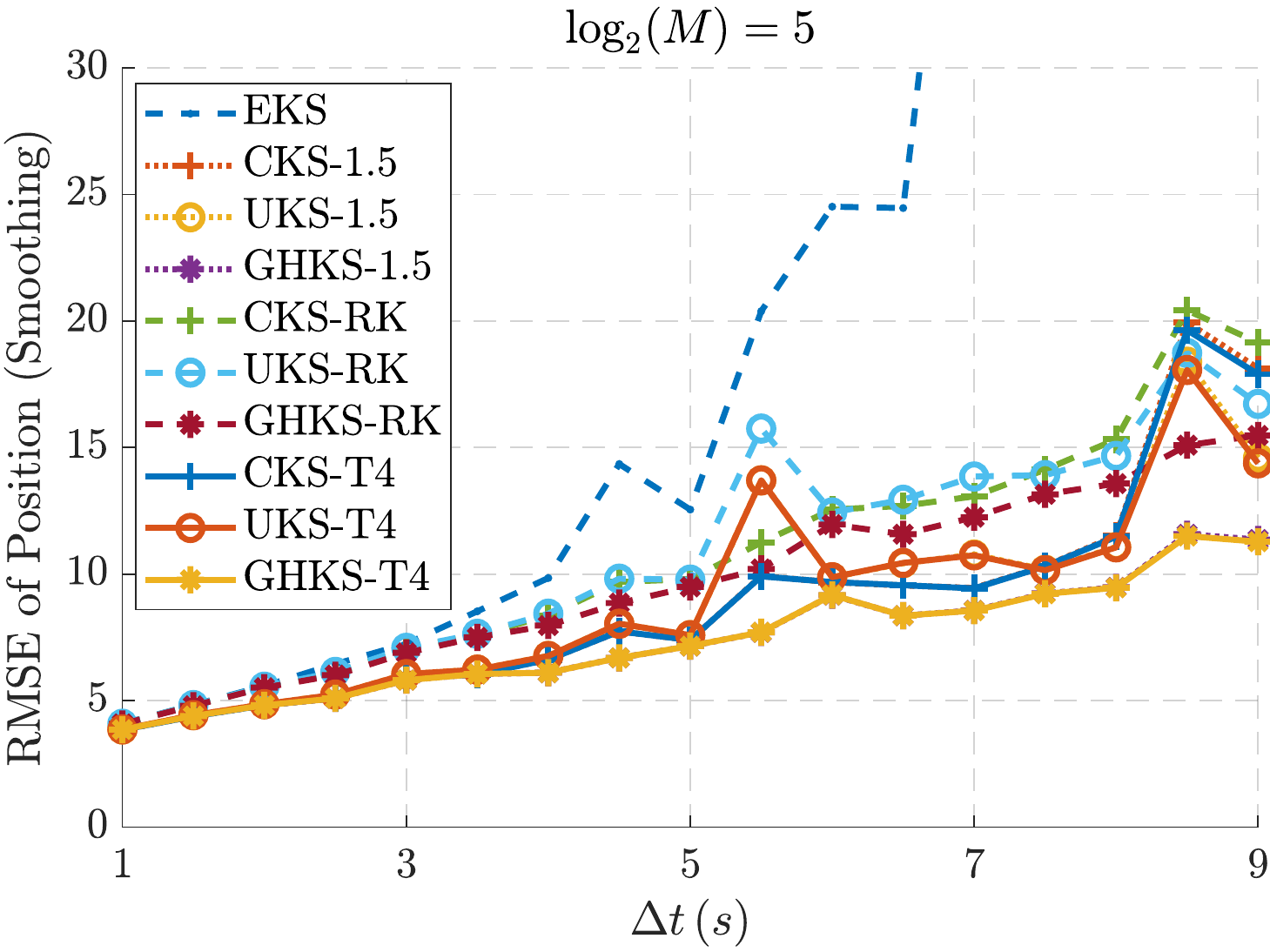}
	\caption{The RMSE of smoothing results over 100 independent Monte Carlo runnings. The $M$ and $\Delta t$ are the number of integration steps and the time interval, respectively. From top to bottom rows, we control the expansion order of TME smoothers to be 2, 3, and 4. From left to right columns, we control the number of integration steps $M$. }
	\label{fig:rmse-smooth-all}
\end{figure*}

\begin{figure*}[t!]
	\centering
	\includegraphics[width=.3\linewidth]{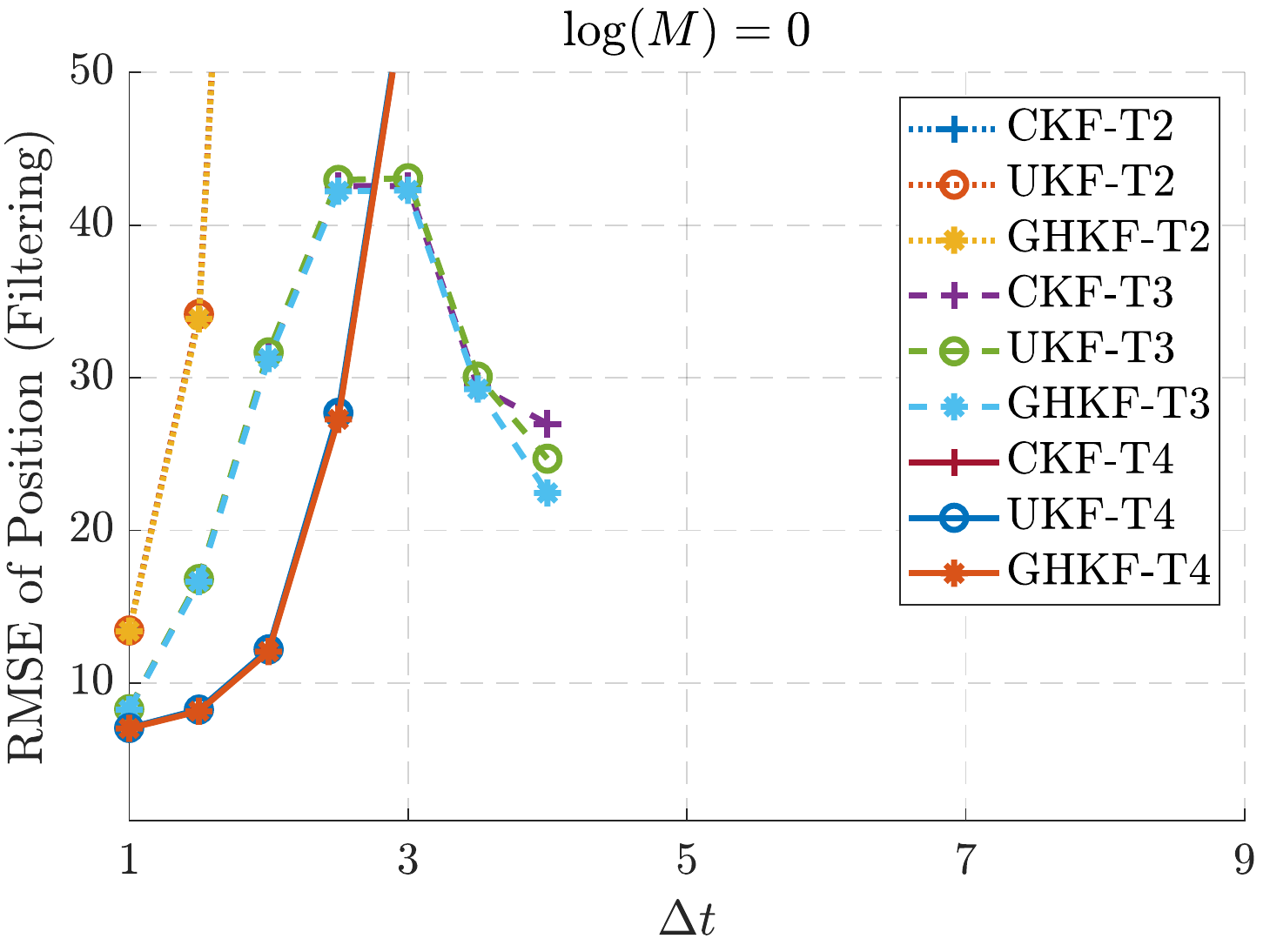}
	\includegraphics[width=.3\linewidth]{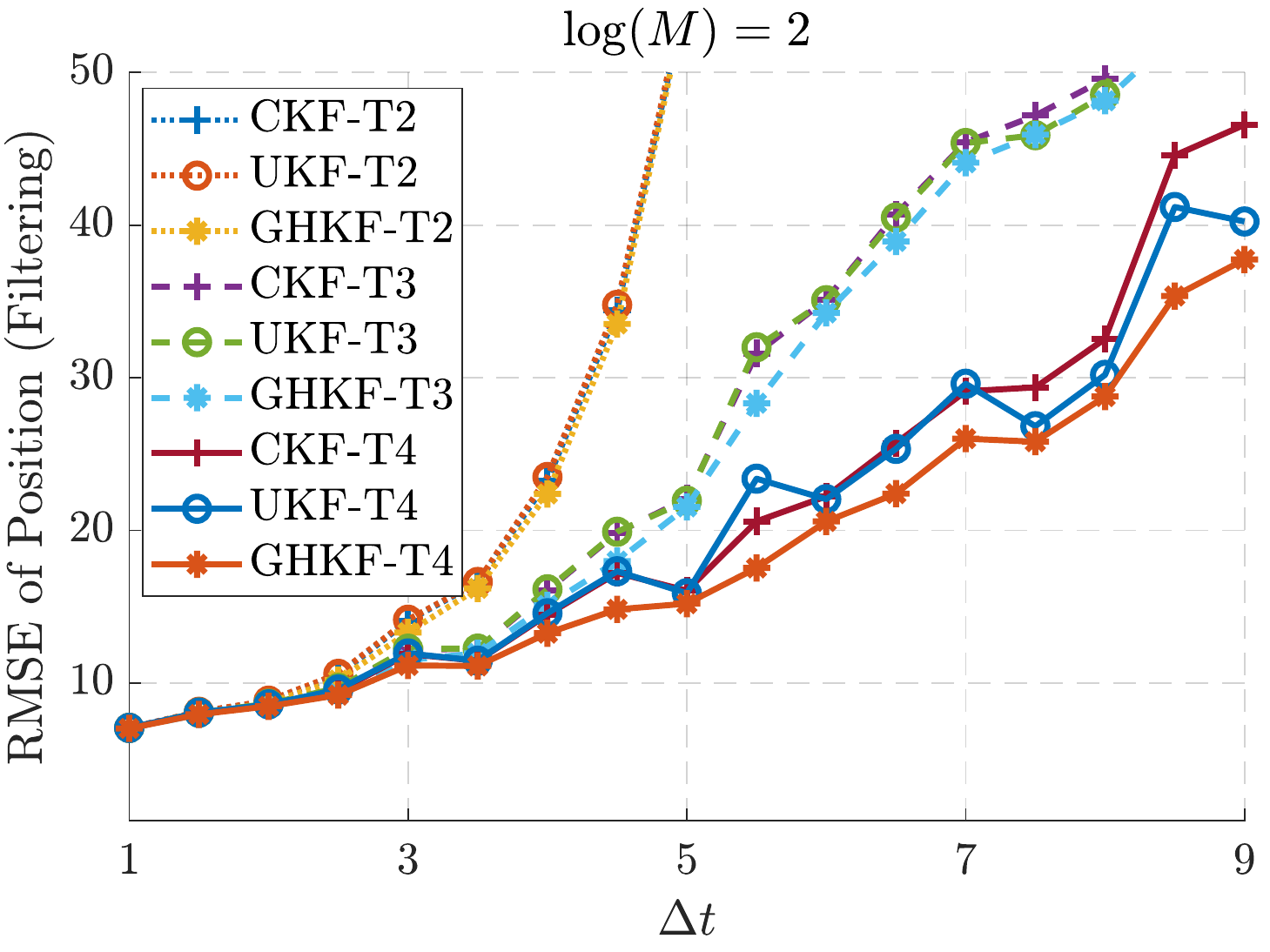}
	\includegraphics[width=.3\linewidth]{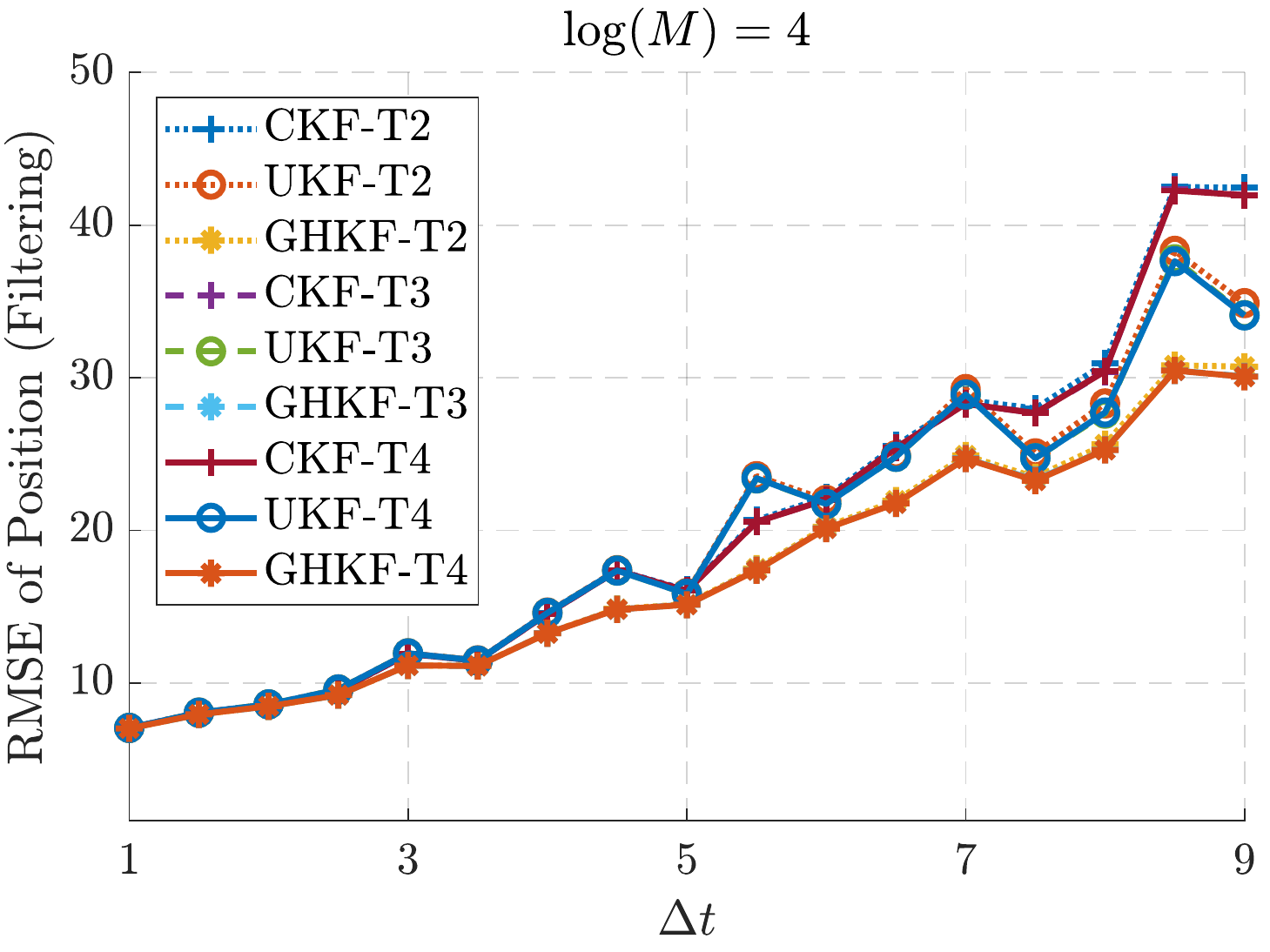}\\
	\includegraphics[width=.3\linewidth]{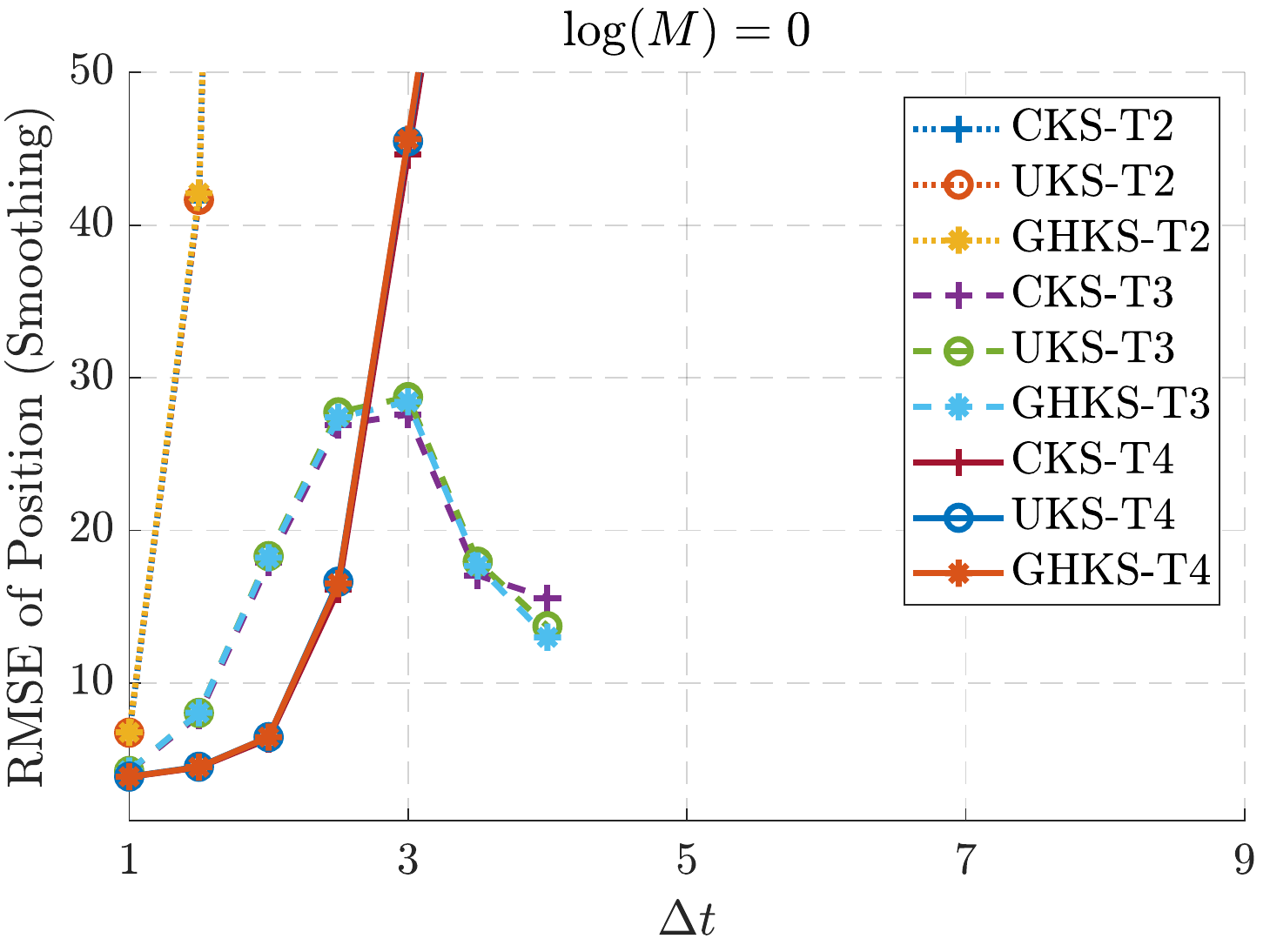}
	\includegraphics[width=.3\linewidth]{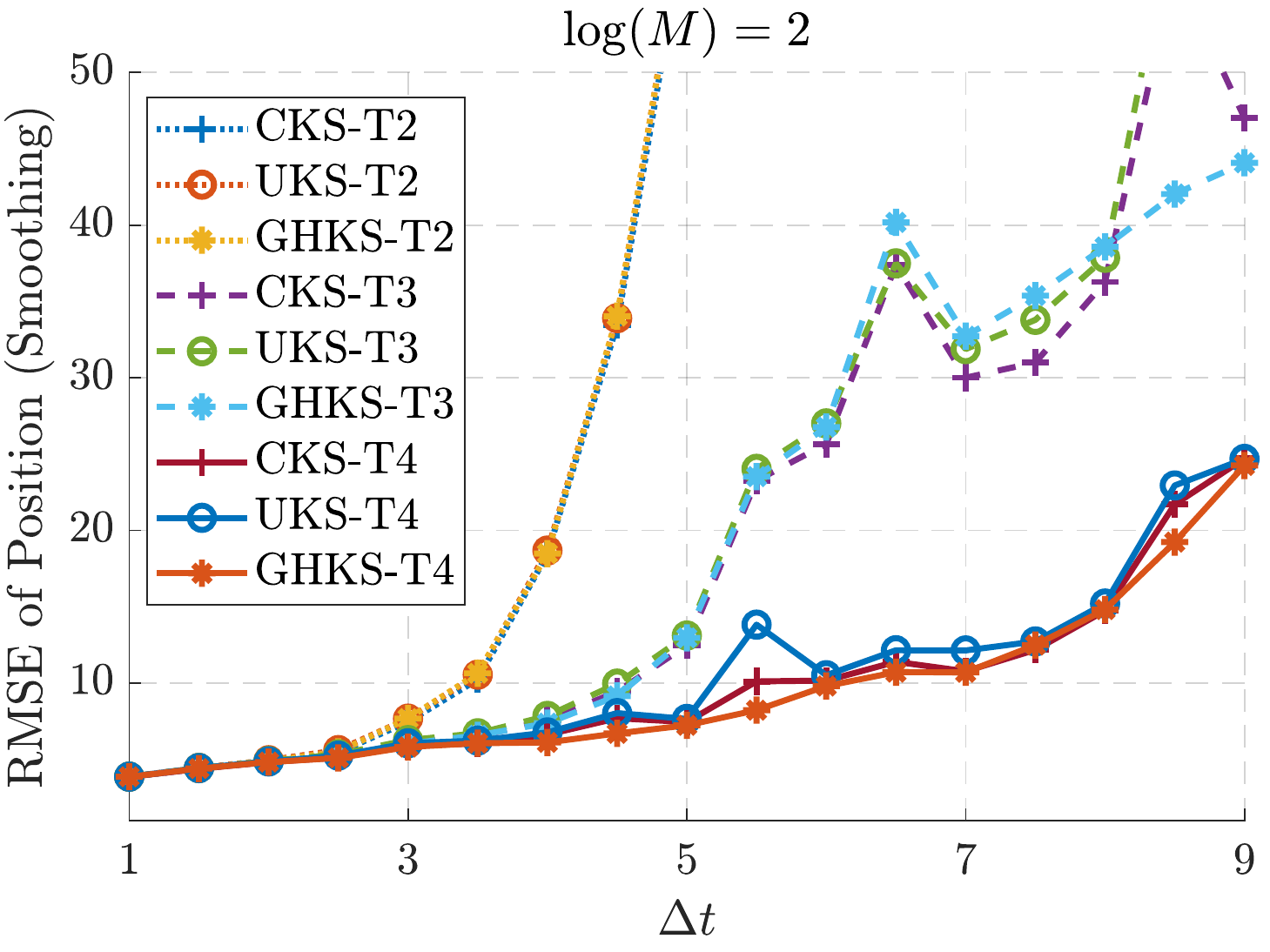}
	\includegraphics[width=.3\linewidth]{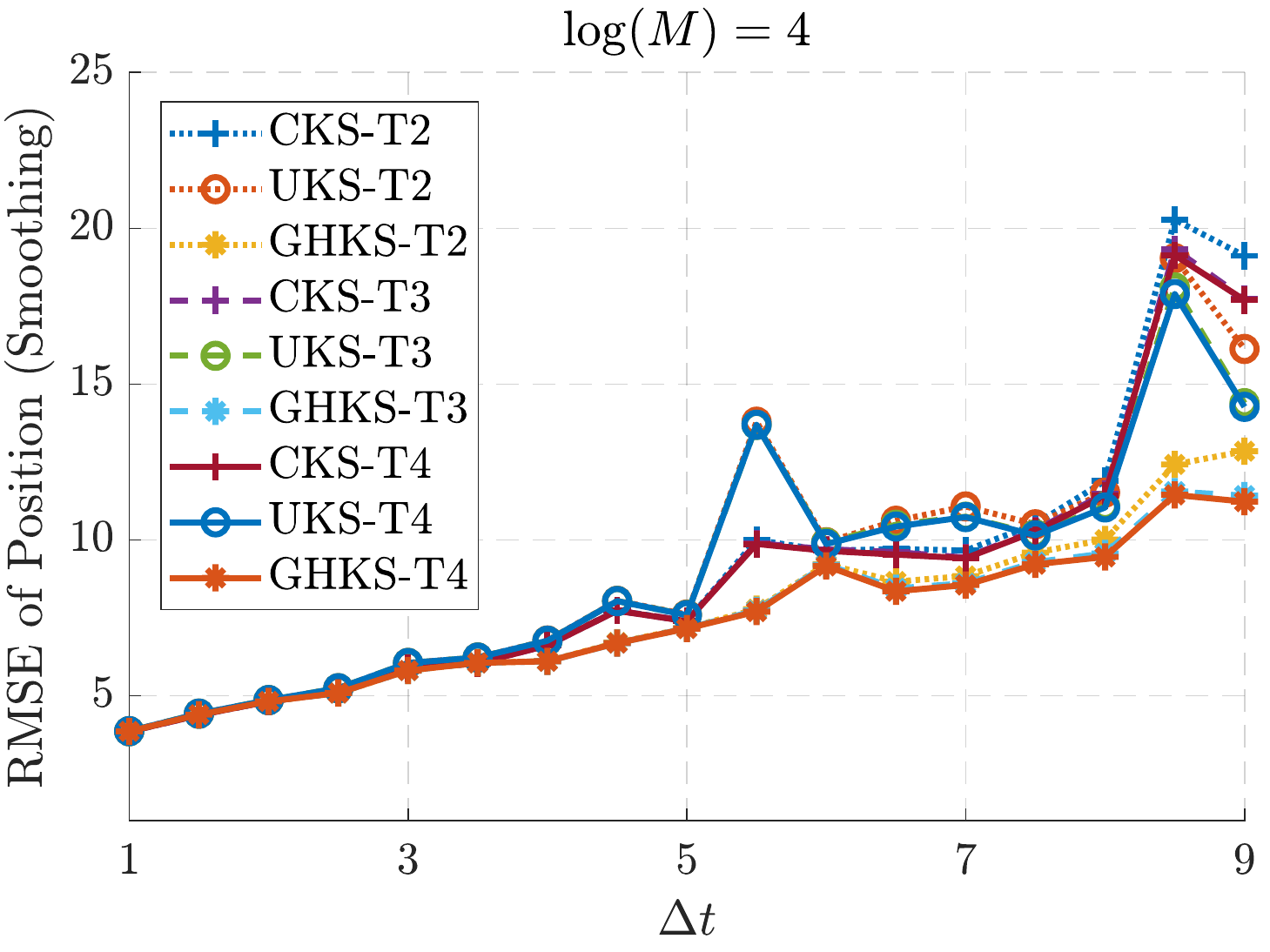}
	\caption{Comparison of TME filters and smoothers with different expansion orders. }
	\label{fig:rmse-TME-order}
\end{figure*}
\begin{figure*}[t!]
	\centering
	\includegraphics[width=.328\linewidth]{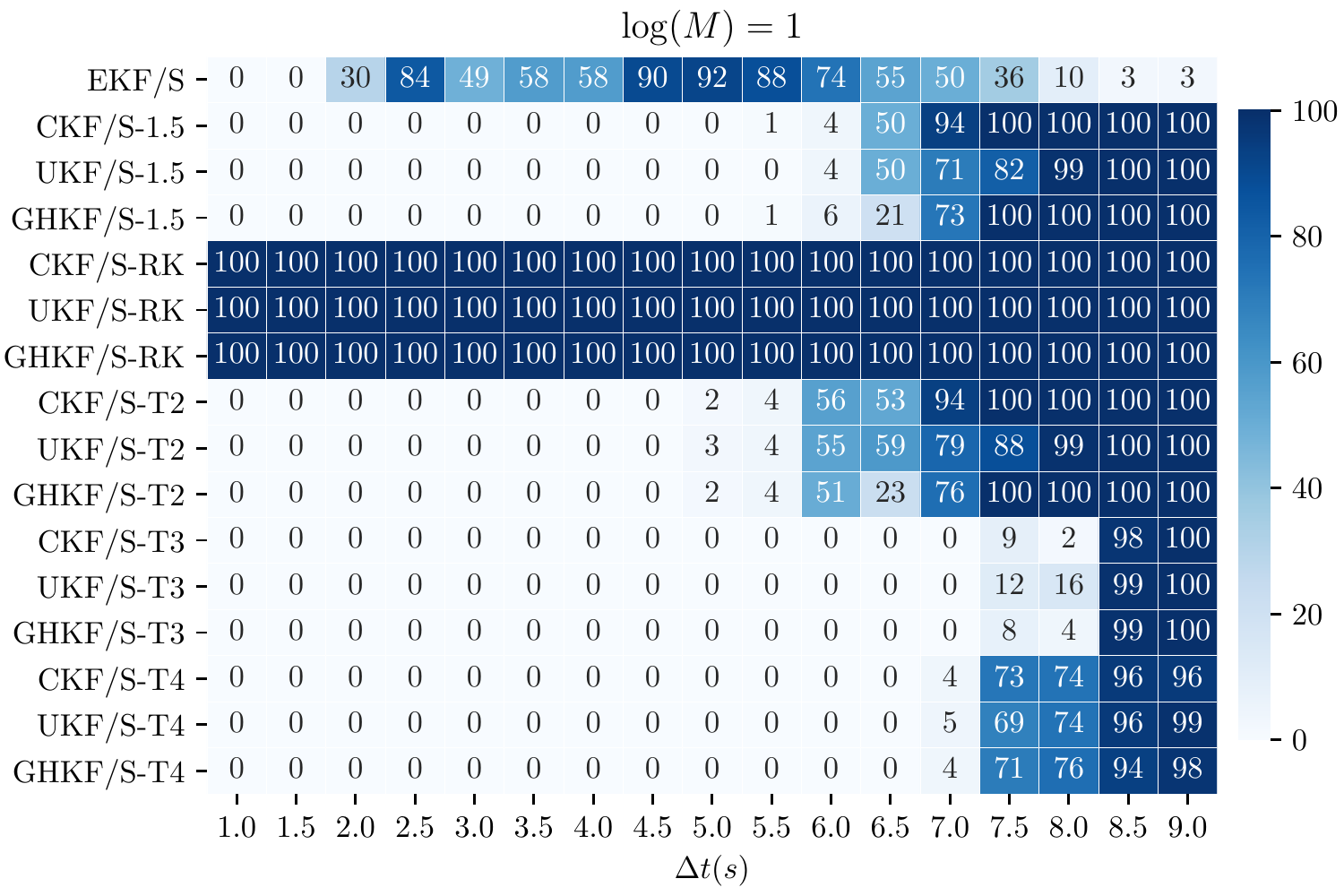}
	\includegraphics[width=.328\linewidth]{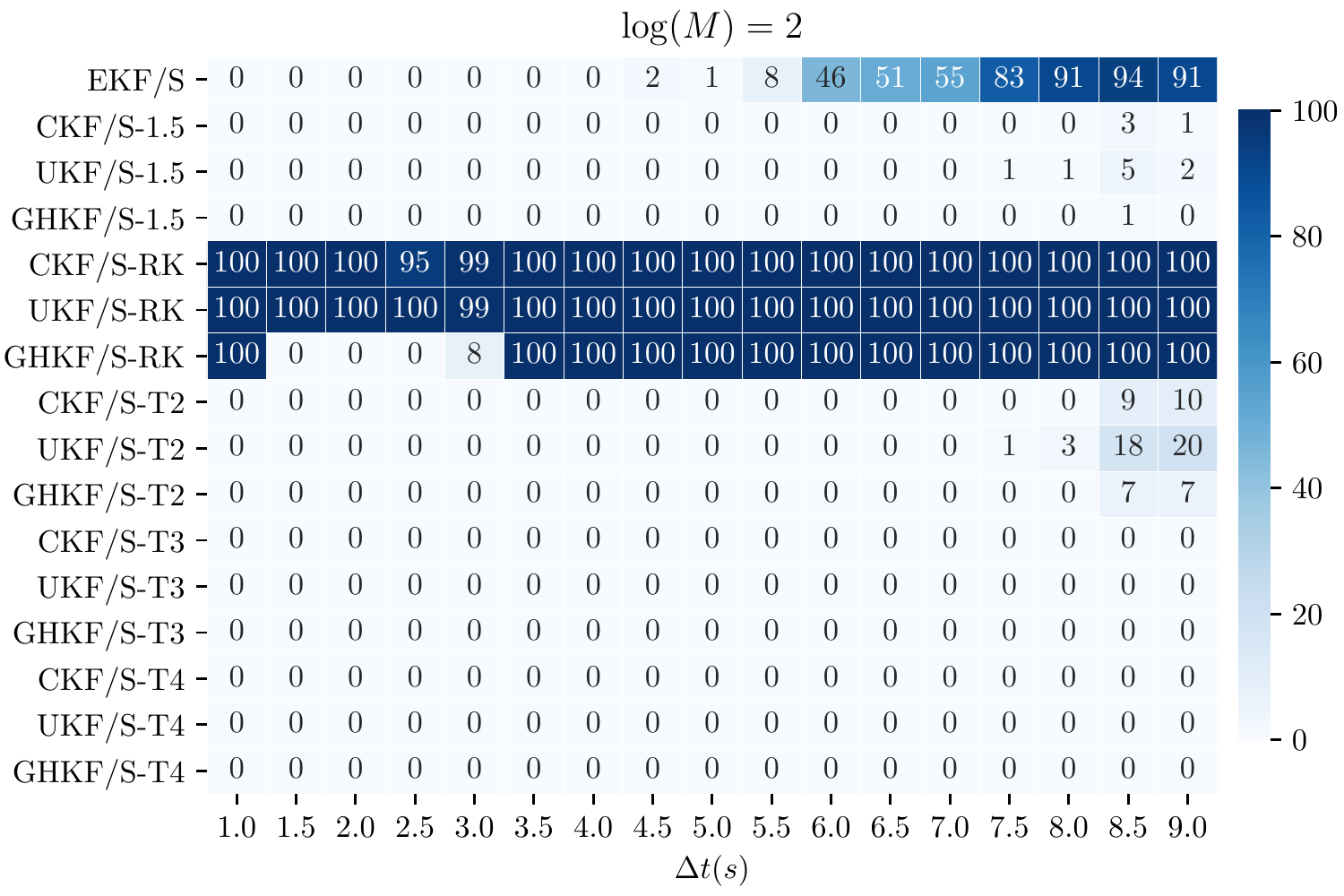}
	\includegraphics[width=.328\linewidth]{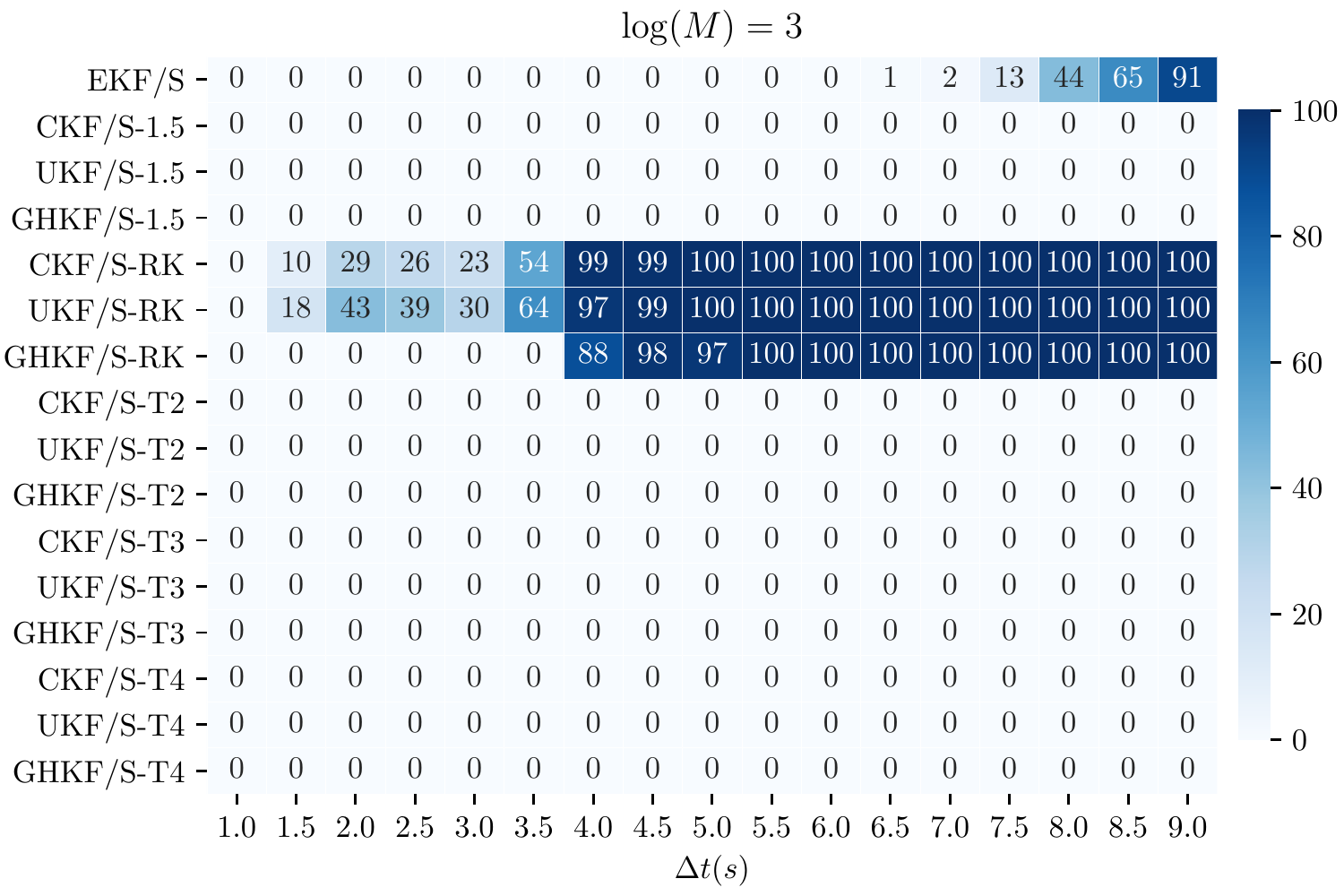}
	\caption{Number of divergence of filters and smoothers. The figures for $\log_2(M)\geq 4$ are of not interests, because It\^{o}-1.5 and TME have no divergences and ODE type of smoothers improve slightly. }
	\label{fig:num-div-all}
\end{figure*}

In Figure~\ref{fig:rmse-filter-all}, the filtering RMSE of the position is shown. It is expected that by increasing the number of integration steps $M$, the accuracy of the filters are all improved significantly, and the differences among filters become smaller. When $\log_2(M)=5$, which is high enough, their results are almost identical, except for CD-EKF. We also observe that the RMSE increases when time interval $\Delta t$ becomes larger, but the growth rate is different for different filters and integration steps. For example, on the last row of Figure~\ref{fig:rmse-filter-all} and $\log_2(M)=2$, the growth rate of RMSE of TME filters is almost linear with respect to $\Delta t$, while for It\^{o}-1.5, it appears to be of higher order. When $\Delta t$ is small enough, the results of all filters are almost the same. However, the Gauss-ODE type of filters (CD-CKD-RK, CD-UKF-RK, and CD-GHKF-RK), only produce reasonable results with large enough number of integration steps. The Linear-ODE filter (CD-EKF-RK) seems to perform well when $\Delta t$ is small, but does not compare favourably with others when $\Delta t$ is large. The performance of the Linear-ODE type filter is not consistent with respect to $\Delta t$ either.

The TME Gaussian filter is competitive with It\^{o}-1.5 and ODE type filters, even with just a $2$nd order of expansion. As shown in the first row of Figure~\ref{fig:rmse-filter-all}, the accuracy of It\^{o}-1.5 and ODE type of filters are only slightly better than TME-2 filters. The TME-3 and TME-4 filters substantially outperform other types of filters, especially when the time interval $\Delta t$ is large. 

Compared to others, the TME filters also require less number of integration steps to achieve a same level of accuracy. As illustrated in Figure~\ref{fig:rmse-filter-all}, the RMSE growth rate of TME-2, TME-3, and TME-4 become nearly linear with respect to $\Delta t$ starting from $\log_2(M)=4, 3, $ and $2$, respectively. For It\^{o}-1.5 filters, only when $\log_2(M)\geq4$, the RMSE growth rate starts to be linear enough. This reveals that when fixing the integration steps $M$, the TME filters have better tolerance for larger $\Delta t$, and less chance for RMSE to grow rapidly. It is also beneficial in real applications, because performing additional integration steps usually comes with a trade-off between the time efficiency and accuracy. 

From Figure~\ref{fig:rmse-filter-all}, we also discover that the 3rd order Gauss--Hermite numerical integration method is generally better than unscented transform and spherical cubature methods. However, in this coordinate turn model, using Gauss--Hermite requires $7^3=2187$ number of sigma-point evaluations, while the cubature and unscented transform methods only need $14$ and $15$ evaluations, respectively.

In Figure~\ref{fig:rmse-smooth-all}, the RMSE results of the smoothers are shown. The accuracies are all improved considerably comparing to their corresponding filtering results. In general, the patterns of the smoothing results are almost the same to the filtering results shown in Figure~\ref{fig:rmse-filter-all}, where the TME Gaussian smoothers also outperform others in terms of accuracy with respect to the time interval $\Delta t$ and integration steps $M$. The Linear-ODE type of smoother performs the worst. The Gauss-ODE type of smoothers are no better than the numerical time-discretisation based smoothers (It\^{o}-1.5 and TME), and only start to produce reasonable results from $\log_2(M)\geq 3$. 

It is also worth noticing from Figure~\ref{fig:rmse-smooth-all} that the Gauss-ODE type of smoothers have almost linear RMSE growth with respect to $\Delta t$. This property seems to be beneficial when dealing with large $\Delta t$. However, in fact, the Gauss-ODE type of smoothers are not even numerically stable with large $\Delta t$ (see, Figure~\ref{fig:num-div-all}). Moreover, for almost every choice of $\Delta t$ and $M$ in Figure~\ref{fig:rmse-smooth-all}, the ODE type of smoothers are actually worse than the TME smoothers.

In Figure~\ref{fig:rmse-TME-order}, we examine the RMSE of TME filters and smoothers with different Taylor expansion orders. Generally, the errors are improved with higher order of expansion. However, the improvement is only found to be significant when the time interval $\Delta t$ is large or the number of integration steps $\log_2(M)$ is small. When $\log_2(M)\geq4$, the TME filters and smoothers are fairly similar with different expansion orders.  

In the first column of Figure~\ref{fig:rmse-TME-order}, it is surprising that the TME-3 filter performs even better with $\Delta t$ growing larger from $\Delta t=3$~s when $\log_2(M)=0$. This phenomenon is also found in the second row of Figure~\ref{fig:rmse-filter-all} when $\log_2(M)=1$. In addition, the TME-3 has particularly better numerical stability, as shown in Figure~\ref{fig:num-div-all}. Considering the TME estimator is highly model-sensitive, the TME-3 appears to be especially affinitive to this coordinate turn model.

Next, pertaining to the numerical stability, we record the number of divergences for the filters and smoothers in Figure~\ref{fig:num-div-all}. This is recorded once the non-positive definite covariance estimate manifests or the estimate is unbounded (NaN errors) in both of the filtering and smoothing procedures. First, we find that the stability of ODE type of filters and smoothers are the worst. The Gauss-ODE type of filters and smoothers even diverge for almost every $\Delta t$ when $\log_2(M)\leq 2$. With $\log_2(M)\geq3$, the stability is improved, but they still diverge from every $\Delta t\geq 5$. The stability of Linear-ODE type of filter and smoothers is also poor and inconsistent. 

The time-discretisation (It\^{o}-1.5 and TME) methods are generally more numerically stable than the ODE type of methods. When $\log_2(M)\geq 2$, the It\^{o}-1.5 and TME filters and smoothers converge for almost every choice of $\Delta t$. The It\^{o}-1.5 filters and smoothers are only comparable to the TME-2 filters and smoothers. By increasing the expansion order to three and four, the number of divergences of TME filters and smoothers decreases, and are less than of It\^{o}-1.5. The TME-3 is more stable than TME-4 in this model when $\log_2(M)=1$. For the numerical integration methods, the filters and smoothers using Gauss--Hermite are slightly more stable than those with spherical cubature or unscented transform. 

\begin{figure}[h!]
	\centering
	\includegraphics[width=.95\linewidth]{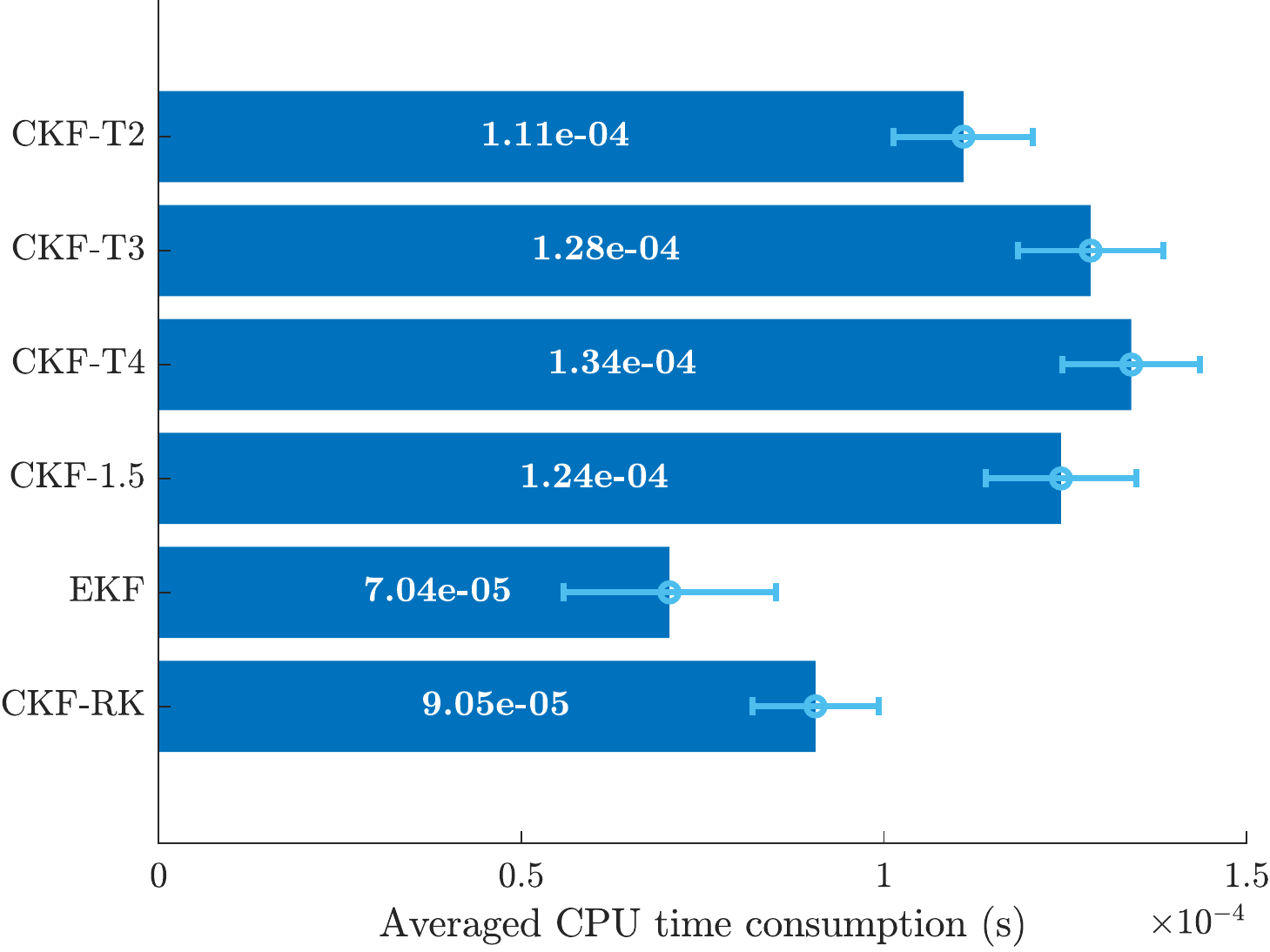}
	\caption{Averaged CPU time through $10^6$ independent running on an Xeon\textsuperscript{\textregistered} E3-1230v5 workstation with MATLAB\textsuperscript{\textregistered} R2019a implementation. The error bars indicate the standard deviation by Monte Carlo runs. }
	\label{fig:cpu-time}
\end{figure}

To compare the actual computational efficiencies of filters, we perform $10^6$ independent runs of the prediction steps and calculate the mean CPU time consumption. We focus on comparing only the prediction step because that is the only difference among the involved Gaussian filters. For simplicity, the smoothers are not compared. We uniformly use the same cubature integration method for the Gaussian integral evaluations. The result is illustrated in Figure~\ref{fig:cpu-time}. We observe that the ODE type of filters (EKF and CKF-RK) are the most efficient. For the time-discretisation based methods (TME and It\^{o}-1.5), the CKF-T2 and CKF-T4 are the least and most time consumable, respectively. However, we also find that the efficiencies of the time-discretisation based methods are very similar, which means that the impacts brought by increasing the expansion order can be negligible to certain levels. For example, the CKF-T4 is only slightly higher than It\^{o}-1.5 with $0.1\times 10^{-4}$~s.

Overall, in terms of estimation accuracy, the second order TME filter and smoother is comparable to the It\^{o}-1.5 and the ODE type of filters and smoothers. With higher order expansions, the third and fourth order TME significantly outperform them on both filtering and smoothing tasks. The combination of TME with Gauss--Hermite integration is moderately better than cubature and unscented transform. Although the covariance estimate of TME is not ensured to be positive definite, the third and fourth order TME filter and smoother actually achieve the least number of divergence, followed by It\^{o}-1.5 and second order TME. The ODE type of methods, especially Gauss-ODE are surprisingly the most unstable among them. Provided with enough integration steps, the time-discretisation based methods TME and It\^{o}-1.5, will converge. 

\section{Conclusion}
\label{sec:conclusion}

In this paper, we have proposed a novel method for continuous-discrete non-linear Gaussian filtering and smoothing, where the system dynamics are characterised by an SDE. The core contribution is in how to form a Gaussian approximation to the transition density of the SDE. Differently from the traditional It\^{o}--Taylor discretisation, Gaussian, and linearised ODE methods, we proposed to exploit a Taylor moment expansion scheme, where the moment functions are time-discretised by a Taylor expansion. The benefit is that the mean and covariance solutions are asymptomatically exact.

We then derived the corresponding TME Gaussian filter and smoother. In addition, we analysed the positive definiteness of the TME covariance estimates and stability of the TME Gaussian filter and smoother. The numerical experiments indicate that even a second order TME Gaussian filter and smoother is in line with the state-of-the-art. With higher expansion order, the proposed TME filter and smoother substantially excel the state-of-the-art methods, in terms of both estimation accuracy and numerical stability. 

\appendices
\section{Proof of Lemma \ref{lemma:general-covariance}}
\label{sec:append-2}
Let $\beta^{u,v}_r\triangleq\beta_r(x_u\,x_v) = \mathcal{A}^r(x_u\,x_v)$ denote the $r$-th iteration of operator $\mathcal{A}$, and $\partial_i \beta^{u,v}_r = \partial \beta^{u,v}_r / \partial x_i$. For $r = 0$ and  $1$, we find
\begin{equation}
\begin{split}
\alpha^u_0 &= x_u, \\
\alpha^u_1 &= f_u,\\
&\vdots \\
\alpha^v_0 &= x_v, \\
\alpha^v_1 &= f_v,\\
&\vdots \\
\beta^{u,v}_0 &= x_ux_v = \alpha^u_0\alpha^v_0,\\
\beta^{u,v}_1 &= \alpha^u_0\alpha^v_1 + \alpha^v_0\alpha^u_1 + \Gamma_{uv}.\\
&\vdots
\end{split}
\label{equ:iter-pattern-multi}
\end{equation}
We can calculate $\Phi^{u,v}_{x, r}$ by \eqref{equ:psd-first} and initially reveal $\Phi^{u,v}_{x, 0} = 0$, $\Phi^{u,v}_{x, 1}=\Gamma_{uv}$. From this pattern above, we will first show $\beta^{u, v}_r$ has a general expression
\begin{equation}
\beta_r = \sum^r_{s=0}\binom{r}{s}\alpha^u_s\alpha^v_{r-s} + \Phi^{u,v}_{x,r},
\label{equ:assump-x2-multi}
\end{equation}
where $\Phi^{u,v}_{x,r}$ is
\begin{equation}
\begin{split}
\Phi^{u,v}_{x,r}&=\sum_{i, j}\sum^{r-1}_{s=0}\binom{r-1}{s}\left( \partial_i\alpha^u_s\,\partial_j\alpha^v_{r-s} \right)\Gamma_{ij} + \mathcal{A}\Phi^{u,v}_{x,r-1},
\end{split}
\label{equ:assump-x2-psi-multi}
\end{equation}
It is apparent that \eqref{equ:assump-x2-multi} and \eqref{equ:assump-x2-psi-multi} hold for $r=1$. By Algorithm~\ref{def:taylor-moment-estimator}, the iteration of $\beta^{u,v}_r$ is
\begin{equation}
\begin{split}
\beta^{u,v}_{r+1} = \mathcal{A}(\beta^{u,v}_{r})=\sum_i\tash{\beta^{u,v}_r}{x_i}f_i + \frac{1}{2}\sum_{i, j}\frac{\partial^2\beta^{u,v}_r}{\partial x_i\partial x_j}\Gamma_{ij}.
\label{equ:A-B-evolution-multi}
\end{split}
\end{equation}
With Equation~\eqref{equ:assump-x2-multi}, we continue Equation~\eqref{equ:A-B-evolution-multi} and deduce that
\begin{equation}
\begin{split}
&\sum_i\sum^r_{s=0}\binom{r}{s}\left(\partial_i\alpha^u_s\alpha^v_{r-s} + \alpha^u_s\partial_i\alpha^v_{r-s} + \partial_i\Phi^{u,v}_{x,r}\right)f_i \\
&\quad+\frac{1}{2}\sum_{i, j}\Biggl(\sum^r_{s=0}\binom{r}{s}\bigl(\tashh{\alpha_s^u}{x_i}{x_j}\alpha^v_{r-s} + \partial_i\alpha_s^u\partial_j\alpha^v_{r-s} \\
&\quad\qquad\qquad\qquad\quad\quad+\partial_j\alpha_s^u\partial_i\alpha^v_{r-s} + \alpha^u_s\tashh{\alpha^v_{r-s}}{x_i}{x_j} \\
&\quad\qquad\qquad\qquad\quad\quad+ \tashh{\Phi^{u,v}_{x,r}}{x_i}{x_j}\bigr)\Biggr)\Gamma_{ij} \\
&= \sum^r_{s=0}\binom{r}{s}\left(\alpha^u_{s+1}\alpha^v_{r-s} + \alpha^u_s\alpha^v_{r-s+1}\right)\\
&\quad+ \sum_{i, j}\left(\sum^r_{s=0}\binom{r}{s}\partial_i\alpha_s^u\partial_j\alpha^v_{r-s} \right)\Gamma_{ij} +\mathcal{A}\Phi^{u,v}_{x,r} \\
&= \sum^{r+1}_{s=0}\binom{r}{s-1}\left(\alpha^u_s\alpha^v_{r-s+1}\right) + \sum^{r+1}_{s=0}\binom{r}{s}\left(\alpha^u_s\alpha^v_{r-s+1}\right)\\
&\quad+ \Phi^{u,v}_{x,r+1} \\
&= \sum^{r+1}_{s=0}\binom{r+1}{s}\alpha_s\alpha_{r-s+1} + \Phi^{u,v}_{x,r+1} \\
&= \beta_{r+1}
\end{split}
\end{equation}
Thus expressions \eqref{equ:assump-x2-multi} and \eqref{equ:assump-x2-psi-multi} are proved by induction for $r\geq 1$. From Lemma \ref{theorem:psd}, we now derive an iterated form of $\Phi_{x, r}$ as
\begin{equation}
\begin{split}
\Phi^{u,v}_{x, r} &= \beta^{u,v}_r - \sum^r_{s=0}\binom{r}{s}\alpha^u_s\alpha^v_{r-s}\\
&=\sum_{i, j}\sum^{r-1}_{s=0}\binom{r-1}{s}\left( \partial_i\alpha^u_s\,\partial_j\alpha^v_{r-s-1}\right)\Gamma_{ij} +\mathcal{A}\Phi^{u,v}_{x,r-1}.
\end{split}
\end{equation}
In matrix form, it is then
\begin{equation}
\begin{split}
\Phi_{x,r}^{u, v} = \sum^{r-1}_{s=0}\binom{r-1}{s}\trace\left((\nabla\alpha^u_s)^\trans\,\nabla\alpha^v_{r-s-1}\,\bm{\Gamma}\right) +\mathcal{A}\Phi_{x,r-1}^{u, v}.
\label{equ:phi-iter-form}
\end{split} 
\end{equation}
We can also alternatively rearrange the iterated \eqref{equ:phi-iter-form} with
\begin{equation}
\Phi^{u,v}_{x, r} = \sum^{r-1}_{s=0}\mathcal{A}^{s}\sum^{r-s-1}_{l=0}\binom{r-s-1}{l}\trace\left((\nabla\alpha^u_s)^\trans\,\nabla\alpha^v_{r-s-1-l}\,\bm{\Gamma}\right),
\end{equation}
starting from $\Phi_{x, 0}=0$. 

\bibliographystyle{IEEEtran}
\bibliography{refs}

\begin{IEEEbiography}[{\includegraphics[width=1in,height=1.25in,clip,keepaspectratio]{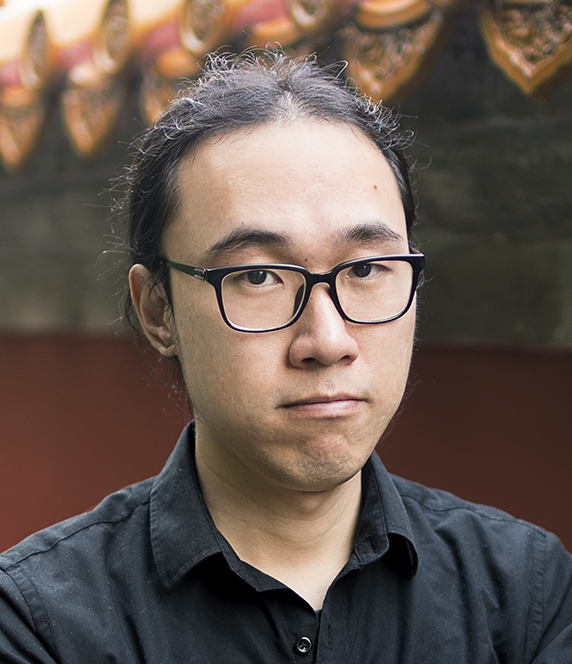}}]{Zheng Zhao}
	received his M.Sc. degree of control science and engineering from Beijing University of Technology in 2017. He is currently a doctoral candidate in Department of Electrical Engineering, Aalto University, Finland. His research interests include non-linear Bayesian filtering and machine learning applications.  
\end{IEEEbiography}
\begin{IEEEbiography}[{\includegraphics[width=1in,height=1.25in,clip,keepaspectratio]{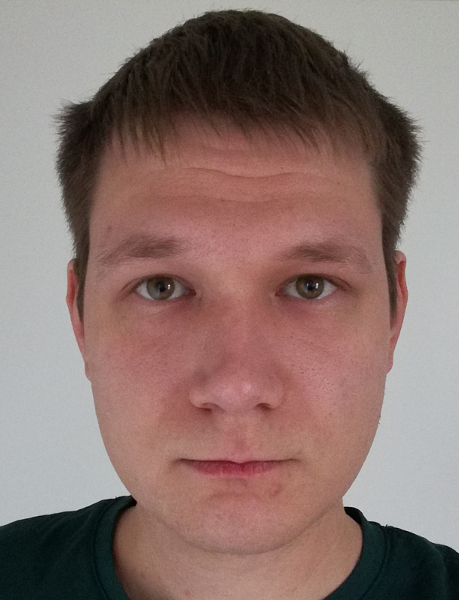}}]{Toni Karvonen} Toni Karvonen received his Master of Science degree in applied mathematics from the University of Helsinki, Finland, in 2015 and his Doctor of Science (Tech.) degree in electrical engineering from Aalto University, Finland, in 2019. He is currently a postdoctoral researcher in the Department of Electrical Engineering and Automation at Aalto University. His research interests are in stochastic state estimation and numerical analysis, in particular in numerical integration and probabilistic numerics.
\end{IEEEbiography}
\begin{IEEEbiography}[{\includegraphics[width=1in,height=1.25in,clip,keepaspectratio]{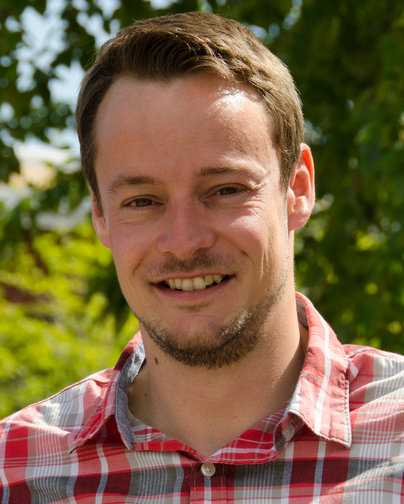}}]{Roland Hostettler}
	(S'10-M'14) received the Dipl. Ing. degree in Electrical and Communication Engineering from Bern University of Applied Sciences, Switzerland in 2007, and the M.Sc. degree in Electrical Engineering and Ph.D. degree in Automatic Control from Lule\aa\, University of Technology, Sweden in 2009 and 2014, respectively. He has held Post-Doctoral Researcher positions at Lule\aa\, University of Technology, Sweden and Post-Doctoral Researcher and Research Fellow positions at Aalto University, Finland. Currently, he is an Associate Senior Lecturer with the Department of Electrical Engineering, Uppsala University, Sweden. His research interests include statistical signal processing with applications to target tracking, biomedical engineering, and sensor networks.
\end{IEEEbiography}
\begin{IEEEbiography}[{\includegraphics[width=1in,height=1.25in,clip,keepaspectratio]{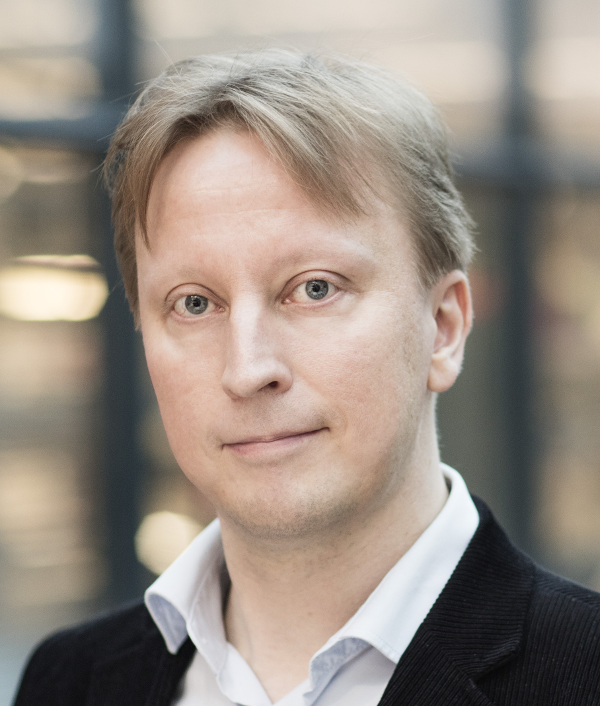}}]{Simo S\"{a}rkk\"{a}}
	is an Associate Professor with Aalto University. His research interests are in multi-sensor data processing systems with applications in location sensing, health technology, machine learning, inverse problems, and brain imaging. He has authored or coauthored over 100 peer-reviewed scientific articles and 3 books. He is a Senior Member of IEEE, serving as an Associate Editor of IEEE Signal Processing Letters, and is a member of IEEE Machine Learning for Signal Processing Technical Committee.
\end{IEEEbiography}

\end{document}